\documentclass{article}
\usepackage[utf8]{inputenc}

\usepackage{authblk}
\usepackage{graphicx}
\usepackage{amsmath}
\usepackage{amssymb}
\usepackage{amsthm}
\usepackage{color}
\usepackage{url}
\usepackage{hyperref}
\usepackage{enumitem}
\usepackage{stmaryrd}
\usepackage{stmaryrd}
\usepackage{tikz}
\usepackage{setspace}
\usepackage{bm}
\usepackage{booktabs}

\usepackage{multirow}

\usepackage{caption}
\usepackage{subcaption}

\usepackage{algorithm}
\usepackage{algpseudocode}
\usepackage{soul}

\usepackage[margin=0.85in]{geometry}

\newcommand{\be}{\begin{equation}}
\newcommand{\ee}{\end{equation}}
\newcommand{\ba}{\begin{array}}
\newcommand{\ea}{\end{array}}
\newcommand{\bea}{\begin{eqnarray}}
\newcommand{\eea}{\end{eqnarray}}

\newcommand{\calF}{{\cal F }}

\newcommand{\FF}{\mathbb{F}}

\renewcommand{\ker}[1]{\mathsf{ker}{(#1)}}

\newtheorem{prop}{Proposition}

\newcommand{\app}[1]{\hyperref[app:#1]{Appendix~\ref*{app:#1}}}

\newcommand{\eq}[1]{Eq.~(\ref{eq:#1})}
\renewcommand{\sec}[1]{\hyperref[sec:#1]{Section~\ref*{sec:#1}}}
\newcommand{\ssec}[1]{\hyperref[ssec:#1]{Subsection~\ref*{ssec:#1}}}
\newcommand{\fig}[1]{\hyperref[fig:#1]{Figure~\ref*{fig:#1}}}
\newcommand{\tab}[1]{\hyperref[tab:#1]{Table~\ref*{tab:#1}}}
\newcommand{\lem}[1]{\hyperref[lem:#1]{Lemma~\ref*{lem:#1}}}
\newcommand{\propos}[1]{\hyperref[propos:#1]{Proposition~\ref*{propos:#1}}}
\newcommand{\thm}[1]{\hyperref[thm:#1]{Theorem~\ref*{thm:#1}}}
\newcommand{\alg}[1]{\hyperref[alg:#1]{Algorithm~\ref*{alg:#1}}}

\title{Fail fast: techniques to probe rare events in quantum error correction}

\author[1]{Michael E. Beverland}
\author[1]{Malcolm Carroll}
\author[1]{Andrew W. Cross}
\author[1]{Theodore J. Yoder}
\affil[1]{IBM Quantum}

\begin{document}
\maketitle

\begin{abstract}
The ultimate goal of quantum error correction is to create logical qubits with very low error rates (e.g.~$10^{-12}$) and assemble them into large-scale quantum computers capable of performing many (e.g.~billions) of logical gates on many (e.g.~thousands) of logical qubits. 
However, it is necessarily difficult to directly assess the performance of such high-quality logical qubits using standard Monte Carlo sampling because logical failure events become very rare. 
Building on existing approaches to this problem, we develop three complementary techniques to characterize the rare-event regime for general quantum low-density parity-check (qLDPC) codes under circuit noise. 
(I) We propose a well-motivated, low-parameter ansatz for the failure spectrum (the fraction of fault sets of each size that fail) that empirically fits all the QEC systems we studied and predicts logical error rates at all physical error rates.
(II) We find min-weight logical operators of syndrome measurement circuits and exactly compute the number of min-weight failing configurations. 
(III) We generalize the splitting method to qLDPC codes using multi-seeded Metropolis sampling to improve convergence for systems with many inequivalent logical operators. 
We apply these tools to distance-6, -12, and -18 bivariate bicycle codes under circuit noise, observing strong low-error-rate performance with the recently proposed Relay decoder but also considerable scope for further improvement.
\end{abstract}

\newpage
\tableofcontents
\setlength{\parskip}{\medskipamount} 

\clearpage
\section{Introduction and overview}
\label{sec:intro}

Reliable estimates of the failure probability of components of a quantum computer is crucial for assessing different quantum computing approaches and architectures~\cite{BicycleArchitecturePaper}. 
However, quantum error-correcting (QEC) codes are typically implemented using large circuits with an exponential number of fault configurations, making accurate analytic characterization infeasible. 
As a result, performance is generally assessed through sampling-based simulations. 
However, direct sampling, often referred to as `standard Monte Carlo', becomes problematic in the regime of practical interest where logical failure rates below $10^{-12}$ are anticipated for large scale fault tolerant quantum computing~\cite{beverland2022assessing,BicycleArchitecturePaper}. 
In such regimes, estimating failure rates with direct sampling would require an infeasibly large number of simulation runs, making alternative approaches essential.
In this work, we study three complementary techniques for characterizing a system across all error rates, including the low error rate regime which is inaccessible to direct sampling, as illustrated in \fig{low-logical-error-methods}.

\begin{figure}[ht!]
    \centering
    \includegraphics[width=1.0\linewidth]{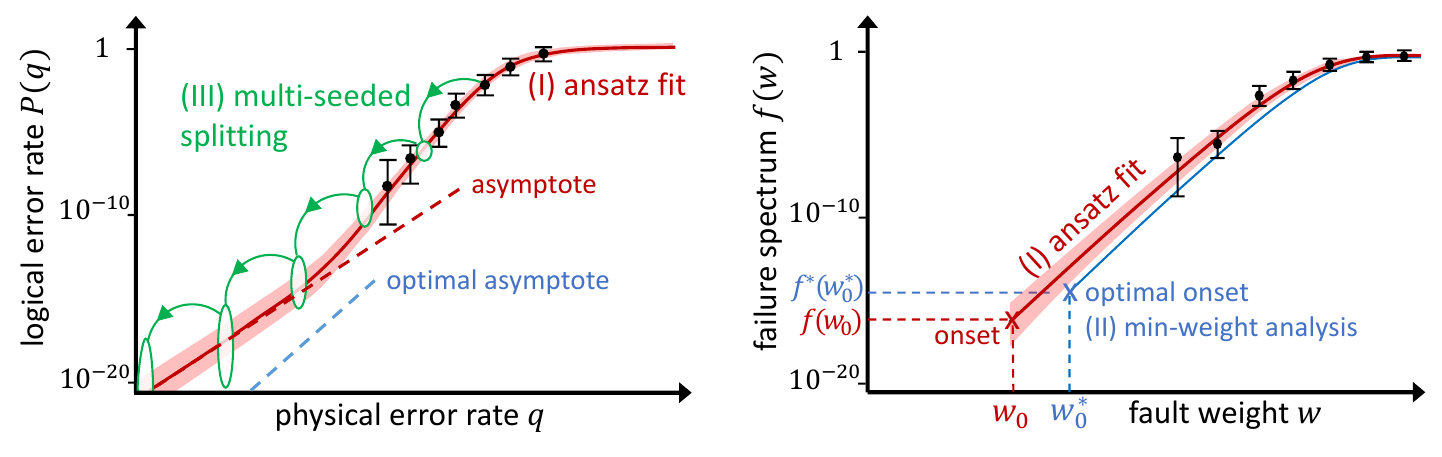}
    \caption{
    Illustration of our techniques for characterizing the logical error rate $P(q)$ as a function of physical error rate $q$ for a given QEC system (a decoding problem and decoder).
    Sampling (black dots) can estimate $P(q)$ directly (left) or indirectly via the failure spectrum $f(w)$ (right), but both become infeasible when failures are too rare.
    Technique~I fits the data to a failure spectrum ansatz (red line), which predicts $P(q)$ across all $q$, with the low $q$ limit (red dashed) determined by the onset point $(w_0, f(w_0))$.
    Technique~II identifies the optimal onset $(w^*_0, f^*(w^*_0))$ achievable by an ideal min-weight decoder. 
    This constrains the failure spectrum since $w_0 \leq w^*_0$ and $f(w^*_0) \geq f^*(w^*_0)$, and identifies the potential gain from an improved decoder.
    Technique~III applies the splitting method to estimate ratios of logical error rates across a sequence of $q$ values using Metropolis sampling.
    Each technique has limitations, but agreement among them strengthens confidence in the predicted low-error regime.
    }
    \label{fig:low-logical-error-methods}
\end{figure}

To describe our alternatives to direct sampling, it is helpful to introduce some notation (see \sec{definitions} for details). 
We assume a fixed quantum circuit with a stochastic noise model where each of $N$ possible faults occurs independently with probability $q$. 
A fault configuration is represented by a bitstring $e \in \mathbb{F}_2^N$, where each bit indicates whether a fault occurs. 
Each $e$ occurs with probability $q^{|e|}(1 - q)^{N - |e|}$ (where $|e|$ is the Hamming weight) and produces a \emph{syndrome} $\sigma(e)$, which is passed to a decoder that attempts a correction that either succeeds or fails. 
We refer to the combination of this decoding problem along with a specific decoding algorithm as the \emph{QEC system}. 
The \emph{onset weight} $w_0$ is the smallest weight of any $e$ that leads the QEC system to fail.
Let $f(w)$ denote the \emph{failure spectrum}, i.e., the fraction of weight-$w$ faults that cause failure, and let the logical error rate $P(q)$ be the expectation of failure given $q$.
Note that $f(w) = 0$ for $w < 0$ and $f(w)$ completely specifies $P(q)$ via 
\begin{equation}
  P(q) = \sum_{w=0}^{N} f(w)\,\binom{N}{w}\, q^{w}\,(1-q)^{N-w}.
  \label{eq:sum-formula}
\end{equation}
We will demonstrate that this formalism applies broadly, from code-capacity analyses to logical operations under circuit-level noise.

Now let us review some already-known approaches that seek to overcome the challenges of direct sampling.
\emph{Importance sampling} restricts the sum in \eq{sum-formula} to a range $w_{\text{min}} \leq w \leq w_{\text{max}}$, while bounding the contribution from terms outside this range.
This works well in some scenarios, for example, in the low $q$ limit, the logical error rate is dominated by the leading term $f(w_0) q^{w_0}$, so estimating the onset $f(w_0)$ alone may suffice. 
However, importance sampling has two main limitations. 
First, $f(w)$ in the relevant range may be extremely small; for instance (as described below) the toric code under bit flip noise has $f(w_0) \approx 10^{-18}$ at $d = 21$, which would be infeasible to estimate by sampling. 
Second, many terms may contribute significantly, requiring estimates of $f(w)$ over a wide range of $w$, which can be prohibitively expensive to obtain by sampling. 
These limitations frequently arise for large QEC systems.

For some QEC families with well-understood structure, the onset term $f(w_0)$ can be computed directly rather than estimated by sampling. 
Notable examples include surface and toric codes under bit-flip noise with minimum-weight decoders, where the structure of logical operators is well characterized~\cite{dennis2002topological}. 
For instance, in the unrotated toric code of even distance $d$ (which has $n=2 d^2$ qubits), min-weight X-type logical operators correspond to topologically nontrivial paths wrapping around the torus. 
There are exactly $2d$ such operators: $d$ wrapping horizontally and $d$ vertically.
The weight-$w_0 = d/2$ failures correspond to exactly half of the weight-$w_0$ restrictions of these operators, giving $f(w_0) = d \binom{d}{d/2}/\binom{n}{d/2}$. 
The rotated toric code presents more complexity due to overlapping logical operators, which makes enumeration of unique restrictions harder, but tight bounds remain achievable~\cite{beverland2019role}. 
These techniques do not generalize easily to other QEC systems, as they rely on specific structural features of toric codes.

The splitting method, a general technique for analyzing rare events~\cite{bennett1976}, was applied to surface codes in Ref.\cite{bravyi2013simulation}. 
It uses a Metropolis algorithm to sample failing configurations at a sequence of error rates and estimates the relative change in logical failure between successive points. 
The method begins at a reference error rate where the logical failure probability is known, typically at large fault probability $q$ where direct Monte Carlo is feasible. 
While the splitting method approach can be readily generalized to arbitrary QEC systems (see \sec{splitting-method} and Ref.~\cite{mayer2025}), it faces two key limitations as pointed out in Ref.~\cite{bravyi2013simulation}. 
First, the Metropolis walk may not be ergodic for a given QEC system; that is, the space of failing configurations may be disconnected under single fault transitions. 
Second, the algorithm’s mixing time is unknown for a given QEC and can grow exponentially with code distance, making the method impractical for some large codes.

We address the limitations for these established approaches, resulting in the following extensions (see \fig{low-logical-error-methods}):

\begin{itemize}
    
    \item[(I)] \textbf{Fitting a failure-spectrum ansatz.}
    We empirically observe that the failure spectrum $f(w)$ varies smoothly with $w$ and exhibits similar behavior across all QEC systems studied. 
    This observation motivates a simple model of a QEC system with a closed form expression for the failure spectrum, from which we derive a low parameter ansatz that fits our data across all regimes and QEC systems.
    The ansatz $f_{\text{ansatz}}(w)$ can be calibrated using data from accessible regimes and provides predictions across the full parameter range for a given QEC system.

    \item[(II)] \textbf{Computing or bounding min-weight properties.}
    We develop methods to compute or lower bound the min-weight decoder’s failure-spectrum onset $f^*(w_0^*)$. 
    For any other decoder, $w_0 \le w_0^*$ and $f(w_0^*) \ge f^*(w_0^*)$.
    Our approach combines new techniques for identifying minimum weight logical operators in general LDPC codes with carefully designed sampling procedures.
    We use this to study a set of notable bivariate bicycle codes \cite{bravyi2024high} under bitflip and circuit noise models.
    
    \item[(III)] \textbf{Multi-seeded splitting for general QEC systems.} 
    We investigate the applicability of the splitting method introduced for surface codes~\cite{bravyi2013simulation} to more general quantum LDPC codes, specifically the unrotated toric and bivariate bicycle codes. 
    To partially address the problem of exploring the entire space of failing configurations in LDPC codes, we seed Markov chains at a high error rate (accessible by Monte Carlo) with multiple typical failing configurations. 
    Then, we use the final failing configurations from chains at high error rates to seed chains at lower error rates where Monte Carlo is infeasible.
    
\end{itemize}

To apply our techniques to error models where different classes of faults can have different probabilities of occurring (such as circuit noise), we include additional copies of each fault in proportion to its probability, yielding an equal-probability representation.

There are a number of directions to improve upon these three methods.
It would be valuable to extend each method to adaptive fault-tolerant protocols~\cite{heussen2024dynamical}, including magic state cultivation~\cite{chamberland2018fault,gidney2024magic}, and to more sophisticated noise models, including for example, Pauli\textsc{+} noise \cite{Google2025quantum} and coherent noise \cite{miller2025efficient}. 
Regarding the splitting method, we are still limited in our understanding of mixing times of the Markov chains for general QEC systems. 
As remarked on in some more detail in Section~\ref{sec:splitting-method}, future work could focus both on bringing in insights from the extensive literature on Metropolis sampling, e.g.~Refs.~\cite{gelman1992inference,neal1996sampling}, and also on exploiting the structure of QEC systems, like that studied here, to increase the chain transition frequency.

The rest of this paper is structured as follows. 
In \sec{definitions}, we provide background and definitions used throughout. 
In \sec{ansatz-approach}, \sec{computing-min-weight-properties}, and \sec{splitting-method}, we present our three techniques, each of which can be read independently. 
Finally, in \sec{validate-relay}, we apply our techniques to evaluate the circuit-noise performance of the recently proposed Relay decoder~\cite{RelayPaper} on bivariate bicycle codes.
This work also supports the resource estimates in Ref.~\cite{BicycleArchitecturePaper} for quantum architectures based on bivariate bicycle codes.

\clearpage
\section{Background and definitions}
\label{sec:definitions}

Here we provide the main definitions and notation we will use throughout the paper.
Throughout this work, binary matrices and vectors use arithmetic over $\mathbb{F}_2$.

\subsection{Quantum error correction systems}

In this work, we represent QEC with a check matrix $H \in \mathbb{F}_2^{M \times N}$, an action matrix $A \in \mathbb{F}_2^{K \times N}$, and a decoding algorithm $\mathcal{C}\colon \mathbb{F}_2^{M} \to \mathbb{F}_2^{N}$. 
Each error $j \in \{1, \dots, N\}$ occurs independently with probability $q$, flipping checks and logical generators specified by the $j$th columns of $H$ and $A$, respectively.
When a set of errors represented by bitstring $e\in \mathbb{F}_2^{N}$ occurs, it produces an observed syndrome $\sigma = H e$.
The decoding algorithm $\mathcal{C}$ provides a correction $c = \mathcal{C}(\sigma)$ such that $H c = \sigma$, succeeding iff $A c = A e$.
We use the same mathematical representation for fault-tolerant circuits, but with a different interpretation of objects (errors $\to$ faults, checks $\to$ detectors etc.). See section 2 in Ref.~\cite{DTDpaper} for an overview of how these objects arise in different settings. 

We define a \emph{QEC system} as the combination of decoding objects $H$, $A$ and $\mathcal{C}$.
We consider a range of QEC systems in this work, which we specify and list in \sec{system-examples}.

A \emph{logical bitstring} (`logical' for short) is a set of faults $e$ such that $He=0$ but $Ae\neq0$. 
The \emph{distance} $D$ of the QEC system is the minimum Hamming weight $|e|$ of a logical bitstring, which is independent of the decoder $\mathcal{C}$.  
Note that $D$ can differ from the \emph{code distance} $d$, as it depends on the noise model and circuit implementation. 

The \emph{onset weight} $w_0$ of the QEC system is the minimum weight of an error which causes the decoder to fail. 
We say that a decoder is \emph{min-weight} if the decoder always produces a minimum-weight correction consistent with the observed syndrome. 
For a min-weight decoder, the onset weight is $w^*_0 := \lceil D/2 \rceil$, and for any decoder $w_0\le w^*_0$.

\paragraph{Non-uniform fault probabilities:}
In the description above, we assume each column in $H \in \mathbb{F}_2^{M \times N}$ and $A \in \mathbb{F}_2^{K \times N}$ is selected with uniform probability $q$. 
We call $(H,A,q)$ the QEC system's \emph{expanded representation}.
Non-uniform fault probabilities are captured by including multiple copies of columns.
For example, in circuit noise (see \sec{system-examples}) with error rate $p$, we set $q=p/15$: each non-trivial two-qubit Pauli after a CNOT has probability $p/15=q$ (one column), whereas a measurement flip has probability $p=15q$ (15 identical columns).
 
More generally, we can form an expanded representation for any QEC system with check/action matrices $\tilde H\in\mathbb{F}_2^{M\times\tilde N}$ and $\tilde A\in\mathbb{F}_2^{K\times\tilde N}$ in which fault $j\in[\tilde N]$ occurs with probability $m_j\,p/b$, with $p \in (0,1/2)$ a global noise parameter and $b \in \mathbb{R}$ chosen so that each $m_j$ is a positive integer. 
We call $(\tilde H,\tilde A, \vec{m}/b, p)$ the \emph{compressed representation}. 
The equivalent expanded representation $(H,A,q)$ is obtained by replicating column $j$ exactly $m_j$ times (so $N=\sum_j m_j$) and setting $q=p/b$. 
Sampling columns of $(H,A)$ independently with uniform probability $q=p/b$ is equivalent to sampling $(\tilde H,\tilde A)$ with probabilities $m_j\,p/b$, up to $O(p^2)$ corrections that we neglect. 

In this work, we use the expanded representation for definitions, but it is often convenient to work with the (typically smaller) compressed representation and then map across. 
For some systems (e.g., bitflip noise), $m_j\equiv1$ and the two representations coincide.

\subsection{Failure spectrum and bitstring sets}
\label{sec:fail-frac}

The logical failure rate $P(q)$ is defined as the total probability of all fault bitstrings that cause the decoder to fail.
It is helpful for us to define the following general transform $\mathcal{T}\{g\}(q)$ of any real-valued function $g(w)$ over the integers $w = 0,1, \dots N$ to a function over the range $0 \leq q \leq 1$ as:
\begin{eqnarray}
\mathcal{T}\{g\}(q) = \sum_{w = 0}^{N} g(w) \binom{N}{w} q^w (1-q)^{N-w}.
\label{eq:transform}
\end{eqnarray}
With this notation, note that \eq{sum-formula} can be written as 
$P(q) = \mathcal{T}\{f\}(q)$, where the function $f(w) \in [0,1]$ is the \emph{failure spectrum}, namely the fraction of weight-$w$ fault bitstrings that fail.
The failure spectrum $f(w)$ therefore fully characterizes the QEC system's error correction performance for any value of $q$.
While $f(w)$ can, in principle, be determined by testing all $\binom{N }{W}$ weight-$w$ errors, this is computationally infeasible except for very small $N$. 
We refer to the first non-zero value $f(w_0)$ of the failure spectrum as the \emph{onset fraction}, where $w_0$ is the onset weight.

\begin{figure}[ht!]
    \centering  (a)\includegraphics[width=0.45\linewidth]{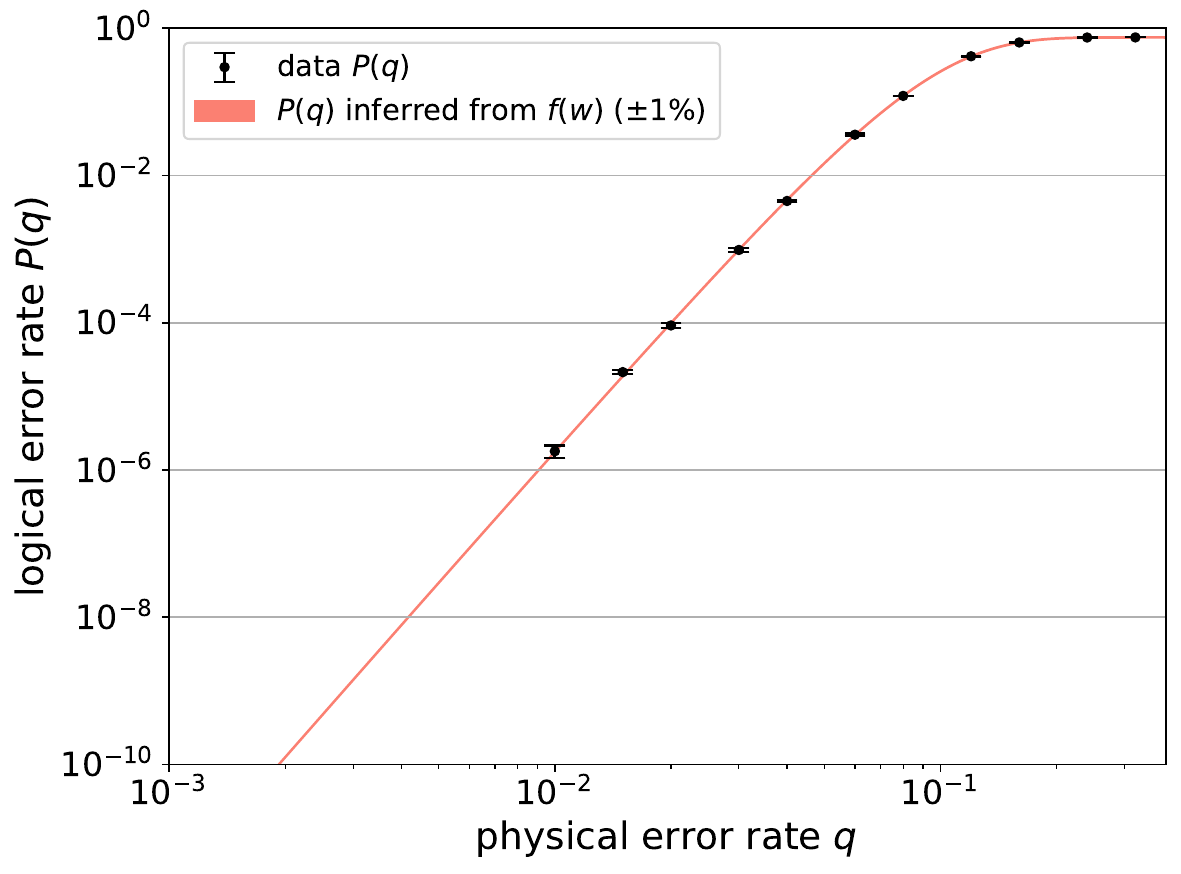}(b)\includegraphics[width=0.45\linewidth]{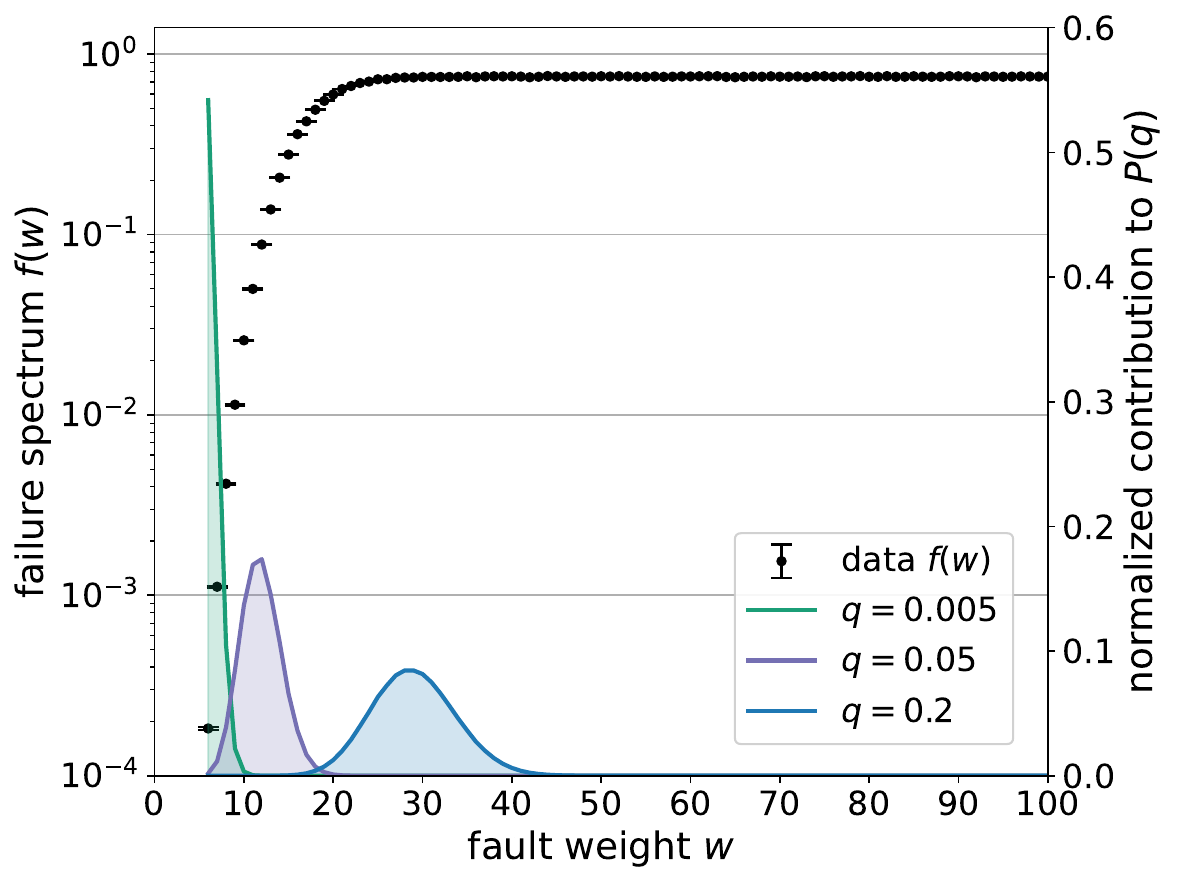}
    \caption{
    Logical error rate and failure spectrum for a distance-12 rotated toric code (see RT(12)-bitflip in \sec{system-examples}).
    (a) Direct sampling estimates $\hat{P}(q)$ of the logical error rate (black), with the overlaid sum $\mathcal{T}\{\hat{f}\}(q)$ (shaded red to include error bars) derived from the failure spectrum estimated in (b).
    (b) Direct sampling estimates $\hat{f}(w)$ of the failure spectrum, showing the normalized contribution of each weight $w$ to $\mathcal{T}\{\hat{f}\}(q)$ for three values of $q$ (colors).
    At $q = 0.005$ (well below pseudo-threshold), the sum is dominated by the minimum-weight term $w_0$, while higher $q$ values show significant contributions from a broader weight range.
    We sample sufficiently to estimate each $f(w)$ to within 1\% relative error for all $w=6,7,\dots,N$, observing that $f(w)$ rises monotonically to $\sim\!0.75$ near $w\approx 30$ and remains near this value. 
    We note that $f(w)$ eventually declines near $w\approx 115$ (beyond the plotted range), but for $q<\tfrac12$ this non-monotonic tail modifies $\mathcal{T}\{\hat{f}\}(q)$ only negligibly.
    \label{fig:failure-spectrum-transform}
    }
\end{figure}

\textbf{Bitstring sets: }
It is helpful to define several sets.
The set of \emph{weight-$w$ bitstrings} is $\mathcal{E}(w) = \{e \in \mathbb{F}_2^N \mid |e| = w \}$.
The set of \emph{weight-$w$ logical bitstrings} is $\mathcal{L}(w) = \{e \in \mathcal{E}(w) \mid He = 0 \text{ and } Ae \neq 0\}$.
The set of \emph{weight-$w$ failing bitstrings} is $\mathcal{F}(w) = \{e \in \mathcal{E}(w) \mid A\mathcal{C}(He) \neq Ae\}$.
These sets relate to some of our earlier definitions.
The failure spectrum $f(w) = |\mathcal{F}(w)|/\binom{N}{w}$.
The system distance $D = \min_{|\mathcal{L}(w)| > 0}(w)$, and min-failing weight $w_0 = \min_{|\mathcal{F}(w)| > 0}(w)$.
We can further define the min-weight failing bitstrings $\mathcal{F}(w_0)$ and min-weight logical bitstrings $\mathcal{L}(D)$.

\subsection{Direct (Monte Carlo) sampling and importance sampling}

Here we briefly review some standard sampling methods that can be used to estimate the logical error rate $P(q)$.

\textbf{Direct sampling: } 
Direct sampling, also known as Monte Carlo sampling,
is the most direct approach to estimate $P(q)$, which involves performing many trials, where in each trial a fault bitstring $e$ is drawn with probability $q^{|e|} (1-q)^{N-|e|}$, and then given the syndrome $\sigma = H e$, the decoder is applied to obtain correction $c = \mathcal{C}(\sigma)$ and we test whether or not $c$ corrects $e$.

Two disadvantages of the Monte Carlo approach to estimate $P(q)$ are: (1) it only provides estimates for a specific value $q$, and (2) for low error rates we may sample many errors which are anyway guaranteed to succeed (for example if the weight of the fault bitstring $e$ is lower than $w_0$ in a trial).
Importance sampling aims to overcome these disadvantages.

\textbf{Importance Sampling: } 
This approach, based on the concentration of contributions in a small range of weights to the sum $P(q) = \mathcal{T}\{f\}(q)$ as illustrated in \fig{failure-spectrum-transform}(b), approximates $P(q)$ by identifying a small set $W$ of fault weights $w$ where the failure spectrum $f(w)$ contributes most to the sum at the relevant $q$.
Then, $f(w)$ is estimated as $\hat{f}(w) = F(w)/T(w)$ for each $w \in W$, by sampling $T(w)$ trials of weight-$w$ noise, of which $F(w)$ result in decoder failure.
For importance sampling, the estimate $\hat{P}_\text{IS}(q)$ of $P(q)$ is then
\begin{eqnarray}
\hat{P}_\text{IS}(q) =  \sum_{w \in W} \hat{f}(w) \binom{N}{W} q^w (1-q)^{N-w}.
\end{eqnarray}
This can deviate from the true value $P$ due to two effects: statistical uncertainty from $\hat{f}(w)$, and the exclusion of terms which are not in $W$, resulting in a missing contribution $0 \leq \Delta P_\text{IS}(q) \leq  \sum_{w \notin W} \binom{N}{w} q^w (1-q)^{N-w}$.
When this missing contribution is small compared with $\hat{P}_\text{IS}(q)$ and the set of estimated $\hat{f}(w)$ values have low statistical error, importance sampling is a very powerful method, requiring far fewer samples than Monte Carlo in the regime $q N \ll w_0$ to obtain comparably accurate estimates of $P(q)$, see \app{LMmethod}.

Importance sampling can be challenging for large codes with moderate physical error rates but very low logical error rates, as is typical in large-scale quantum computing. 
First, in this regime, $f(w)$ can be very small for significantly contributing weights $w \in W$, and an impractically large number of trials may be needed to estimate it accurately. 
Second, the set of significantly contributing weights $W$ can be very large, requiring a correspondingly large number of samples.
These challenges motivate the alternative approach to characterize codes we propose in this work.

\paragraph{Numerical methods:}
When sampling either with fixed error rate $q$ or fixed error weight $w$, when $T$ trials are performed, $F$ of which are observed to fail, our estimate of the failure fraction is $F/T$, and our estimate of the statistical uncertainty on our estimated failure fraction is $\sqrt{(F/T)(1-F/T)/T}$.

\subsection{QEC system examples}
\label{sec:system-examples}

We study a number of examples QEC systems to test our methods, which are summarized in \tab{example-set}.
Here we specify details to allow the reconstruction of each example.

Each example is fully specified by its check matrix $H$, its logical action matrix $A$ and its decoder $\mathcal{C}$.

\textbf{QEC codes: }
We build each of our examples from a QEC code.
Specifically, we consider bivariate bicycle codes~\cite{kovalev2013quantum} with distances $d = 6$, $12$, and $18$ from Ref.\cite{bravyi2024high}, denoted $\text{BB}(d)$; unrotated toric codes~\cite{kitaev2003fault}, denoted $\text{UT}(d)$; rotated toric codes~\cite{wen2003quantum}, denoted $\text{RT}(d)$; and rotated surface codes~\cite{bravyi1998quantum}, denoted $\text{RS}(d)$.
We also consider an asymmetric toric code $\text{UT}(d_1,d_2)$, which is on a rectangular torus with distances $d_1$ and $d_2$ associated with the two cycles around the torus.
Note that all these codes are Calderbank-Shor-Steane (CSS) codes~\cite{calderbank1996good,steane1996multiple}, allowing $X$ and $Z$ errors to be detected and corrected independently.

\textbf{Noise models: }
We consider two noise models.
The first noise model (called \emph{bitflip} in \tab{example-set}) we apply a $Z$ error to each qubit independently with uniform probability $p$, then measure the $X$-type stabilizers of the code.

The second noise model (called \emph{circuit} in \tab{example-set}) is the standard circuit noise model. 
Each way a failure can occur is modeled as a non-trivial Pauli operator being applied to the operation's support, along with a bit flip of measurement outcomes.
State preparations and measurement outcomes fail independently with probability $p$. 
Idle gates fail independently as $X$, $Y$, or $Z$ errors, each with probability $p/3$. 
CNOT gates fail independently as one of the 15 nontrivial two-qubit Pauli errors, each with probability $p/15$.
For circuit noise we parameterize the logical error rate as $P(p)$. 
When computing $P(p)$ inferred by the failure spectrum $f(w)$,
we compute $\mathcal{T}\{f\}(q)$ and then substitute $q \leftarrow p/15$.

All but one of our circuit noise examples are for the fault-tolerant memory of the code.
For this, we use a circuit formed by repeating a QEC cycle (which measures a complete set of stabilizer generators of the code) $d$ times subjected to the noise, followed by a single fault-free 
cycle.\footnote{The fault-free cycle is a standard modeling assumption used to return to the codespace and determine the presence or absence of a logical error. 
See, for instance, \cite{fowler2009high}.}
We use Stim’s built-in rotated-surface-code circuit~\cite{gidney2021stim} and the bivariate-bicycle circuits of Refs.~\cite{bravyi2024high,bravyi2024github}.
Example $\text{BB}(12)$-circuit-$Y_1$ is the measurement of the $Y_1$ logical operator using an ancilla system as specified in Ref.~\cite{BicycleArchitecturePaper}.

\textbf{CSS decoding: }
While the techniques used in this work do not require it, some of the decoders we consider rely on, or perform better with, separate decoding using the syndrome from the $X$-type and $Z$-type checks, which is possible since all of our examples use CSS codes.
We follow the naming convention of Ref.~\cite{bravyi2024github}, where a subscript $X$ is used on an object involved in decoding $X$-type noise (which is detected by $Z$ type checks) etc.

To extract the objects for separate $X$ and $Z$ decoding, start from the full check matrix $H$ (rows ordered: $X$-detectors first, then $Z$-detectors) and the logical action matrix $A$ (rows ordered: logical-$X$ flips, then logical-$Z$ flips). 
Build $H_X$ by selecting from $H$ the columns that trigger at least one $Z$ check and keeping only the $Z$-check rows; define $H_Z$ analogously. 
Some columns appear in both $H_X$ and $H_Z$ when a fault triggers both types. 
Form $A_X$ and $A_Z$ by taking the same column selections from $A$ and then restricting rows to logical-$Z$ (for $A_X$) or logical-$X$ (for $A_Z$).

For simplicity, we generate errors separately for the $X$ and $Z$ systems to estimate logical error rates $P_X$ and $P_Z$, and approximate the total logical error rate of the combined system by $P_X + P_Z$.
This is an approximation because some columns of the original check matrix appear in both $H_X$ and $H_Z$, inducing correlations between the $X$ and $Z$ errors.
We relax this approximation in \app{XYZ-vs-XZ}.

\textbf{Decoder choices: }
For all surface and toric code QEC systems, we use the Matching decoder~\cite{dennis2002topological,Higgott2025sparseblossom} implemented in Ref.~\cite{higgott2022pymatching}. 
For bivariate bicycle codes under bitflip noise, we use BP-OSD \cite{roffe2020} with 100 rounds of BP, followed by combination sweep order 10.
For bivariate bicycle codes under circuit noise, we use the Relay decoder~\cite{RelayPaper} with the parameters therein. 
Specifically, for BB(6), BB(12), and BB(18) we set $S=6$, run 80 iterations in leg~1 and 60 per subsequent leg, up to 600 legs. 
The first leg uses memory strength $\gamma=0.125$ for all nodes. 
For each subsequent leg, each node's $\gamma$ is drawn independently from $\mathrm{Unif}[-0.24,0.66]$ for BB(6) and BB(12), and from $\mathrm{Unif}[-0.161,0.815]$ for BB(18).
We also apply the order-zero version BP-LSD~\cite{hillmann2024lsd} to the same BB-circuit decoding problems for comparison in \sec{validate-relay}, using min-sum BP as a pre-decoder with 30 iterations with a scale factor 0.625.
The decoders are instantiated with the compressed representation at $p=0.001$ for all cases.
In \sec{computing-min-weight-properties}, we use decoding algorithms not as part of a QEC system but as subroutines for finding min-weight properties of the codes, and we provide details therein.

\begin{table}[h!]
\centering
\begin{tabular}{l l l l c c c c}
\toprule
\textbf{QEC system} & \textbf{code} & \textbf{noise} & \textbf{decoding} & $\boldsymbol{\tilde{N}}$ & $\boldsymbol{N}$ & $\boldsymbol{M}$ & $\boldsymbol{K}$ \\
\midrule
$\text{BB}(6)$-bitflip            & bivariate bicycle         & bitflip        & BP-OSD-Z & 72 &  72 & 36 & 12 \\
$\text{BB}(6)$-circuit            & bivariate bicycle       & circuit         & Relay-Z & 2233 &  $46224$ & $252$ & $12$  \\
$\text{BB}(12)$-bitflip           & bivariate bicycle       & bitflip        & BP-OSD-Z & 144 & 144 & 72 &  12 \\
$\text{BB}(12)$-circuit           & bivariate bicycle        & circuit         & Relay-Z & 8785 &  $184896$ & $936$ & $12$  \\
$\text{BB}(12)$-circuit-$Y_1$     & bivariate bicycle      & circuit   & Relay-XYZ & 79591 & 400117 & 2144 & 24 \\
$\text{BB}(18)$-bitflip           & bivariate bicycle        & bitflip        & BP-OSD-Z & 288 & 288 & 144 & 12 \\
$\text{BB}(18)$-circuit           & bivariate bicycle        & circuit         & Relay-Z & 26209 &   $554688$ & $2736$ & $12$  \\
$\text{UT}(d)$-bitflip            & unrotated toric       & bitflip        & Matching-Z & $2d^2$ &  $2d^2$ & $d^2$   & 2 \\
$\text{UT}(d_1,d_2)$-bitflip            & asymmetric toric       & bitflip        & Matching-Z & $2d_1d_2$ &  $2d_1d_2$ & $d_1d_2$   & 2 \\
$\text{RT}(d)$-bitflip            & rotated toric        & bitflip        & Matching-Z & $d^2$ & $d^2$ & $d^2/2$ & 2 \\
$\text{RS}(6)$-circuit            & rotated surface  & circuit         & Matching-Z & 554 & 16556  & 126 & 1
\\
$\text{RS}(12)$-circuit            & rotated surface  & circuit         & Matching-Z & 4778 &  134426 & 936 & 1 
\\
$\text{RS}(18)$-circuit            & rotated surface  & circuit         & Matching-Z & 16562 & 455984 & 3078 & 1 \\
\bottomrule
\end{tabular}
\caption{
QEC system examples.
All entries are memory experiments, except $Y_1$ which measures a single logical-$Y$ operator on one of the 12 logical qubits.
We use -Z to denote CSS decoding of $(H_Z, A_Z)$, while -XYZ denotes non-CSS decoding.
$M$ and $K$ denote the row counts of the check and action matrices respectively. 
In the expanded (compressed) representation, both matrices have $N$ ($\tilde N$) columns.
}
\label{tab:example-set}
\end{table}

\clearpage
\section{Technique I: Fitting to a failure-spectrum ansatz}
\label{sec:ansatz-approach}

In this section, we observe that the failure spectrum $f(w)$, which fully characterizes the performance of a QEC system, appears to have a very consistent behavior across different QEC systems.
We introduce a simple model for quantum error correction in \sec{qec-system-model}, and derive a closed form expression for the failure spectrum for that model.
Then in \sec{model-ansatz}, inspired by this closed form expression, we introduce an ansatz for the failure spectrum for any QEC system, where the ansatz has a number of free parameters that can be fit to data.

\begin{figure}[ht!]
    \centering
    \includegraphics[width=0.47\linewidth]{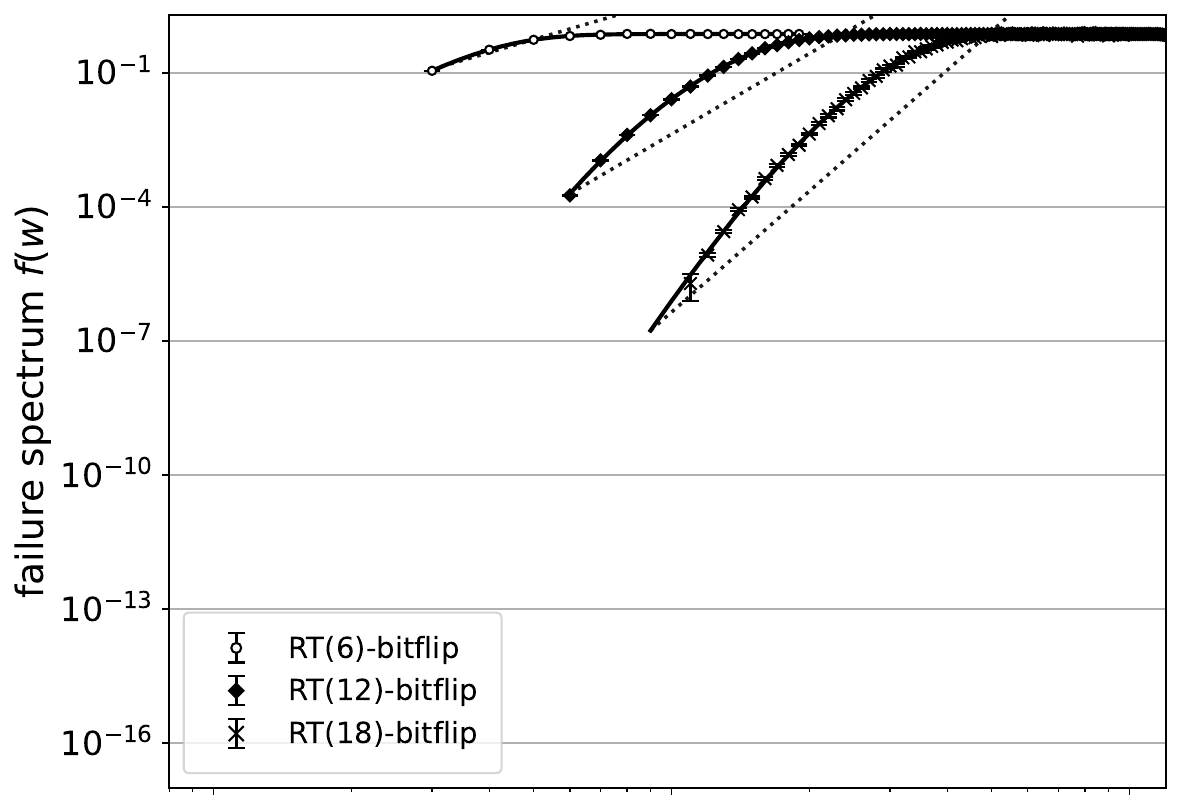}
    \includegraphics[width=0.47\linewidth]{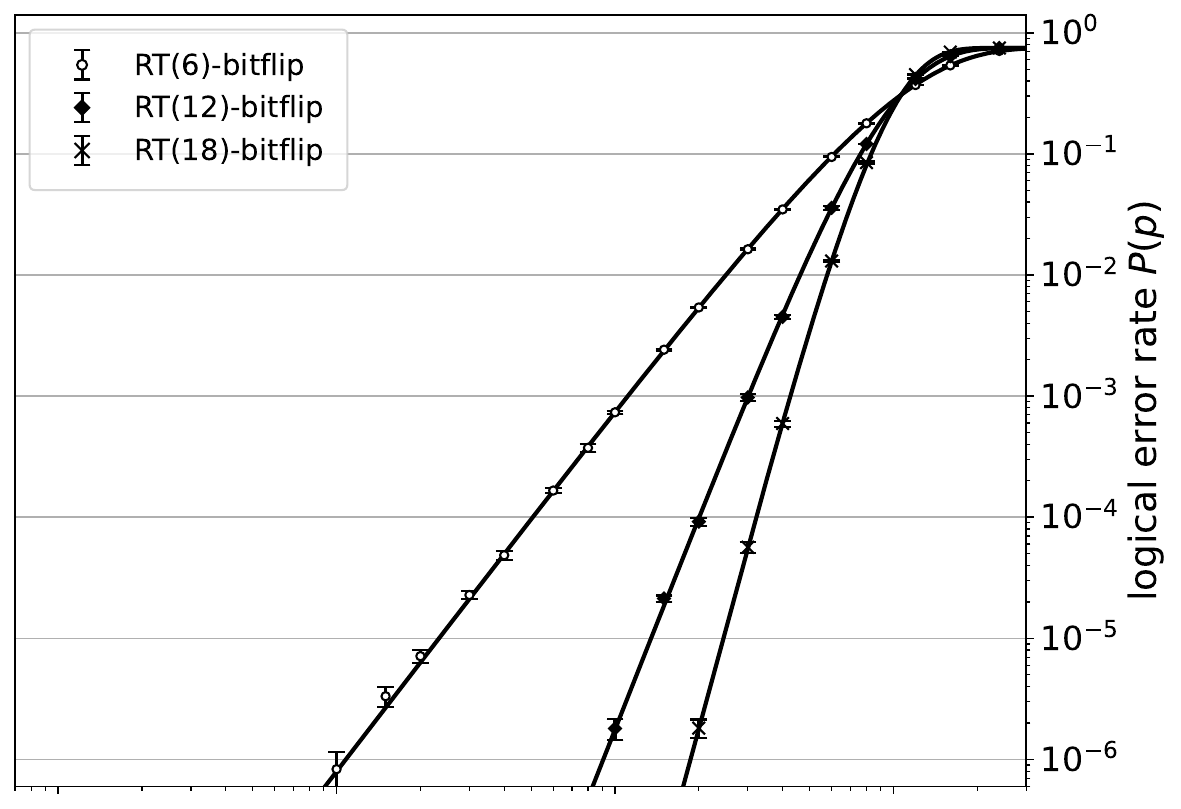}
    \includegraphics[width=0.47\linewidth]{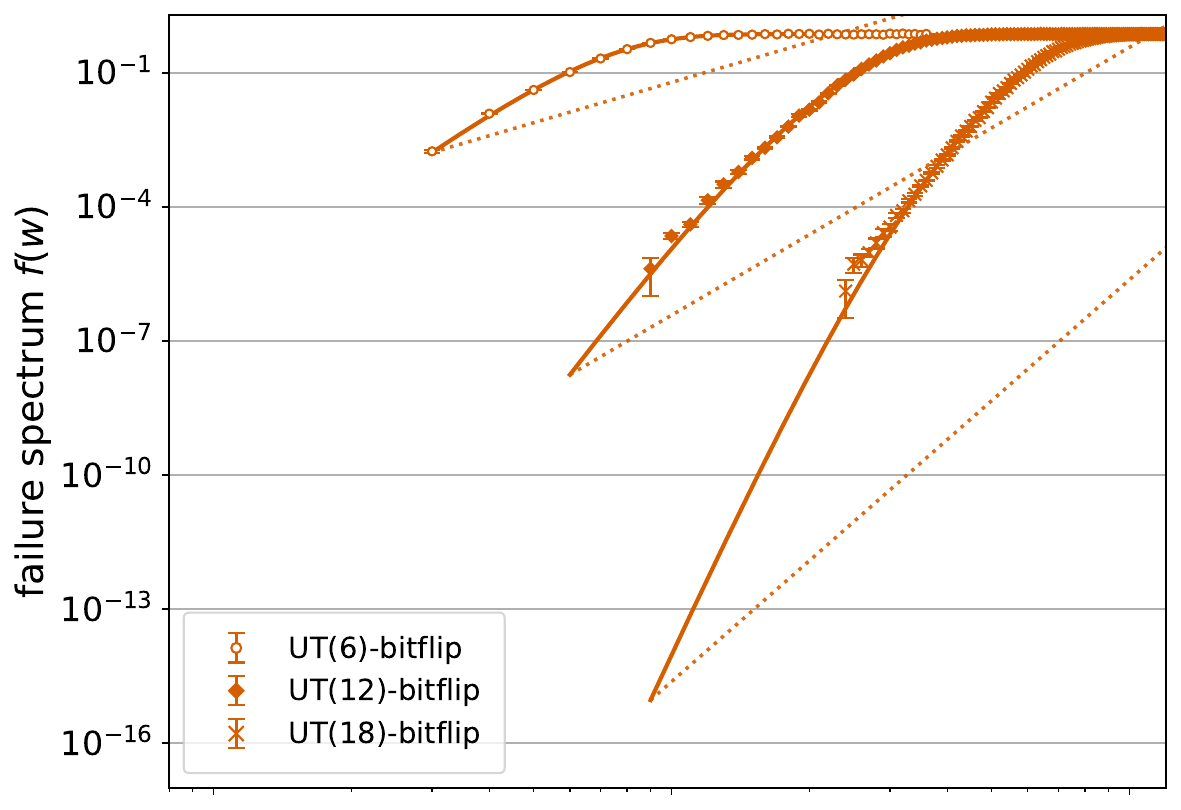}
    \includegraphics[width=0.47\linewidth]{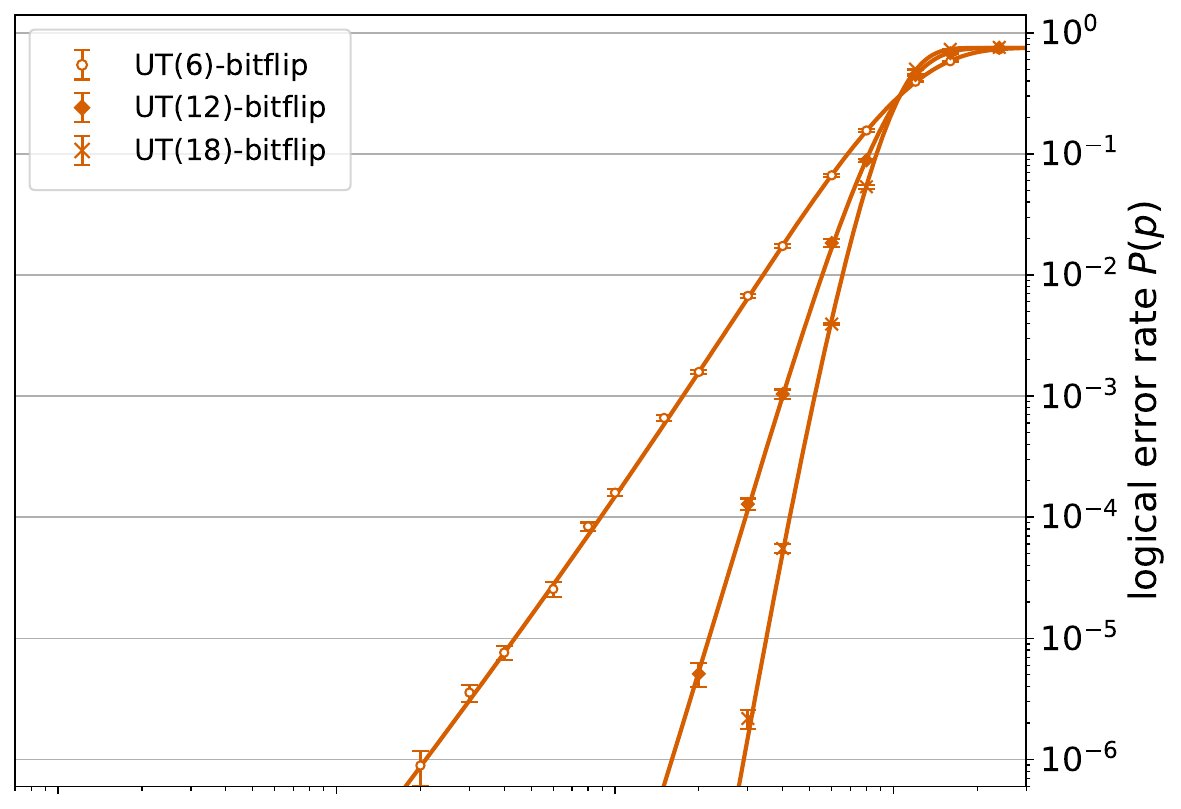}
    \includegraphics[width=0.47\linewidth]
    {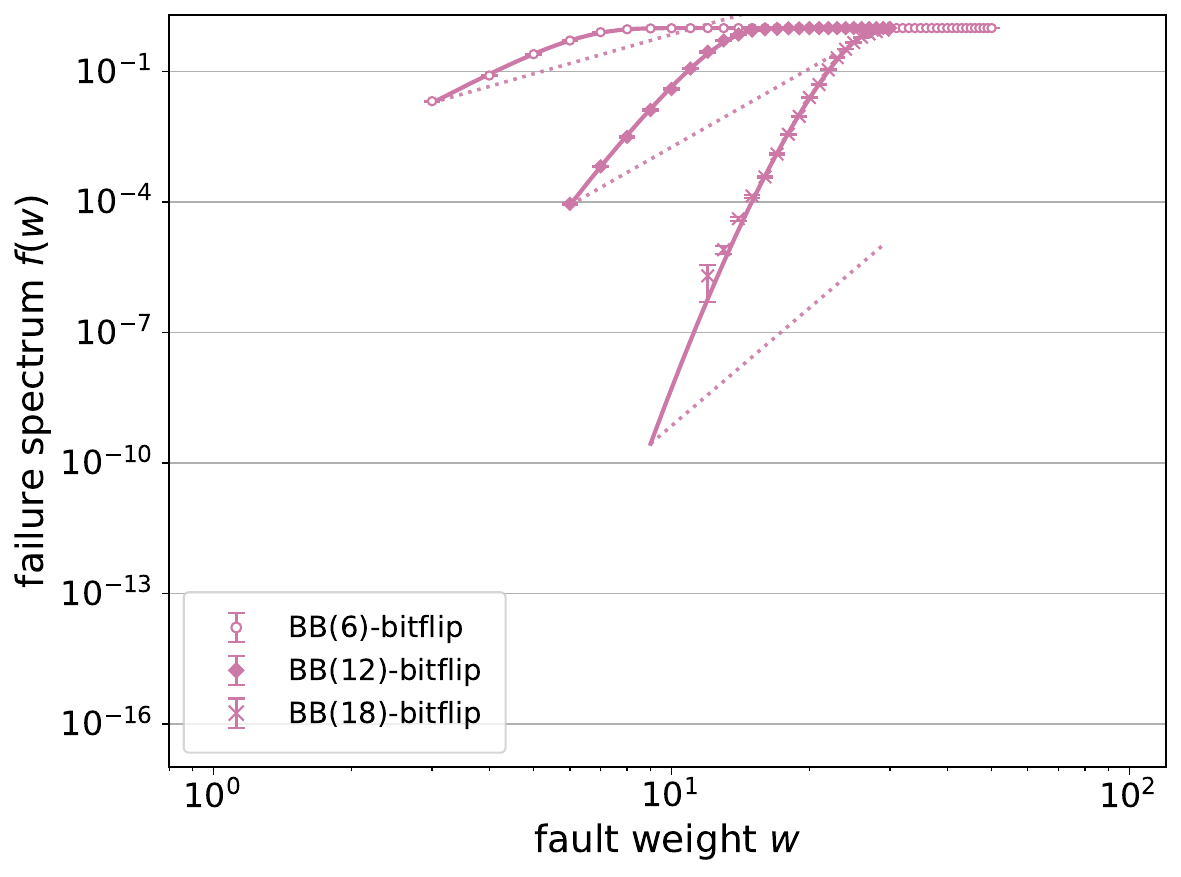}
    \includegraphics[width=0.47\linewidth]
    {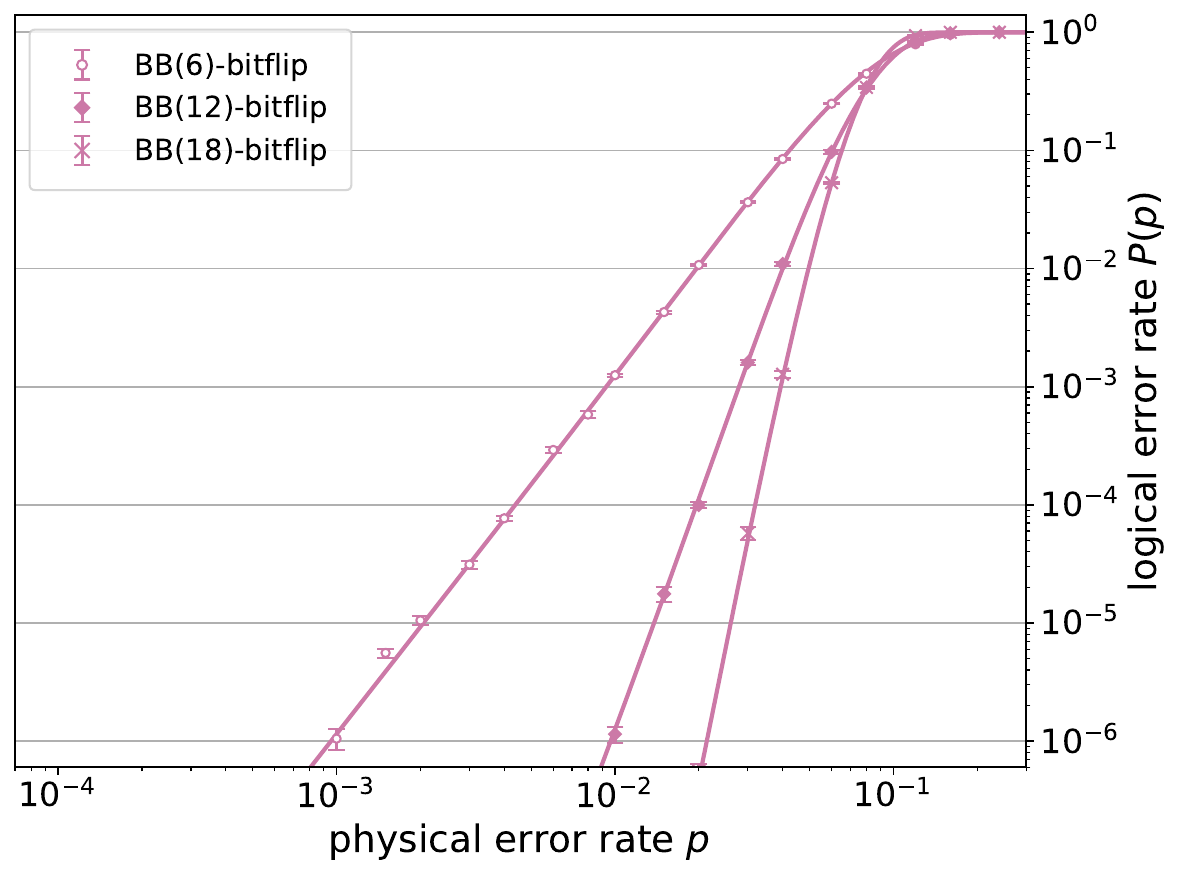}
    \caption{
    Failure spectrum and logical error rate curves for example QEC systems under bit-flip noise from \sec{system-examples}.
    Solid lines in the left subfigures show fits of the ansatz $f^{(5)}_\text{ansatz}(w)$ from \sec{model-ansatz}, and we include dashed lines showing $f_0\cdot (w/w_0)^{w_0}$ for reference (with $f_0$ set by the onset of $f^{(5)}_\text{ansatz}(w_0)$).
    Solid lines in the right subfigures show the corresponding logical error curve $P^{(5)}_\text{ansatz}(p):=\sum_w f^{(5)}_\text{ansatz}(w) \binom{N}{w} p^{w}(1-p)^{N-w}$.
    Subfigures in the same column share a common $x$-axis.
    }
    \label{fig:QEC-system-examples-bitflip}
\end{figure}

\begin{figure}[ht!]
    \centering
    \includegraphics[width=0.47\linewidth]{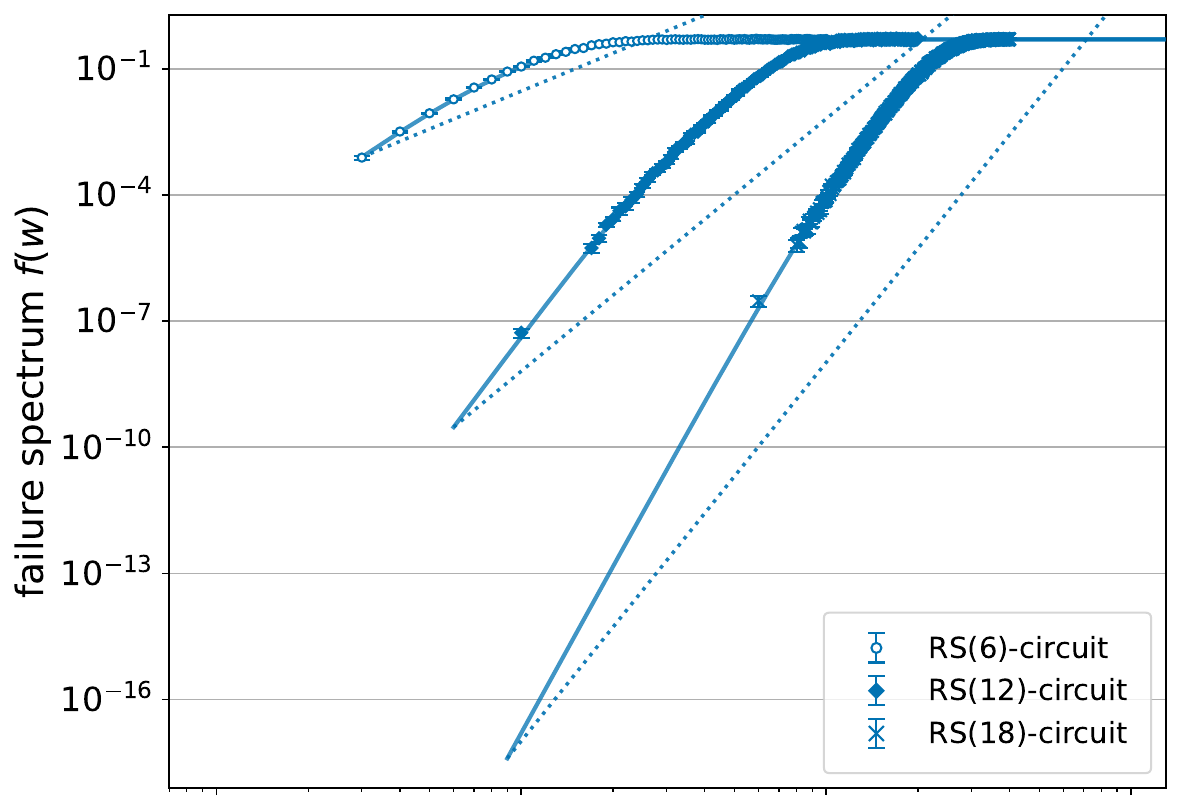}
    \includegraphics[width=0.47\linewidth]{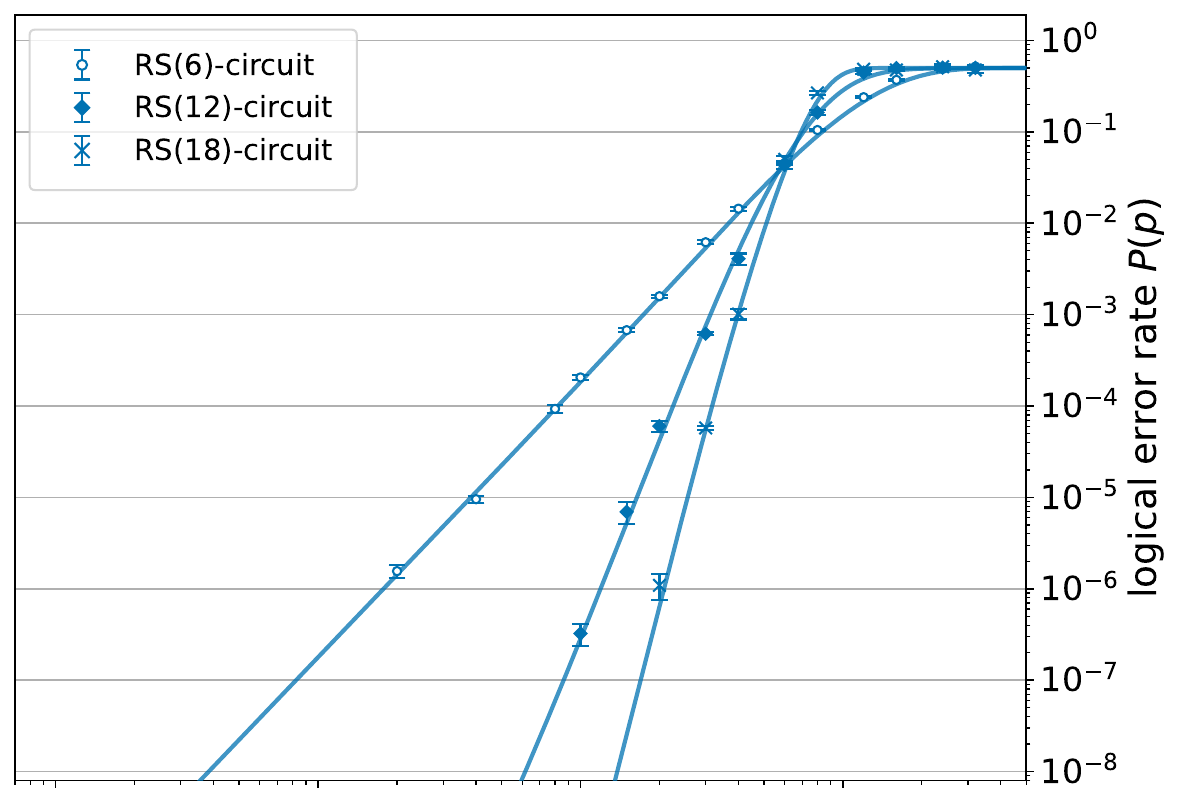}
    \includegraphics[width=0.47\linewidth]{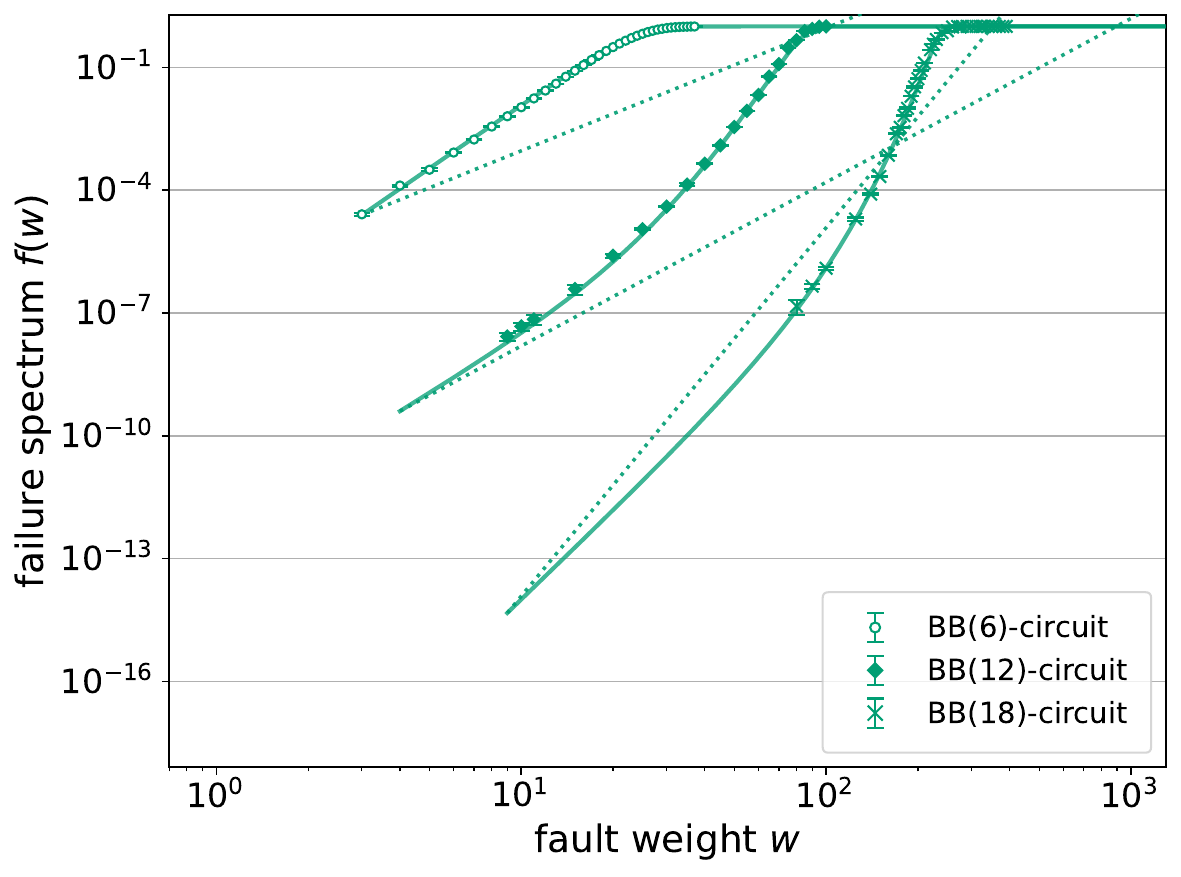}
    \includegraphics[width=0.47\linewidth]{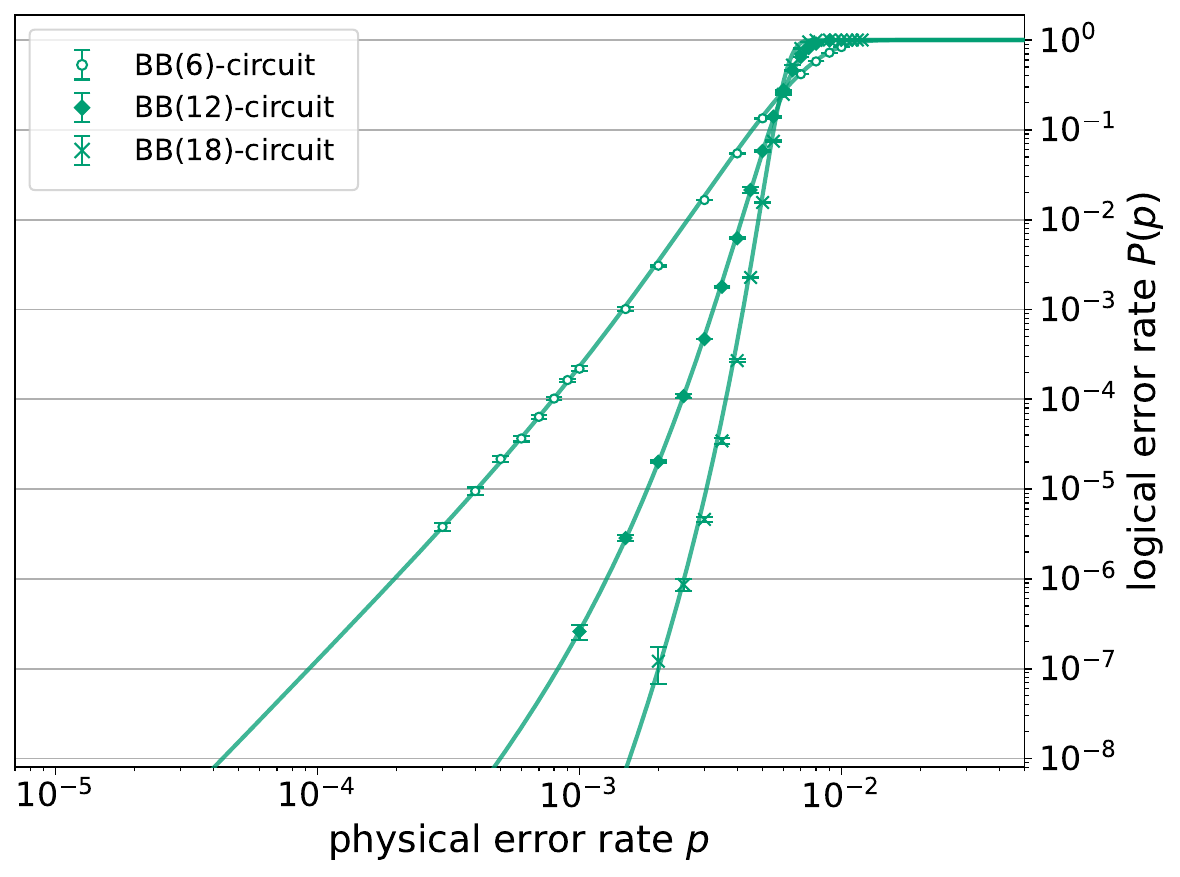}
    \caption{
    Failure spectrum and logical error rate curves for example QEC systems under circuit noise.
    Labeling conventions are as in \fig{QEC-system-examples-bitflip}.
    }
    \label{fig:QEC-system-examples-circuit}
\end{figure}

We summarize the following empirical observations about the failure spectrum $f(w)$ across our examples: 
\begin{itemize} 
    \item[(i)] $f(w)$ increases monotonically with $w$. 
    \item[(ii)] $f(w)$ asymptotes to a value $a = 1 - \frac{1}{2^{K}}$. 
    \item[(iii)] $f(w)$ varies smoothly with $w$. 
\end{itemize}
Some of these features can be understood intuitively.
(i) Adding more errors generally makes decoding harder, leading to a higher failure rate.
(ii) At large $w$, errors start to behave like random bitstrings, which succeed with probability $\frac{1}{2^{K}}$.
As was noted in \fig{failure-spectrum-transform}, for very large values of $w$ comparable to $N$, the value of $f(w)$ may eventually drop, but we ignore this since it has negligible effect on logical error rates in the regime of interest.
(iii) While less obvious, the apparent smoothness of $f(w)$ may arise from a counting argument.
Moreover, the shape of the failure spectrum curves for each system looks very similar. 
This encourages us in this section to look for some universal behavior to formulate an ansatz for $f(w)$.

\subsection{Min-fail enclosure model of QEC systems}
\label{sec:qec-system-model}

Here we propose a simple model of error correction which gives insight into the functional form that can be expected for the fail fraction $f(w)$.

We say a bitstring $e'$ \emph{encloses} $e$ if $e'_i = 1$ for all $i$ where $e_i = 1$, denoted $e \subset e'$.
Recall that $\mathcal{F}(w_0)$ denotes the set of all min-weight failing bitstrings.  
Our model is based on two simplifying assumptions.
First, we assume that a bitstring $b$ fails if and only if it encloses some element of $\mathcal{F}(w_0)$. 
Second, we assume that the bitstrings in $\mathcal{F}(w_0)$ are independent (in a sense that will be made clear below). 
In what follows, we estimate the failure spectrum that results from these assumptions.

First consider the probability $\mathbb{P}(F \subset B)$ that a randomly selected weight-$w$ bitstring $B$ encloses a particular min-weight failing bitstring $F \in \mathcal{F}(w_0)$:
\begin{eqnarray}
\mathbb{P}(F \subset B) = \binom{N-w_0}{w-w_0}/\binom{N}{w},
\end{eqnarray}
since there are $\binom{N}{w}$ weight-$w$ bitstrings to choose from, and $\binom{N-w_0}{w-w_0}$ of them enclose $F$.

Next, consider a pair of bitstrings $F,F' \in\mathcal{F}(w_0)$.
The probability that a randomly selected weight-$w$ bitstring $B$ encloses \emph{either} $F$ or $F'$ is:
\begin{eqnarray}
\mathbb{P}(F \subset B ~\text{or}~ F' \subset B) \approx 1- \left(1-\mathbb{P}(F \subset B) \right) \cdot \left(1-\mathbb{P}(F' \subset B) \right) = 1 - \left(1- \frac{\binom{N-w_0}{w-w_0}}{\binom{N}{w}} \right)^2, \nonumber
\end{eqnarray} 
where, in our model, we treat the events ${F \subset B}$ and ${F' \subset B}$ as independent. 
Extending this analysis to the entire set $\mathcal{F}(w_0)$, and assuming independence of the enclosure events for all $F \in \mathcal{F}(w_0)$, the failure spectrum of the model is
\begin{eqnarray}
f_\text{model}(w) &=& 1 - \left(1-\frac{ \binom{N-w_0}{w-w_0} }{\binom{N}{w}} \right)^{|\mathcal{F}(w_0)|}, \label{eq:model-line1}\\
&\approx& 1 - \exp\left(-  |\mathcal{F}(w_0)| \frac{ \binom{N-w_0}{w-w_0} }{\binom{N}{w}} \right), \label{eq:model-line2}\\ 
&=& 1 - \exp\left(-  f(w_0) \binom{w}{w_0} \right), \label{eq:model-line3}\\
&\approx& 1 - \exp\left(-  f(w_0) \left( \frac{w}{w_0} \right)^{\!\! w_0} \right). \label{eq:model-line4}
\end{eqnarray}
Going from \eq{model-line1} to \eq{model-line2} we have made use of the general bound 
$e^{-np - n p^2} \leq (1-p)^n \leq e^{-np}$  for $p \in [0, 1]$.
Going from \eq{model-line2} to \eq{model-line3} we have made use of the equality $\binom{N-w_0}{w-w_0} /\binom{N}{w} = \binom{w}{w_0} / \binom{N}{w_0}$ and that $f(w_0) = |\mathcal{F}(w_0)|/\binom{N}{w_0}$ by definition.
Going from \eq{model-line3} to \eq{model-line4} we replace the binomial coefficient $\binom{w}{w_0}$ in~\eq{model-line3} with a power law $\left(\frac{w}{w_0}\right)^{w_0}$, which can be easier to work with for large values of $w$ and $w_0$. 
This replacement is motivated by the asymptotic expansion (from~\cite{wikipedia_binomial}, citing~\cite{abramowitz_stegun_1965}), which states that for large $w$ with $w_0/w \to 0$,
\begin{eqnarray}
\binom{w}{w_0} \sim \left( \frac{w e}{w_0} \right)^{w_0} (2 \pi w_0)^{-1/2} \exp\left(- \frac{w_0^2}{2w}(1 + o(1)) \right).\label{eq:binomial-asymptotic}
\end{eqnarray}
The leading $w$-dependence is thus $\sim w^{w_0}$. 
We fix the prefactor by using the form $\left(w / w_0\right)^{w_0}$ to ensure the expected behavior of the failure fraction at $w = w_0$.

While simplistic, the expression in~\eq{model-line4} captures several essential features of realistic QEC failure spectra. 
It takes the correct value at $w = w_0$ (up to the approximation $1 - e^{-x} \approx x$ for small $x$) and increases smoothly and monotonically with $w$. 
Nonetheless, it also exhibits several oversimplifications. 
The large-$w$ asymptote is 1, exceeding the expected saturation value $a$. 
Furthermore, although we observe that the $w_0$-power behavior aligns approximately with many QEC systems, it fails to capture important deviations present in others.
We adopt this model as a starting point and address its oversimplifications in the next subsection to develop a flexible ansatz for the failure spectrum.

\subsection{Failure-spectrum ansatzes}
\label{sec:model-ansatz}

The closed-form expression for the model failure spectrum in~\eq{model-line4} provides a natural functional form. 
By introducing fit parameters, we obtain the following ansatz versions that can be tuned to accurately capture the failure spectra of specific QEC systems:
\begin{eqnarray}
f_\text{ansatz}^{(2)}(w) &=& a \left[ 1 - \exp\left(- \frac{1}{a} f_0  \left(\frac{w}{w_0}\right)^{w_0}  \right) \right],\nonumber\\
f_\text{ansatz}^{(3)}(w) &=& a \left[ 1 - \exp\left(- \frac{1}{a} f_0  \left(\frac{w}{w_0}\right)^\gamma  \right) \right],\nonumber \\
f_\text{ansatz}^{(5)}(w) &=& a \left[ 1 - \exp\left(- \frac{1}{a} f_0  \left(\frac{w}{w_0}\right)^{\gamma_1} \left(\frac{1+(w/w_c)^c}{1+(w_0/w_c)^c} \right)^{\frac{\gamma_2 - \gamma_1}{c}}  \right) \right], ~~\text{with}~c=2.\label{eq:model-ansatz}
\end{eqnarray}
Each ansatz $f_\text{ansatz}^{(l)}(w)$ assumes the failure spectrum is zero for $w < w_0$, with the superscript $l$ indicating the number of parameters. 
The parameters $w_0 \in \mathbb{Z}_{>0}$ and $f_0 \in (0,1)$ denote the onset weight and the corresponding failure fraction; either or both may be treated as fit parameters if not known a priori.

\begin{figure}[ht!]
    \centering
    (a)\includegraphics[width=0.45\linewidth]{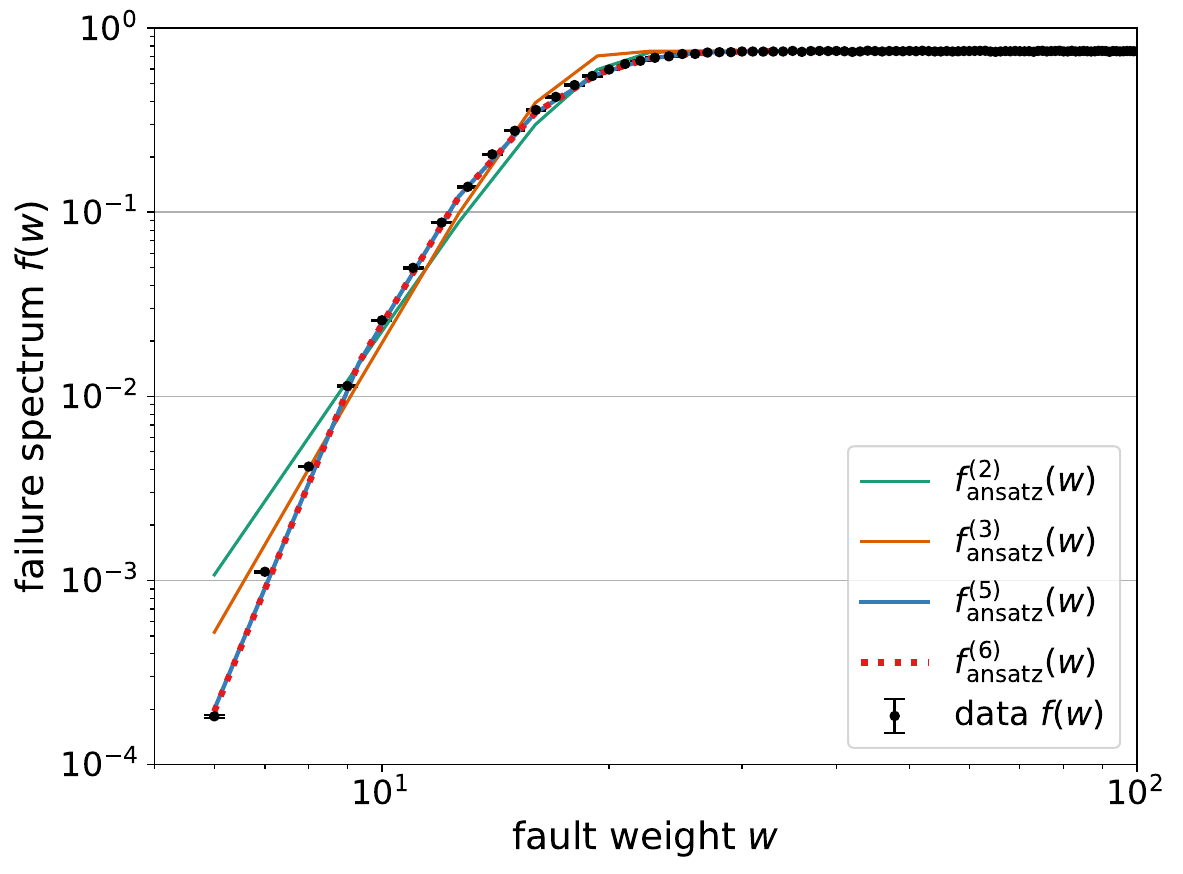}
    (b)\includegraphics[width=0.45\linewidth]{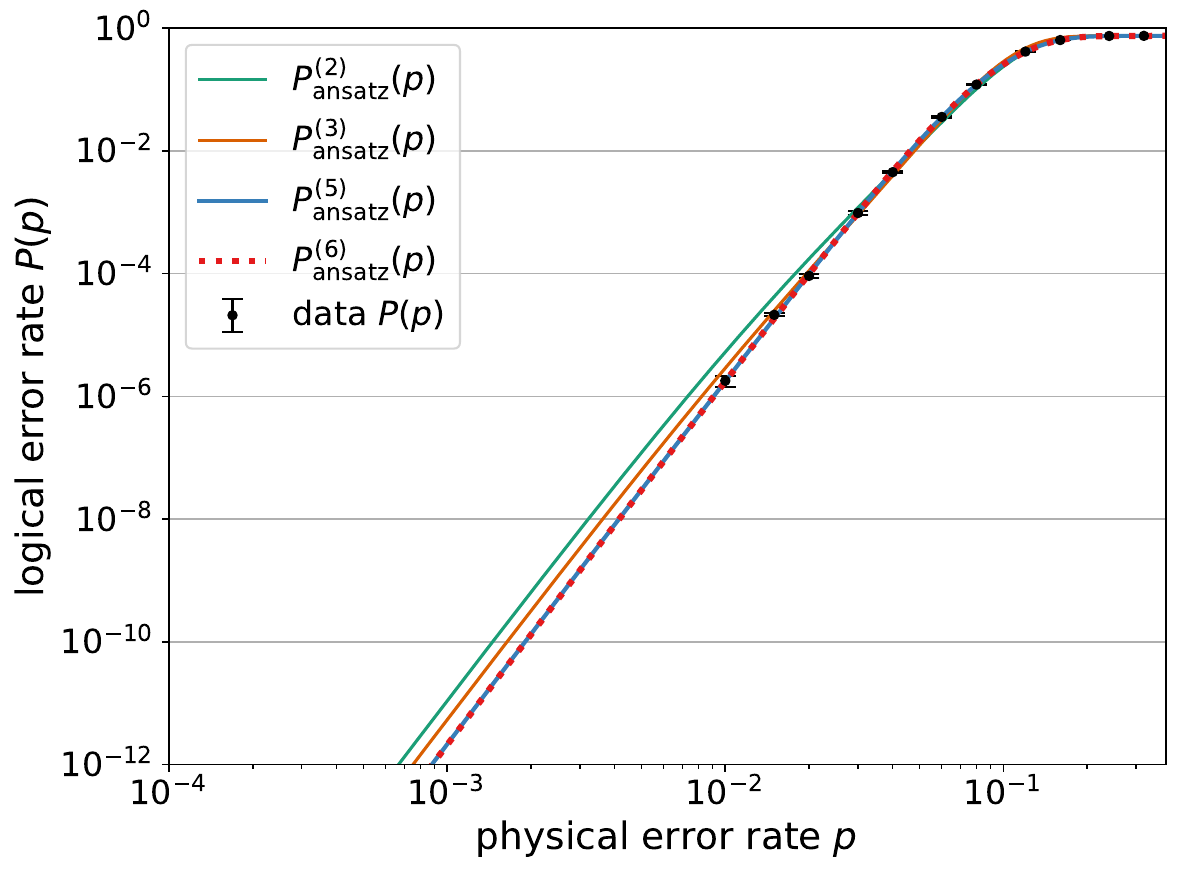}
    (c)\includegraphics[width=0.95\linewidth]{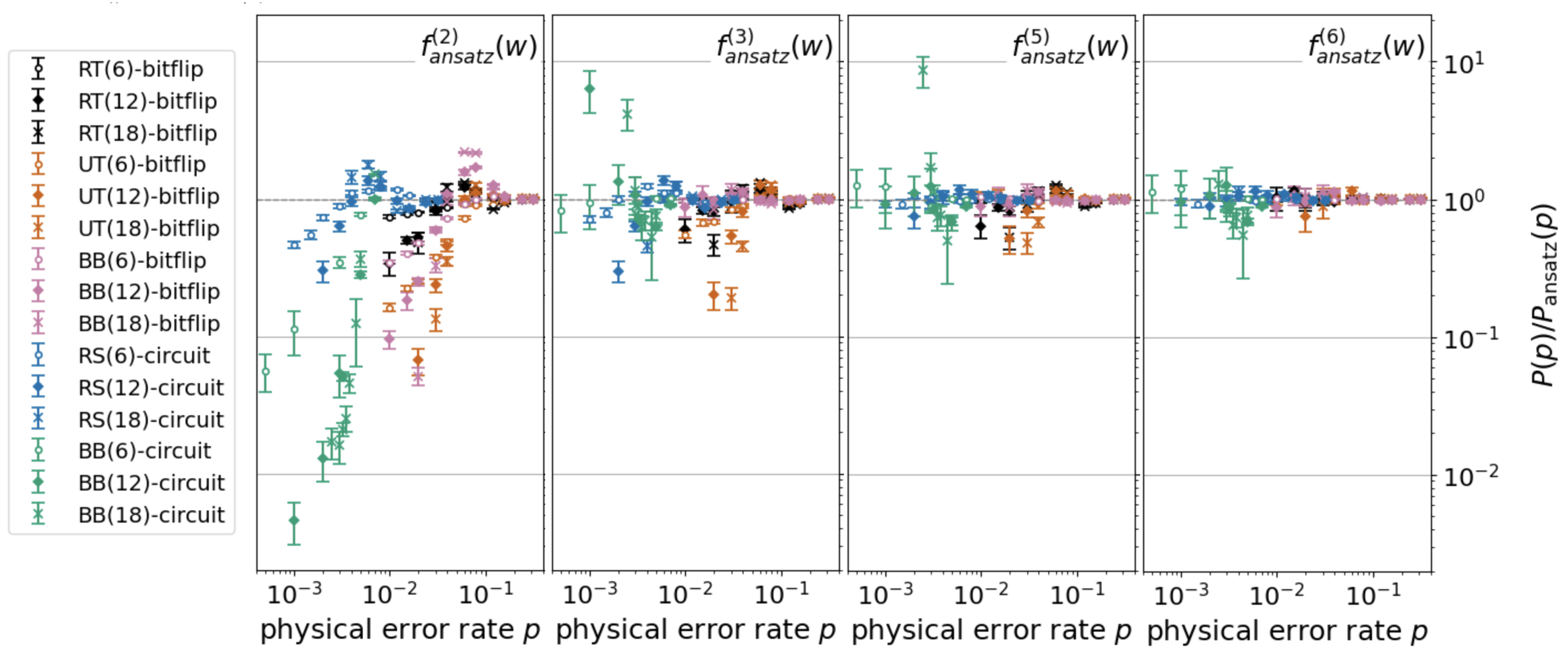}
    \caption{
    (a) Comparison of four ansatz models $f^{(l)}_\text{ansatz}(w)$ fitted to the failure spectrum of the $\mathrm{RT}(12)$ surface code under bitflip noise, with the corresponding inferred logical error rates $P^{(l)}_\text{ansatz}(p)$ shown in (b).
    The fits for $l=5$ and $l=6$ both match the data very closely.
    (c) Ratio of empirical logical error rates to those predicted by each ansatz across all QEC systems shown in \fig{QEC-system-examples-bitflip} and \fig{QEC-system-examples-circuit}. 
    For each system, we fix $w_0$ and fit the remaining $l-1$ parameters of $f^{(l)}_\text{ansatz}(w)$ to failure spectrum data. 
    Agreement improves systematically with increasing ansatz complexity.
    }
    \label{fig:data-ansatz-logical-error-ratios}
\end{figure}

The baseline ansatz, $f_\text{ansatz}^{(2)}(w)$, corresponds to the model in~\eq{model-line4}, modified to include the parameter $a = 1 -1/2^K$, reflecting the expectation that the success probability for random high-weight bitstrings saturates at $1/2^K$ rather than vanishing, as implied by~\eq{model-line4}. 
The exponential argument is scaled by $1/a$ to preserve the behavior of the spectrum in the regime where values are much smaller than one.
In principle, replacing the power-law term $(w / w_0)^{w_0}$ in $f_\text{ansatz}^{(2)}(w)$ with a general function would allow the resulting ansatz to represent any failure spectrum exactly.
However, our aim is to achieve accurate fits for typical QEC systems using as few parameters as possible.
The three-parameter form $f_\text{ansatz}^{(3)}(w)$ extends the baseline ansatz by promoting the fixed exponent $w_0$ to a free parameter $\gamma$.
A further generalization, $f_\text{ansatz}^{(5)}(w)$, replaces the single power law with a smooth interpolation between $(w / w_0)^{\gamma_1}$ at small $w$ and $\sim w^{\gamma_2}$ at large $w$, with $w_c$ controlling the crossover.
If needed, additional generalizations can be introduced to better capture specific system behavior.
For example, our choice of $c=2$ sets how sharply the transition from $w^{\gamma_1}$ to $w^{\gamma_2}$ scaling occurs\footnote{We use $c = 2$ as it closely approximates the binomial for $2 \leq w_0 \leq 20$: $\binom{w}{w_0} \approx (w/w_0)^{\gamma} \left[(1 + (w/w_0)^2)/2\right]^{(w_0 - \gamma)/2}$, with $\gamma = w_0 + (2w_0 - \ln(2\pi w_0))/\ln 2$. 
This reproduces the exact value at $w = w_0$ and the correct large-$w$ prefactor and scaling of \eq{binomial-asymptotic}.
}; $c$ can be promoted to another free parameter to form $f_\text{ansatz}^{(6)}(w)$. 
We compare these four ansatz forms in \fig{data-ansatz-logical-error-ratios}.

As shown in \fig{data-ansatz-logical-error-ratios}, the baseline ansatz $f_\text{ansatz}^{(2)}(w)$ yields $P_\text{ansatz}^{(2)}(p)$ that agrees with some data but deviates significantly elsewhere, reflecting the departure of the min-fail inclusion model’s $w^{w_0}$ scaling in \fig{QEC-system-examples-bitflip} and \fig{QEC-system-examples-circuit}.
The more flexible forms $f_\text{ansatz}^{(3)}(w)$, $f_\text{ansatz}^{(5)}(w)$ and $f_\text{ansatz}^{(6)}(w)$ fit the data well empirically making them useful in practice.
We highlight two limitations of the min-fail enclosure model that likely underlie its shortcomings.
First, the $w^{w_0}$ scaling approximates the combinatorial factor $\binom{w}{w_0}$ in \eq{model-line4}, which rises more steeply at small $w$ before asymptotically approaching $\sim w^{w_0}$ for $w \gg w_0$.
Second, the model assumes enclosures of irreducible failures of weight $w_0$ dominate for all $w$, whereas a sufficiently large number of irreducible failure modes of weight $w_0' > w_0$ could lead to a crossover from $\sim w^{w_0}$ to $\sim w^{w_0'}$ beyond some $w_c$.
These effects are discussed further in \app{beyond-monomial}.

The ansatz family can be extended in several natural directions.
One option is to generalize the progression from the three-parameter to the five-parameter form in \eq{model-ansatz} by introducing a sequence of log–log slopes $\gamma_1,\gamma_2,\gamma_3,\dots$ with transition points at $w_1,w_2,\dots$.
Another option (which we use later in \fig{final-results-decoder-comparison} to model BB(6)-circuit decoded using the bplsd decoder) is to adopt a hybrid model in which the low-weight region $w\approx w_0$, where finite-size effects may disrupt smooth scaling, is captured by explicit terms in $f(w)$, while the tail is described by a smooth ansatz.

\paragraph{Fitting considerations.}
In \app{LMmethod} and \app{fit-strategies-all-p} we use analytic and numerical approaches to better understand strategies for fitting a failure spectrum ansatz to data in different settings.
Based on this analysis, we propose the following general guidelines for sample allocation when the goal is to characterize QEC behavior over a broad range of error rates:
\begin{enumerate}
    \item Distribute samples across weights to yield a roughly uniform number of observed failures per point.
    In particular, avoid collecting vastly more failures at high weights than at low weights, which can bias the fit toward the former regime.
    \item Choose a sufficiently wide range of weights such that the highest-weight data remain distinguishable from $a$, while the lowest-weight data accumulate enough failures to ensure low uncertainty.
    Avoid densely sampling beyond the point where convergence on $a$ occurs.
    \item Select weights with approximately uniform spacing on a logarithmic scale to prevent over-fitting in the high-weight region.
    \item Fit to $P(p)$ as well as $f(w)$ data, particularly when the onset weight $w_0$ is unknown and treated as a fit parameter.
\end{enumerate}

\clearpage
\section{Technique II: Computing or bounding min-weight properties}
\label{sec:computing-min-weight-properties}

In this section, we describe computational methods for characterizing the min-weight features of a QEC system with check and action matrices $H \in \mathbb{F}_2^{M \times N}, A \in \mathbb{F}_2^{K \times N}$.
In \sec{distance}, we review existing techniques to compute the distance $D = \min\{\,|l| : l \in \ker{H}\setminus\ker{A}\,\}$, and to upper-bound it in cases where exact computation is difficult. 
The distance sets the optimal onset $\lceil D/2 \rceil$ which is achieved by a min-weight decoder.
In \sec{logicals}, we review and improve upon existing tecnhiques to find min-weight logicals $\mathcal{L}(D) = \{\,l \in \ker{H}\setminus\ker{A} : |l|=D\,\}$.
In \sec{fails}, we provide techniques to find or estimate the smallest number of min-weight fails $|\mathcal{F}(\lceil D/2 \rceil)|_\text{min}$ that can be achieved by any decoder, which is achieved by a specific min-weight decoder\footnote{The max-class min-weight decoder outputs a correction from the largest set of logically equivalent minimum-weight solutions consistent with the syndrome.
}. 

In \tab{min-weight-properties} we show results for a number of our example QEC systems.
The numbers reported here are for $Z$-type noise for $H_Z$ and $A_Z$. 
We drop the subscripts $Z$ for brevity throughout this section.
Reported numbers are for the expanded representation, but to compute those numbers we often work in the compressed representation and map back.

\begin{table}[h]
    \centering
    \begin{tabular}{lcccccc}
        \toprule
        \textbf{QEC System} & \textbf{Distance} & \textbf{Compressed Logicals} & \textbf{Logicals} & \textbf{Restrictions} & \textbf{Fails}  \\
         & $D$ & $|\tilde{\mathcal{L}}(D)|$ & $|\mathcal{L}(D)|$ & $|\mathcal{L}(D)|_{D/2}|$ & $|\mathcal{F}(D/2)|$ \\
        \midrule
        BB(6)-bitflip      & $6$      & $84$ & $84$                      & $1392$                & -- \\
        BB(12)-bitflip     & $12$     & $1884$ & $1884$                    & $1580496$             & -- \\
        BB(18)-bitflip     & $18$     & $\ge 336$ & $336$                 & $16334304$            & -- \\
        UT$(d)$-bitflip    & $d$      & $2d$ & $2d$                      & $2d\,\binom{d}{d/2}$ & $d\,\binom{d}{d/2}$ \\
        RT$(d)$-bitflip    & $d$      & $d + d\,\binom{d}{d/2}$ & $d + d\,\binom{d}{d/2}$  & --                    & -- \\
        BB(6)-circuit      & $6$      & $\ge 1524$ & $6.01\times 10^{12}$  & $8.12\times 10^{8}$   & $3.83\times 10^{8}$ \\
        BB(12)-circuit     & $\le 10$ &  $\ge3456$ & $1.19\times 10^{17}$  & $6.55\times 10^{12}$  & $9.54\times 10^{11}$ \\
        BB(18)-circuit     & $\le 18$ & $\ge 1981650$ & $5.99\times 10^{34}$  & $5.7\times 10^{23}$ & $2.3\times 10^{23}$ \\
        RS(6)-circuit      & $6$      & $\ge 3562$ & $2.42 \times 10^{12}$ & $9.78 \times 10^{8}$  & $2.13\times 10^{8}$ \\
        RS(12)-circuit     & $12$     & $\ge 4.02\times 10^{7}$ & $3.11\times 10^{24}$                        & $2.6 \times 10^{18}$                   & $3.1 \times 10^{17}$ \\
        RS(18)-circuit     & $18$     & $7.9\times 10^{11*}$ & $9.1\times 10^{37*}$                        & $1.7\times 10^{28*}$                    & $1.0\times 10^{27*}$ \\         
        \bottomrule
    \end{tabular}
    \caption{
    Min-weight properties for QEC systems ($D$ even for all).
    $\mathcal{L}(D)$: set of all minimum-weight logical bitstrings.
    $\mathcal{L}(D)|_{D/2}$: weight-$D/2$ restrictions of $\mathcal{L}(D)$.
    $\mathcal{F}(D/2) \subset \mathcal{L}(D)|_{D/2}$: minimum-weight failing bitstrings for an optimal min-weight decoder.
    The upper bound on $D$ is the weight of the lowest-weight logical operator found; the lower bound on $|\mathcal{L}(D)|$ is the number of unique logical operators found of that weight.
    The quantities $|\tilde{\mathcal{L}}(D)|$, $|\mathcal{L}(D)|_{D/2}|$ and $|\mathcal{F}(D/2)|$ are computed assuming the found subset of $\mathcal{L}(D)$ is complete.
    We estimate the properties of RS(18)-circuit by extrapolation (\app{min-weight-extensions}) since $\tilde{\mathcal{L}}(D)$ is prohibitively large.
    }
    \label{tab:min-weight-properties}
\end{table}

\subsection{Distance}
\label{sec:distance}

Here we describe approaches to compute or upper bound the distance $D$.

\textbf{Exact distance computation: }
We use essentially the same approach as is used to find distances exactly in Ref.~\cite{bravyi2024high} but noting that the integer-programming decoder from \cite{landahl2011fault} can be replaced by any general-purpose min-weight decoder.
First construct an extended check matrix $H^{(i)} \in \mathbb{F}_2^{(M+1) \times N}$ for each $i \in {1, \dots, K}$ by appending the $i$th row of $A$ to $H$.
Set the syndrome $\sigma^{(i)} \in \mathbb{F}_2^{M+1}$ such that $\sigma^{(i)}_j = 1$ only for $j = M + 1$ and use the decoder to find a min-weight correction $F^{(i)} \in \mathbb{F}_2^N$, which represents a logical bitstring for $H$ that is non-trivial for the $i$th row of $A$.
The distance is then $D = \min_{i} |F^{(i)}|$ since any non-trivial logical bitstring must be non-trivial for at least one row of $A$.

\textbf{Upper-bounding the distance: }
We review the approach in \cite{bravyi2024high}.
First, fix integer $T$, and then for each $t =1,2,\dots,T$ do the following. 
Pick a random element $h_t$ in the row space of $H$ and, independently, a random $l_t\neq0$ in the row space of $A$. 
Add them to form $g_t=h_t+l_t$.
Next, construct an extended check matrix $H^{(t)} \in \mathbb{F}_2^{(M+1) \times N}$ by appending $g_t$ to $H$.
Set the syndrome $\sigma^{(t)} \in \mathbb{F}_2^{M+1}$ such that $\sigma^{(t)}_j = 1$ only for $j = M + 1$ and use any general purpose decoder to find a correction $F^{(t)} \in \mathbb{F}_2^N$, which represents a logical bitstring for $H$ that has non-trivial action since it satisfies $HF=0$ and $AF \neq 0$.
An upper bound is then $D_\text{max} = \min_{t} |F^{(t)}|$.

\subsection{Low-weight logicals}
\label{sec:logicals}

Here we introduce approaches to find min-weight logical bitstrings $\mathcal{L}(D)$.
We also describe how to find low-weight logicals, including $\mathcal{L}(D+1)$ which can be useful for odd $D$ as we discuss in \app{min-weight-extensions}.

\textbf{Exact enumeration of all min-weight logical bitstrings: }
Here we briefly describe the approach in Appendix A.5 of Ref.~\cite{DTDpaper} to find all elements of $\mathcal{L}(D)$.
A naive approach would involve exhaustively testing all weight-$D$ faults and filtering those satisfying $HF = 0$ and $AF \neq 0$, but this quickly becomes computationally infeasible. 
For instance, testing all weight-12 $X$-operators in the gross code with bit-flip noise would require evaluating approximately $10^{17}$ cases. 
This approach can be improved upon by leveraging two key insights. 
First, minimum-weight logical bitstrings induce a connected subgraph of the decoding graph $G$, where all fault vertices in a logical bitstring’s support can be connected through paths avoiding fault vertices outside the support. 
This structural property enables an efficient search using a decision tree, where at each layer in the tree, an unsatisfied check is selected, and for each adjacent fault vertex a branch is formed (unless that fault vertex was already visited higher up in the branch).
By starting this process with any single fault, the $D$th level of the tree includes all weight-$D$ logical bitstrings that contain that fault.
This restricts the search space to $N \cdot (r - 1)^{D-1}$ (approximately $10^9$ for the gross code). 
Second, a lower-bound on the weight of a correction for a given syndrome can be used to prune the decision tree, by removing nodes that cannot lead to logical bitstrings of weight $D$, reducing the number of decision-tree nodes that need to be visited to approximately $N \cdot (r - 1)^{D/2}$ (approximately $10^6$ for the gross code).

The same approach can also be applied to find $\mathcal{L}(w)$ for any $w=D+1,D+2, \dots, D+s-1$, where $s$ is the minimum weight of any `stabilizer bitstring', i.e., $s = \min_b(|b| ~\text{s.t.}~ b \in \mathbb{F}_2^N | Hb = 0, Ab=0)$.
For larger weight logicals, there is no-longer a guarantee that they induce a connected subgraph of $G$.

\textbf{Logical operator search: }
This is essentially the same approach as we use to upper bound the distance, but we store all weight-$D$ the logical bitstrings found that have weight $D$ are stored and added to a set. 
The process can be repeated until no new distinct logical bitstrings are found (see \fig{min-weight-estimates}).
The same approach can also be applied in principle to find a subset of $\mathcal{L}(w)$ for any $w$, by simply storing weight-$w$ logicals that are found.
However, if the decoder typically returns minimum-weight corrections, it may be helpful to artificially modify the fault weights fed to the decoding algorithm to find higher-weight corrections more often.

\begin{figure}[ht!]
    \centering
    (a)\includegraphics[width=0.45\linewidth]{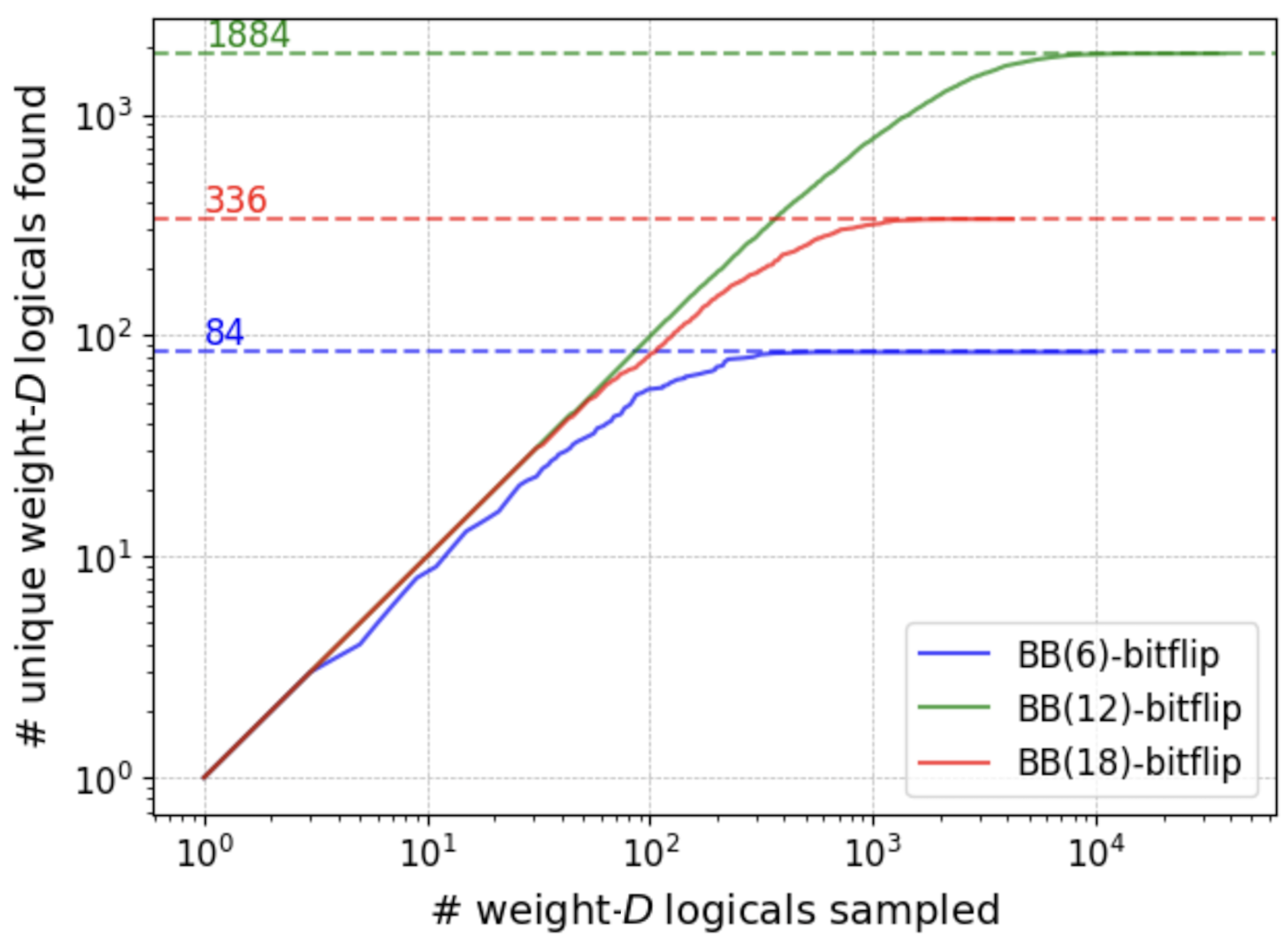}
    (b)\includegraphics[width=0.45\linewidth]{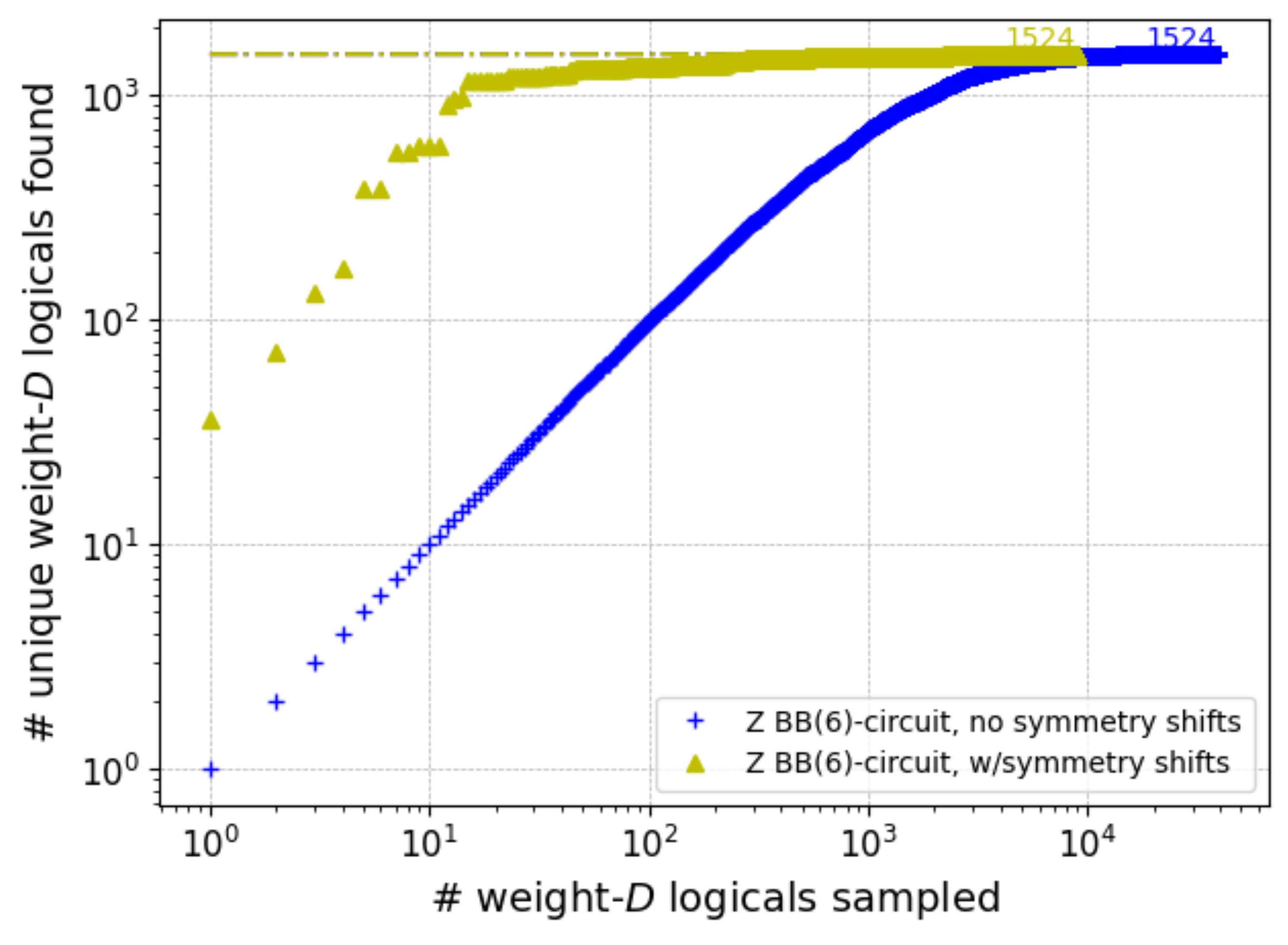}
    \caption{
    (a) Search for $\mathcal{L}(D)$, the set of weight-$D$ logicals, proceeds by sampling (as described in the text) until no new ones appear. 
    For BB(6)-bitflip and BB(12)-bitflip, the method recovers all logicals found by the exact approach, supporting its reliability for larger cases such as BB(18)-bitflip, where it finds $336$ weight-$D$ logicals, conjectured to equal $|\mathcal{L}(D)|$.
    (b) For BB(6)-circuit (compressed $\tilde{H}$, $\tilde{A}$), the symmetry-based method (yellow), which identifies multiple min-weight logicals from each sampled operator, converges substantially faster than the method in (a) (blue).
    }
    \label{fig:min-weight-estimates}
\end{figure}

For some QEC systems and decoder choices, these techniques may be too slow to identify the target logical operator sets in a reasonable time, motivating faster search methods.
Many of the techniques described below are heuristic and lack a clean theoretical justification. 
This is inconsequential for our purposes: every candidate min-weight logical operator is explicitly verified, and any additional ones we uncover only serve to tighten the resulting bounds.

\textbf{Faster search by decimating the appended row:}  
Our method to search for min-weight logicals appends to $H$ an additional check row $g = h + l$, where $h$ is chosen uniformly at random from the row space of $H$ and $l \neq 0$ is chosen uniformly at random from the row space of $A$.  
In practice, $g$ often has high weight, which can degrade decoder performance, as most decoders are designed for check matrices with low column and row weight.  
This can reduce the likelihood of finding a min-weight logical for each random choice of $g$.  
To address this, we employ \emph{decimation}~\cite{mezard2002analytic,montanari2007solving}, in which certain bit positions are fixed prior to decoding.  
Specifically, we fix all bits indexed by $\mathrm{supp}(g)$, the support of the new row $g$.
We assign them so that an odd number are set to $1$, ensuring that the newly added check is unsatisfied as required.  
This assignment modifies the values of other checks, producing a new syndrome $\sigma \in \mathbb{F}_2^M$. 
The decoder is then run on the reduced problem of $N -|\text{supp}(g)|$ unfixed bits and remaining $M$ checks, all now of low weight.  
The decoded solution, combined with the fixed bits, yields a logical operator.  
We find that sampling from low odd-weight assignments can substantially increase the frequency of finding low-weight corrections in the presence of high-weight appended checks. We empirically observe that decimation accelerated the search for min-weight logicals restricted to two QEC cycles, the BB(18)-circuit, for example, after $10^{4}$ decoding instances decimation found 7 logicals of weight 18, the upper bounded min-weight, while without decimation no logicals of weight 18 were found. 
The total time to complete the min-weight logical search instances was also $\approx3\times$ greater than when using decimation, $\approx4.5\times10^{4}$ seconds. 
Finding low weight logicals with decimation is parameter sensitive and an initial parameter optimization was required. 
We note appreciable speed up using decimation was not always found for other circuits.

Another effective technique we found helpful when searching for min-weight logicals is to assign a low uniform prior to all columns in the compressed representation, then perturb these priors randomly between runs. This encourages the search procedure to explore different regions of the space and recover distinct logical operators.

\textbf{Faster search by fault restrictions: } 
Another approach that we find useful is to set a large fraction of bits to zero in the original check matrix.
One way that we make use of this, which has also been considered in prior works, is to use a reduced number of QEC cycles in circuit noise settings, which is mathematically equivalent to setting the bits for faults for all locations in later QEC cycles to zero, eliminating them from the decoding problem. 

\textbf{Faster search using symmetries: }
In some cases, symmetries may be used to simplify the enumeration of logical bitsrings. 
We call a column permutation $\pi:\{1, \dots, N \}\rightarrow \{1, \dots, N \}$ a \emph{symmetry} of $H$ and $A$ if it holds that $\pi(l)$ is a logical bitstring if and only if $l$ is a logical bitstring. 
Here, $\pi(l)$ means that the support of $l$ is permuted according to $\pi$.  
Suppose we have identified a group (under function composition) of symmetries $\Pi$.  
If a min-weight logical bitstring $l$ is obtained by any method, then every element of $\Pi(l) := \{\, \pi(l) : \pi \in \Pi \,\}$ is also a min-weight logical bitstring.  
Thus, many logical bitstrings can be generated from a single example without finding each individually. 

We also define \emph{partial symmetries}, where $\pi : \{1, \dots, N \} \rightarrow \{1, \dots, N \}$ is a permutation such that $\pi(l)$ is a logical bitstring for at least one logical $l$.  
Unlike symmetries, partial symmetries need not form a group and are not guaranteed to map all logicals to logicals.  
However, they can still generate candidates: if $\pi(l)$ satisfies $H\,\pi(l) = 0$ and $A\,\pi(l) \neq 0$, it is stored as a logical operator.

\paragraph{Distance and logical operator computations for \tab{min-weight-properties}:}
All listed distance values and bounds are from prior literature; RT($d$)-bitflip and UT($d$)-bitflip logical operator counts are from~\cite{beverland2019role}.  
Exact logical operator enumeration was used for BB(6)-bitflip and BB(12)-bitflip, while BB(18)-bitflip was handled with the search method, yielding a lower bound on $|\mathcal{L}(D)|$ via BP-OSD. 

Generally, we observe that min-weight decoders (e.g., MaxSat~\cite{noormandipour2024maxsat} or height-bound DTD~\cite{DTDpaper}) were typically slower than heuristic decoders (e.g., BP-OSD~\cite{roffe2020} or BP-LSD~\cite{hillmann2024localized}) and may fail to terminate in reasonable time for large codes. 
For BB-circuit logical operator searches, we used BP-OSD with a maximum of 10,000 BP iterations and an order-100 combination sweep~\cite{roffe2020}.
For RS-circuit logical operator searches, we found a much lighter version of the decoder was sufficient (perhaps because there were so many more logical operators): we used BP-OSD with a maximum of 5 BP iterations and an order-3 combination sweep.

For BB-circuit systems we enumerate min-weight logicals using exact search augmented by symmetries, partial symmetries, and fault restrictions. 
For both bit-flip and circuit noise we exploit shift automorphisms~\cite{bravyi2024high} (as in~\cite{BicycleArchitecturePaper}, App.~A.1). 
For BB(18)-circuit we accelerate the search via decimation, by restricting faults to the first two QEC cycles and applying code automorphism symmetries (finding 230194 compressed logicals), then we generate additional logicals by applying partial time-translation symmetries. 

For all values from logical operator searches in \tab{min-weight-properties}, when we were confident that convergence was close, we ran BP-OSD (without applying symmetries or restrictions) to obtain an additional independent sample set, and ensured that a large fraction of this new sample set were already found.

\subsection{Min-weight fails}
\label{sec:fails}

In this subsection we study the optimal onset of a QEC system given a set of min-weight logical operators 
$\tilde{\mathcal{L}}_\text{found}(D)$ of the compressed representation.
Here we present the approach for even $D$, in which case the the lowest weight that fails is $\lceil D/2 \rceil=D/2$, and describe the extension to odd $D$ in \app{min-weight-extensions}.

Recall that the optimal onset $f^*(D/2)$ is determined by $\mathcal{F}(D/2)$, the set of weight-$D/2$ bitstrings in the expanded representation that fail under a max-class min-weight decoder. 
For even $D$, $\mathcal{F}(D/2)$ is contained entirely within $\mathcal{L}(D)|_{D/2}$.  
We present two techniques. 
First, a direct approach computes $\mathcal{F}(D/2)$ from $\tilde{\mathcal{L}}_\text{found}(D)$, yielding the exact value of $f^*(D/2)$ when the set of found logicals is complete 
($\tilde{\mathcal{L}}_\text{found}(D) = \tilde{\mathcal{L}}(D)$), and otherwise a lower bound. 
Second, for larger QEC systems where enumeration is infeasible, we propose a sampling method that estimates $f^*(D/2)$ under the same completeness assumption, and otherwise provides an estimated lower bound, where the statistical uncertainty can be made arbitrarily small with sufficient sampling.

\paragraph{Max-class min-weight decoder:}
The optimal onset is realized by a specific min-weight decoder for the expanded representation as follows. 
Given a syndrome $\sigma$, we can define an error class $\mathcal{E}_\text{min}(\sigma,a)$ for each $a$:
\[
\mathcal{E}_\text{min}(\sigma,a) = \{\,e \in \mathbb{F}_2^N : H e = \sigma \ \text{and}\ A e = a \,\}.
\]
All errors in $\mathcal{E}_\text{min}(\sigma,a)$ are equivalent.  
The max-class min-weight decoder outputs a correction $c(\sigma)$ from among the largest error classes.
Given set sizes such that:
\[
n_a = |\mathcal{E}_\text{min}(\sigma,a)|, \quad 
n_{\max} = \max_{a'} n_{a'}, \quad 
A_{\max} = \{\,a : n_a = n_{\max}\,\},
\]
then the decoder chooses $a \in A_{\max}$ uniformly at random and sets the correction to $c(\sigma) \in \mathcal{E}_\text{min}(\sigma,a)$.

When computing the optimal onset, the decoder determines whether a given min-weight error leads to failure. 
We define an indicator function $g(e)$, which equals the probability the max-class min-weight decoder fails given an error $e$. 
We extend this definition to the compressed representation by setting $g(\tilde{e}) = 1$ if the max-class min-weight decoder fails for an error $e$ corresponding to $\tilde{e}$, and $g(\tilde{e}) = 0$ otherwise.

\paragraph{Exact optimal onset computation:}
The optimal onset is 
$f^*(D/2) = |\mathcal{F}(D/2)|/\binom{N}{D/2}$,
where $\mathcal{F}(D/2)$ is the set of failing weight-$D/2$ bitstrings in the expanded representation under the max-class min-weight decoder.  
A key property is provided in \propos{ld2-characterization}

\begin{prop}\label{propos:ld2-characterization}
Let $H \in \FF_2^{M \times N}$ and $A \in \FF_2^{K \times N}$, with even distance $D = \min\{\,|l| : l \in \ker{H}\setminus\ker{A}\,\}$.
Let $\mathcal{L}(D) = \{\,l \in \ker{H}\setminus\ker{A} : |l|=D\,\}$ be the set of weight-$D$ logicals, and
$\mathcal{L}(D)|_{D/2}$ their weight-$D/2$ restrictions.
Let $\mathcal{F}(D/2)$ be the set of weight-$D/2$ failing errors under a min-weight decoder.  
Define the set
\[
\mathcal{X} = \{\,x \in  \FF_2^{N} : |x|=D/2,\ \exists f\in \mathcal{F}(D/2)\ \text{with } Hx=Hf\,\}.
\]
Then the following holds:
\[
\mathcal{L}(D)\big|_{D/2} \,=\, \mathcal{X} \, \supset \, \mathcal{F}(D/2).
\]
\end{prop}

\begin{proof}
\emph{($\mathcal{X} \supset \mathcal{F}(D/2)$)}  
Let $f \in \mathcal{F}(D/2)$.  
Then $f$ must also be in $\mathcal{X}$ since $|f| = D/2$ and $Hf = Hf$.

\emph{($\mathcal{L}(D)\big|_{D/2} \subseteq \mathcal{X}$)}  
Let $x \in \mathcal{L}(D)|_{D/2}$.  
Then there exists a weight-$D$ logical $\ell \in \ker{H}\setminus\ker{A}$ with $\mathrm{supp}(x) \subseteq \mathrm{supp}(\ell)$.  
Define $\bar{x} := \ell + x$.  
We have $|x| = |\bar{x}| = D/2$, $Hx = H\bar{x}$, and $Ax \neq A\bar{x}$.  
Thus the syndrome $\sigma := Hx$ admits at least two distinct actions among its min-weight errors (those of $x$ and $\bar{x}$).  
Any deterministic min-weight decoder must then fail on at least one such action.  
Hence there exists $f \in \mathcal{F}(D/2)$ with $Hf = \sigma$.
Since $|x|=D/2$ and $Hx = Hf$, we see that $x \in \mathcal{X}$.

\emph{($\mathcal{L}(D)\big|_{D/2} \supseteq \mathcal{X}$)}  
Let $x \in \mathcal{X}$, such that $x$ is a weight-$D/2$ error with $Hx = Hf$ for some $f \in \mathcal{F}(D/2)$.  
Since $f$ fails, there exists a min-weight error $e$ with the same syndrome $Hx$ but different action, $Ae \neq Ax$.  
Then $\ell := x + e$ satisfies $H\ell = 0$ and $A\ell \neq 0$, so $\ell$ is a nontrivial logical.  
Moreover $|\ell| \le |x|+|e| = D$, while by definition of distance $|\ell|\ge D$; hence $|\ell|=D$.  
Thus $x$ is a weight-$D/2$ restriction of the logical $\ell$, i.e.\ $x \in \mathcal{L}(D)|_{D/2}$.

Therefore $\mathcal{F}(D/2)$ is contained in $\mathcal{X}$, and $\mathcal{X}$ is equal to $\mathcal{L}(D)|_{D/2}$.
\end{proof}

Thus $\mathcal{F}(D/2)$ can be identified by partitioning $\mathcal{L}(D)|_{D/2}$ by syndrome and logical action, and counting set sizes.  
Working in the compressed representation and mapping to the expanded one, the procedure is:

\begin{enumerate}
    \item \emph{Restriction.}  
    Extract $\tilde{\mathcal{L}}(D)|_{D/2}$, the set of all weight-$D/2$ restrictions of min-weight logicals $\tilde{\mathcal{L}}(D)$.
    
    \item \emph{Partition.}  
    Decompose into disjoint error classes labeled by syndrome $\sigma$ and logical action $a$:
    \[
    \tilde{\mathcal{L}}(D)\big|_{D/2}
      = \bigsqcup_{\sigma} \;\bigsqcup_{a \in \mathcal{A}(\sigma)} 
        \tilde{\mathcal{E}}_{\min}(\sigma,a),
    \]
    where $\mathcal{A}(\sigma)$ denotes the set of logical actions compatible with $\sigma$. 
    
    \item \emph{Map to the expanded representation.}  
    Each compressed bitstring $\tilde{e}$ corresponds to $\rho(\tilde{e}) = \prod_{j \in \mathrm{supp}(\tilde{e})} m_j$ expanded bitstrings with the same $(\sigma,a)$.  
    Hence $|\mathcal{E}_\text{min}(\sigma,a)| = \sum_{\tilde{e} \in \tilde{\mathcal{E}}_\text{min}(\sigma,a)} \rho(\tilde{e})$,
    and 
    \[
    |\mathcal{L}(D)|_{D/2}| = \sum_{\sigma} \sum_{a \in \mathcal{A}(\sigma)} |\mathcal{E}_\text{min}(\sigma,a)|.
    \]
    For each pair $(\sigma,a)$ in this decomposition, we record the cardinality 
    $\lvert \mathcal{E}_{\min}(\sigma,a) \rvert$.
    
    \item \emph{Count failures.}  
    For each $\sigma$, the decoder succeeds only in the largest error class $\mathcal{E}_\text{min}(\sigma,a^*(\sigma))$. 
    Here $a^*(\sigma)$ is the action $a \in \mathcal{A}(\sigma)$ such that $|\mathcal{E}_\text{min}(\sigma,a)|$ is largest (if there are multiple values of $a$ which have equally large set sizes, we take one of them).
    All bitstrings in $\mathcal{E}_\text{min}(\sigma,a)$ with $a \neq a^*(\sigma)$ are failures.  
    Hence,
    \[
    |\mathcal{F}(D/2)| = \sum_{\sigma} \sum_{a \in \mathcal{A}(\sigma)\setminus a^*(\sigma)} |\mathcal{E}_\text{min}(\sigma,a)|.
    \]
\end{enumerate}

This procedure yields the exact $|\mathcal{F}(D/2)|$, and therefore the exact onset $f^*(D/2)$.
When the same procedure is applied to a subset $\tilde{\mathcal{L}}_\text{found}(D) \subseteq \tilde{\mathcal{L}}(D)$ rather than to $\tilde{\mathcal{L}}(D)$, it yields a rigorous lower bound on $|\mathcal{F}(D/2)|$.

\paragraph{Approximate optimal onset by sampling:}
We now describe a sampling method to estimate the size of the failure set $\mathcal{F}(D/2)$.
While min-weight fail counts in \tab{min-weight-properties} were obtained by the exact method described above, we foresee this sampling method being useful for even larger QEC systems, where both $|\tilde{\mathcal{L}}(D)|$ and $D$ are large, making enumeration of all weight-$D/2$ restrictions infeasible.  
Our approach is based on the fact that $|\mathcal{F}(D/2)|$ can be obtained by the following sum:
\begin{eqnarray}
|\mathcal{F}(D/2)| &=& \sum_{\tilde{e} \in \tilde{\mathcal{L}}(D)|_{D/2}} \rho(\tilde{e})g(\tilde{e}), \label{eq:estimate-justification-1}\\
&=& \sum_{\tilde{l} \in \tilde{\mathcal{L}}(D)}  \left(\sum_{\tilde{r} \in \tilde{l}|_{D/2} } \frac{\rho(\tilde{r})g(\tilde{r})}{\mu(\tilde{r})} \right), \label{eq:estimate-justification-2}
\end{eqnarray}
Here $\tilde{l}|_{D/2}$ is the set of weight-$D/2$ restrictions of $\tilde{l}$, $g(\tilde{e})$ is the failure indicator for the max-class min-weight decoder, and $\mu(\tilde{r})$ is the number of logicals in $\tilde{\mathcal{L}}(D)$ which enclose $\tilde{r}$. 
In \eq{estimate-justification-1} each element of $\tilde{\mathcal{L}}(D)|_{D/2}$ appears once, whereas in \eq{estimate-justification-2} it appears $\mu(\tilde{r})$ times, and the factor $1/\mu(\tilde{r})$ corrects for this repetition.

Our approach to estimate $|\mathcal{F}(D/2)|$ is to uniformly sample from the $|\tilde{\mathcal{L}}(D)| \binom{D}{D/2}$ terms in the sum in \eq{estimate-justification-2}. 
To do this, we first uniformly sample among the set of logicals, to obtain an element $\tilde{l}^{(t)} \in \tilde{\mathcal{L}}(D)$, and then sample uniformly among the $\binom{D}{D/2}$ weight-$D/2$ restrictions of $\tilde{l}^{(t)}$ to form $\tilde{r}^{(t)}$.
We repeat this process to obtain $T$ samples, and our estimate is 
\begin{eqnarray}
|\hat{\mathcal{F}}(D/2)| &:=& |\tilde{\mathcal{L}}(D)| \binom{D}{D/2}  ~\frac{1}{T} \sum_{t=1}^T \frac{\rho(\tilde{r}^{(t)})g(\tilde{r}^{(t)})}{\mu(\tilde{r}^{(t)})}.
\label{eq:estimate}
\end{eqnarray}
To implement this sampling strategy, we compute terms on the fly via \alg{term_computation}.

\begin{algorithm}
\caption{Compute the contribution $\rho(\tilde{r})\,g(\tilde{r})/\mu(\tilde{r})$ for a fixed restriction $\tilde{r}$}
\label{alg:term_computation}
\begin{algorithmic}[1]
\State \textbf{Input:} $\tilde{\mathcal{L}}(D)$; an element $\tilde{\ell}_0\in\tilde{\mathcal{L}}(D)$ and a weight-$D/2$ restriction $\tilde{r}\subset \tilde{\ell}_0$.
\State \textbf{Output:} $\rho(\tilde{r})\,g(\tilde{r})/\mu(\tilde{r})$.
\State $\rho \gets \prod_{j \in \mathrm{supp}(\tilde{r})} m_j$ \Comment{compute multiplicity of $\tilde{r}$ in expanded rep.}
\State \textbf{// loop over logicals}
\State $\mu \gets 0,\quad \tilde{\mathcal{E}}_{\min}(\sigma) \gets \{\}, \quad \tilde{r}' \gets \tilde{\ell}_0 \setminus \tilde{r}$ \Comment{initialize repetition counter, error set and error complement}
\For{each $\tilde{\ell}\in\tilde{\mathcal{L}}(D)$}
    \State \textbf{if } $\mathrm{supp}(\tilde{r}) \subseteq \mathrm{supp}(\tilde{\ell})$ \textbf{ then } $\tilde{\mathcal{E}}_{\min}(\sigma) \gets \tilde{\mathcal{E}}_{\min}(\sigma) \cup \{\,\tilde{\ell}\setminus \tilde{r}\,\},\ \mu \gets \mu+1$ \Comment{add complement of $\tilde{r}$ and update count}
    \State \textbf{if } $\mathrm{supp}(\tilde{r}') \subseteq \mathrm{supp}(\tilde{\ell})$ \textbf{ then } $\tilde{\mathcal{E}}_{\min}(\sigma) \gets \tilde{\mathcal{E}}_{\min}(\sigma) \cup \{\,\tilde{\ell}\setminus \tilde{r}'\,\}$ \Comment{add complement of $\tilde{r}'$}
\EndFor 
\Comment{now $\tilde{\mathcal{E}}_{\min}(\sigma)$ consists of all weight-$D/2$ errors in the compressed rep. with same syndrome as $\tilde{r}$}
\State \textbf{// max-class failure probability} 
\State $n_a \gets \sum_{\tilde{q}\in \tilde{\mathcal{E}}_{\min}(\sigma):\ a(\tilde{q})=a}\ \prod_{j \in \mathrm{supp}(\tilde{q})} m_j \quad \text{for each action } a$ \Comment{size of error classes (in expanded rep.)} 
\State $n_{\max}\gets \max_a n_a$, \quad $A_{\max}\gets \{\,a : n_a = n_{\max}\,\}$, \Comment{identify largest error class (or classes if a tie)} \State \textbf{if } $a(\tilde{r}) \in A_{\max}$ \textbf{ then } $g \gets 1 - 1/|A_{\max}|$ \textbf{ else } $g \gets 1$ 
\State \Return $(\rho \cdot g)/\mu$ \end{algorithmic} 
\end{algorithm}

The justification that this strategy for building the set $\tilde{\mathcal{E}}_{\min}(\sigma)$ is captured in \propos{closure-correctness}.
This procedure yields an estimate of $|\mathcal{F}(D/2)|$.
When the same procedure is applied to a subset $\tilde{\mathcal{L}}_\text{found}(D) \subseteq \tilde{\mathcal{L}}(D)$ rather than to $\tilde{\mathcal{L}}(D)$, it yields an estimate of a lower bound on $|\mathcal{F}(D/2)|$.
In both cases, the accuracy of the estimate can be improved by increasing the number of samples $T$.

\begin{prop}\label{propos:closure-correctness}
Let $H \in \FF_2^{M \times N}$ and $A \in \FF_2^{K \times N}$, with even distance $D = \min\{\,|l| : l \in \ker{H}\setminus\ker{A}\,\}$.  
Let $\mathcal{L}(D) = \{\,\ell \in \ker{H}\setminus\ker{A} : |\ell|=D\,\}$ be the set of weight-$D$ logicals.  
For $\ell_0 \in \mathcal{L}(D)$ and a restriction $r \subseteq \ell_0$ of weight $D/2$, let $r'=\ell_0\setminus r$.  
Then the set of all weight-$D/2$ errors with syndrome $\sigma=Hr$ is
\[
\mathcal{E}_{\min}(\sigma) = \{\,\ell\setminus r : \ell\in\mathcal{L}(D),\ r\subseteq \ell\,\}
 \;\cup\;
\{\,\ell\setminus r' : \ell\in\mathcal{L}(D),\ r'\subseteq \ell\,\}.
\]
\end{prop}

\begin{proof}
\emph{($\supseteq$)}  
By definition, every element on the right-hand side has the form 
$\ell\setminus r$ for some $\ell\in\mathcal{L}(D)$ with $r\subseteq \ell$, or the form 
$\ell\setminus r'$ for some $\ell\in\mathcal{L}(D)$ with $r'\subseteq \ell$.  
In either case this is a weight-$D/2$ restriction with syndrome $\sigma=Hr=Hr'$, so it lies in $\mathcal{E}_{\min}(\sigma)$.

\emph{($\subseteq$)}  
Conversely, consider any element on the left-hand side $q\in\mathcal{E}_{\min}(\sigma)$.  
Then $|q|=D/2$ and $Hq=\sigma$.  
Consider $q+r$ and $q+r'$.  
Both lie in $\ker{H}$, and at least one must be a nontrivial logical (if both were trivial, then $(q+r)+(q+r')=r+r'=\ell_0$ would be trivial, a contradiction).  
Such a logical has weight at least $D$, but since $|q+r|,|q+r'|\le D$, it must have weight exactly $D$.  
Thus $q$ is a weight-$D/2$ restriction of a nontrivial logical, either $q+r$ or $q+r'$, and so appears on the right-hand side.

Therefore the two sets coincide.
\end{proof}

\textbf{Min-weight fail computations for \tab{min-weight-properties}: } 
The restrictions and min-weight fails are computed using the exact optimal onset computation, except for RS(12)-circuit and BB(18)-circuit for which we use the approximate optimal onset by sampling, and RS(18)-circuit, which is estimated by extrapolation as described in \app{min-weight-extensions}.

\clearpage
\section{Technique III: Multi-seeded splitting for general QEC systems} 
\label{sec:splitting-method}

The `splitting method' forms a general framework for analyzing rare events~\cite{bennett1976}, although in this work we use the term specifically for its application to QEC based on the approach in Ref.~\cite{bravyi2013simulation}.
In \sec{splitting} we restate the method, which was written for surface codes in Ref.~\cite{bravyi2013simulation}, in a general language applicable to any binary decoding problem. 
We note that the concurrent, independent work Ref.~\cite{mayer2025} also provides a formulation of the splitting method for circuit noise and applies it to rotated surface codes. 
In \sec{asymm_toric} we highlight a problem that can arise for the splitting method which can arise for systems encoding multiple logical qubits, which we address by including a modification to the splitting method in \sec{chain_init}.

\subsection{Overview of the splitting method}\label{sec:splitting}

Here we restate the splitting method in a general language applicable to any binary decoding problem. 
Let $\mathcal{F}$ denote the set of all failing bitstrings.
For probabilities $p_0, p_1, \dots, p_t$ with $p_t = p$ and known $P(p_0)$, define $\pi_i(E) = p_i^{|E|}(1-p_i)^{N-|E|}$, and rewrite $P(p)$ as\footnote{For simplicity of presentation in this section, we will work in the expanded representation of the decoding problem in which each fault has the same probability of occurring, see Section~\ref{sec:definitions}. For our simulations, we make a straightforward generalization to the compressed representation -- take $p_0,p_1,\dots,p_t$ to be a sequence of physical error rate scalars (e.g.~in the standard circuit noise model, the error rate of each circuit component), associate to each a vector of fault probabilities $\vec p_i$, and redefine $\pi_i(E)=\prod_{\alpha=0}^{N-1}\frac12(1+(-1)^{E_\alpha}(1-2[\vec p_i]_\alpha))$, where $\alpha$ indexes into the vector of faults $E$ and vector of fault probabilities $\vec p_i$. Note that if $\vec p_i$ is uniform with all elements equal to $p_i$, $\pi_i(E)$ reduces to the stated expression in the text.}:
\begin{equation}
\label{SM1}
P(p) = P(p_0) \prod_{j=1}^{t} \frac{P(p_{j})}{P(p_{j-1})}.
\end{equation}
The splitting method estimates $P(p)$ by computing each ratio $P(p_{j})/P(p_{j-1})$ separately.
It can be applied to general decoding problems; however, to ensure accurate and efficient estimation of $P(p)$, the following conditions should hold:
\begin{enumerate}
    \item \textbf{Accurate Initialization:} The estimate of $P(p_0)$ must be reliable, which can be ensured by sufficient Monte Carlo sampling (for large $p_0$) or, when available, by exact computation (for small $p_0$).
    \item \textbf{Sufficient Overlap:} Consecutive distributions $\pi_j(E|\mathcal{F}) = \pi_j(E)/P(p_j)$ and $\pi_{j-1}(E|\mathcal{F}) = \pi_{j-1}(E)/P(p_{j-1})$ must overlap significantly.
    This can be ensured by selecting a sufficiently dense sequence $p_0, p_1, \dots, p_t$.
    \item \textbf{Ergodicity:} Any two failing bitstrings $E, E' \in \mathcal{F}$ must be connected by a sequence of bitstrings in $\mathcal{F}$, each differing by a single bit. 
    This property has been proven in certain settings (e.g. for surface codes~\cite{bravyi2013simulation}).
    \item \textbf{Fast Mixing Time:}  
    The method relies on Metropolis sampling, which generates samples from the desired distribution if the number of steps greatly exceeds the Markov process's mixing time.  
    Slow mixing demands an impractically large number of samples for convergence.  
    The authors of Ref.~\cite{bravyi2013simulation} note that while mixing time may be modest for a single-qubit surface code patch, they conjecture it scales as $p^{-\Omega(d)}$ as $p \to 0$ for multi-qubit patches.  
    They argue geometrically that to connect low-weight failing bitstrings of inequivalent logical bitstrings under local Metropolis steps will require a transition through high-weight failing bitstrings.  
    Similarly, we may worry the splitting method to struggle to converge when applied to general codes with many logical qubits. In \sec{asymm_toric}, we provide an example in which poor convergence does occur if the Markov processes are initialized naively, and propose a possible solution in \sec{chain_init}.

\end{enumerate}

The splitting method (\alg{splitting-method}) comprises two key steps. 
The first, \emph{Metropolis sampling} (\alg{metropolis-step}), generates bitstrings from the conditional distribution $\pi_j(E|\mathcal{F})$. 
The second, \emph{ratio estimation}, computes the ratio $P(p_{j})/P(p_{j-1})$ using the sampled bitstrings.

\textbf{Metropolis Sampling: }  
To sample a bitstring $E \in \mathcal{F}$ from the conditional distribution  
$\pi(E|\mathcal{F}) = \pi(E)/P(p)$ for any $p$, we use a Metropolis-type subroutine.  
Each Metropolis step takes an input bitstring $E \in \mathcal{F}$ and proposes a new bitstring $E' \in \mathcal{F}$ by flipping a single bit, as described in \alg{metropolis-step}.  
Let $\pi(E \to E')$ be the probability of transitioning from $E$ to $E'$ in a single step. This process satisfies detailed balance: 
\begin{equation}
\label{DB}
\pi(E) \pi(E \to E') = \pi(E') \pi(E' \to E) \quad \forall E, E' \in \mathcal{F}.
\end{equation}
The Metropolis step defines a reversible Markov process with transition probabilities $\pi(E \to E')$, ensuring $\pi(E|\mathcal{F})$ as the steady-state distribution. We refer to the sequence of failing bitstrings $(E_0,E_1,\dots,E_T)$ sampled this way as a (Markov) chain.

\textbf{Ratio estimation:}  
Here we explain how in \alg{splitting-method} we estimate the ratio $P(p_{j})/P(p_{j-1})$ numerically given samples $E_1, E_2, \dots, E_T$ from $\pi_j(E|\mathcal{F})$ and $E'_1, E'_2, \dots, E'_{T'}$ from $\pi_{j-1}(E|\mathcal{F})$. 
Note we choose to exclude the starting configurations of the chains $E_0$ and $E'_0$ from this computation.
For any function $f: \mathcal{E} \to \mathbb{R}$, we define:
\begin{equation}
\label{SM2}
\langle f \rangle_j = \sum_{E\in \mathcal{F}} \pi_j(E|\mathcal{F}) f(E) = \frac{1}{P(p_j)} \sum_{E\in \mathcal{F}} \pi_j(E) f(E),
~~~\text{sample estimated as}~~~ \langle \hat{f} \rangle_j = \frac{1}{T} \sum_{\alpha=1}^{T} f(E_\alpha).
\end{equation}
It is straightforward to show that the ratio 
$P(p_{j})/P(p_{j-1}) = \langle \pi_{j} \rangle_{j-1}/\langle \pi_{j-1} \rangle_{j}$, which can be sample estimated as $\langle \hat{\pi}_{j} \rangle_{j-1}/\langle \hat{\pi}_{j-1} \rangle_{j}$.
A more accurate estimate from~\cite{bennett1976} uses $g(x) = 1/(1+x)$:
\begin{equation}
\label{eq:SM3}
r_j = \frac{P(p_{j})}{P(p_{j-1})} = c \cdot \frac{\langle g(c \pi_{j-1} / \pi_{j}) \rangle_{j-1}}{\langle g(c^{-1} \pi_{j} / \pi_{j-1}) \rangle_{j}}, ~~~\text{sample estimated as}~~~ c \cdot \frac{\frac{1}{T} \sum_{\beta=1}^{T'} g(c \pi_{j-1}(E'_\beta)/\pi_{j}(E'_\beta))}{\frac{1}{T'} \sum_{\alpha=1}^{T} g(c^{-1} \pi_{j}(E_\alpha)/\pi_{j-1}(E_\alpha))},
\end{equation}
which holds for any constant $c > 0$.
We are therefore free to chose a value of $c$.
Ref.~\cite{bennett1976} shows that statistical error is minimized when $c$ is equal to $r_j$.
We use an estimate of $r_j$ to select $c$ in \alg{splitting-method}.

\textbf{Estimating Precision:}  
In \alg{splitting-method}, we estimate the precision of ratio estimates to adjust the number of Metropolis samples. 
For samples from $\pi_{j}(E|\calF)$, we consider the numerator of $r_{j-1}$ and the denominator of $r_{j}$ in \eq{SM3}.  
First, we consider the standard error in the sample estimate of each of these terms.
Second, since the Metropolis \alg{metropolis-step} is a Markov process with an unknown mixing time $\tau$, reliable sampling requires $T \gg \tau$. 
To indicate sufficient mixing has occurred, we compare the full-sample estimate of these terms with the estimate obtained from just the first half of the samples (note that this indicator can be incorrect in some scenarios).
For both estimates, we set $c = 1$ for simplicity, as the optimized values of $c$ have not yet been determined when the precision estimates are computed in \alg{splitting-method}.
Since the samples for $\pi_{j}(E|\calF)$ for each of the $t$ values of $j$ are independent, we assume that errors combine incoherently, and as such we require that the total relative error for each $j$ is at most $\epsilon/\sqrt{t}$, where $\epsilon$ is the allowed relative error for the whole algorithm. In all our experiments, $\epsilon=0.25$, with the expectation that relative error of $25\%$ is acceptable when the goal is to estimate extremely small logical error rates, e.g.~below $10^{-7}$.

\begin{algorithm}
\caption{Splitting Method for Estimating $P(p)$}
\label{alg:splitting-method}
\begin{algorithmic}[1]
\State \textbf{Input:} Probability sequence $p_1, \dots, p_t = p$, initial chain length $T_\text{init}$, initial failure probability $P(p_1)$.
\State \textbf{Parameters:} Scaling factor $\lambda = 2$, target relative error $\epsilon=0.25$, function $g(x) = 1/(1+x)$.
\State \textbf{Output:} Estimated failure probability $\hat{P}(p)$ with relative error $< \epsilon$.

\State Initialize $\hat{P} \gets P(p_1)$.

\For{$j = 1$ to $t$}

    \rule{\linewidth}{0.4pt}  
    \Statex \textbf{Gather Metropolis samples until desired precision reached}  
    \State Initialize Markov chain $\mathcal{M}_j = [E_0]$, where $E_0$ is a known failing fault configuration.
    \State Set sequence length $T_j \gets T_\text{init}$.
    \State Initialize $\lambda'\gets\lambda$
    \While{$|\mathcal{M}_j| < T_j + 1$} 
        \State Extend $\mathcal{M}_j$ using Metropolis updates (\alg{metropolis-step} with $p_j$) until length $T_j+1$: $\mathcal{M}_j = [E_0, E_1, \dots, E_{T_j}]$.
        
        \State Compute sample estimates: $\langle \hat{g}_{\pm} \rangle_j = \frac{1}{T_j} \sum_{\alpha=1}^{T_j} g(\pi_j(E_\alpha) / \pi_{j \pm 1}(E_\alpha))$. \Comment{$\langle \hat{g}_{-} \rangle_j$ for ratio $r_{j-1}$, $\langle \hat{g}_{+} \rangle_j$ for $r_{j}$.}
        \State Compute sample variances: $s^2_{\pm} = \frac{1}{T_j-1} \sum_{\alpha=1}^{T_j} (g( \pi_j(E_\alpha) / \pi_{j \pm 1}(E_\alpha)) - \langle \hat{g}_{\pm} \rangle_j )^2$.
        \State Estimate statistical error: $\sigma = \max(s_{+}/\langle \hat{g}_{+} \rangle_j,s_{-}/\langle \hat{g}_{-} \rangle_j) / \sqrt{T_j}$. \Comment{Relative sample estimate uncertainty.}
        
        \State Compute sample subset estimates: $\langle \hat{g}'_{\pm} \rangle_j = \frac{1}{\lceil T_j/2 \rceil} \sum_{\alpha=1}^{\lceil T_j/2 \rceil} g(\pi_j(E_\alpha) / \pi_{j \pm 1}(E_\alpha))$. \Comment{First half of samples.}
        \State Estimate mixing discrepancy $\Delta = \max\left( \frac{|\langle \hat{g}_{+} \rangle_j - \langle \hat{g}'_{+} \rangle_j|}{\langle \hat{g}_{+} \rangle_j}~,~\frac{|\langle \hat{g}_{-} \rangle_j - \langle \hat{g}'_{-} \rangle_j|}{\langle \hat{g}_{-} \rangle_j} \right)$. \Comment{Small if $T_j \gg$ mixing time.}
        \State \textbf{if} $(\sigma + \Delta) > \epsilon/\sqrt{t}$ \textbf{then} $T_j \gets T_j + \lceil \lambda' \cdot T_\text{init} \rceil$; $\lambda'\gets\lambda\cdot\lambda'$\Comment{Increase $T_j$ if relative error is too large.}

    \EndWhile

    \rule{\linewidth}{0.4pt}  
    \Statex \textbf{Estimate ratio using metropolis samples from current round and previous round}  
    \If{$j \geq 2$}  \Comment{Previous round's samples: $\mathcal{M}_{j-1} = [E_0, E'_1, \dots, E'_{T_{j-1}}]$.}
    \State Initialize $c \gets 1$.  
    \For{$i \in \{1, 2, 3\}$} \Comment{Converge on $c$ minimizing error in $\hat{r}_{j-1}$; a few iterations should suffice.}

    \State Compute ratio estimate:~~~$
    \hat{r}_{j-1} = c \cdot \left[ \sum_{\alpha=1}^{T_{j-1}} g \left( \frac{c \pi_{j-1}(E'_\alpha)}{\pi_{j}(E'_\alpha)} \right) \right] / \left[ 
    \sum_{\alpha=1}^{T_j} g \left( \frac{\pi_j(E_\alpha)}{c \pi_{j-1}(E_\alpha)} \right) \right]$.

    \State $c \gets \hat{r}_{j-1}$.
    \EndFor

    \State Update failure probability estimate: $\hat{P} \gets \hat{P} \cdot \hat{r}_{j-1}$.
    \EndIf
\EndFor
\State \Return $\hat{P}$.
\end{algorithmic}
\end{algorithm}

\begin{algorithm}
\caption{Metropolis Step for Sampling from $\pi(E | \mathcal{F}) := p^{|E|}(1-p)^{N-p}/P(p)$}
\label{alg:metropolis-step}
\begin{algorithmic}[1]
\State \textbf{Input:} Current failing bitstring $E \in \mathcal{F}$, probability $p$.
\State \textbf{Output:} Updated failing bitstring $E' \in \mathcal{F}$.

\State Flip a random bit in $E$ to form a trial bitstring $E'$.
\State Compute $q = \min\bigl(1, \tfrac{\pi(E')}{\pi(E)}\bigr)$
\State With probability $1-q$, \Return $E$ \Comment{To ensure detailed balance, stay at $E$}
\If{$E' \in \mathcal{F}$} \Comment{Ensure failing bitstring by testing with decoder}
    \State \Return $E'$.
\EndIf
\State \Return $E$.
\end{algorithmic}
\end{algorithm}

\textbf{Selecting splitting probabilities: }
\alg{splitting-method} requires a list of splitting probabilities as an input.
The following heuristic choice\footnote{To motivate this choice we note that the $w_j$ provides a rough estimate of  the average weight $|E|$ of a random bitstring $E\in \calF$ drawn from the distribution
$\pi_j(E|\calF)$. 
Since the probability of each bit is independent and $p_j\ll 1$, one should expect that the random variable $|E|$ is concentrated near its mean with standard deviation $O(\sqrt{w_j})$.
Therefore one can use a bound
\[
\left(\frac{p_{j+1}}{p_j}\right)^{-O(\sqrt{w_j})} \le \frac{P_{j+1}(E|\calF)}{\pi_j(E|\calF)} \le
\left(\frac{p_{j+1}}{p_j}\right)^{O(\sqrt{w_j})}
\]
for all `typical' bitstrings $E$. 
}
of the splitting sequence was proposed as a reasonable tradeoff between the statistical error and the number of splitting steps in Ref.~\cite{bravyi2013simulation}:
\begin{equation}
\label{steps1}
p_{j\pm1}=p_j 2^{\mp1/\sqrt{w_j}}, \quad w_j=\max{(D/2,p_j N)}.
\end{equation}
One can start with $p_0$ and use Eq.~\eqref{steps1} to define smaller $p_0>p_1>p_2>\dots$ and larger $p_0<p_{-1}<p_{-2}<\dots$ physical error rates. The ``downward" splitting sequence $p_0,p_1,p_2,\dots$ can be analyzed as described above. The ``upward" splitting sequence $p_0,p_{-1},p_{-2},\dots$ can be analyzed similarly (simply relabel the negative indices with their absolute values and the same description above applies). Although logical error rates at $p_{-1},p_{-2},\dots$ are likely to be easy to estimate with Monte Carlo if that is the case at $p_0$, we include some upward splitting just to compare the results against Monte Carlo. We note however that some deviation may be expected as the (pseudo)threshold of the QEC system is approached \cite{bravyi2013simulation}.

\subsection{Example: asymmetric, unrotated toric code}\label{sec:asymm_toric}
We use asymmetric toric codes with $Z$ (phase-flip) noise to demonstrate the difficulty of achieving a fast mixing time for a potential pitfall of the splitting method applied to general QEC codes -- the difficulty of mixing between different logical sectors.

In an asymmetric toric code, denoted UT($d_1,d_2$) with $d_1<d_2$, there is a basis of four logical Pauli operators $Z_S$, $X_L$, $Z_L$, $X_S$, in which ``short" operators $Z_S$ and $X_S$ are weight $d_1$ and ``long" operators $Z_L$ and $X_L$ are weight $d_2$. To study $Z$ noise, we begin Markov chains using either logical operator $Z_S$ or $Z_L$ as the initial failing configurations. If the chains are continued long enough to reach the mixing time, we would expect the same logical error estimates from either initialization. If the mixing time is not reached, however, then chains initialized with $Z_L$ will tend to linger on atypically high-weight failing configurations, and are therefore likely to predict a lower error rate compared both to chains initialized with $Z_S$ and to Monte Carlo sampling. This situation is indeed what is observed in our simulations, see \fig{asymm_toric}.

\begin{figure}[t]
(a)\includegraphics[width=0.46\linewidth]{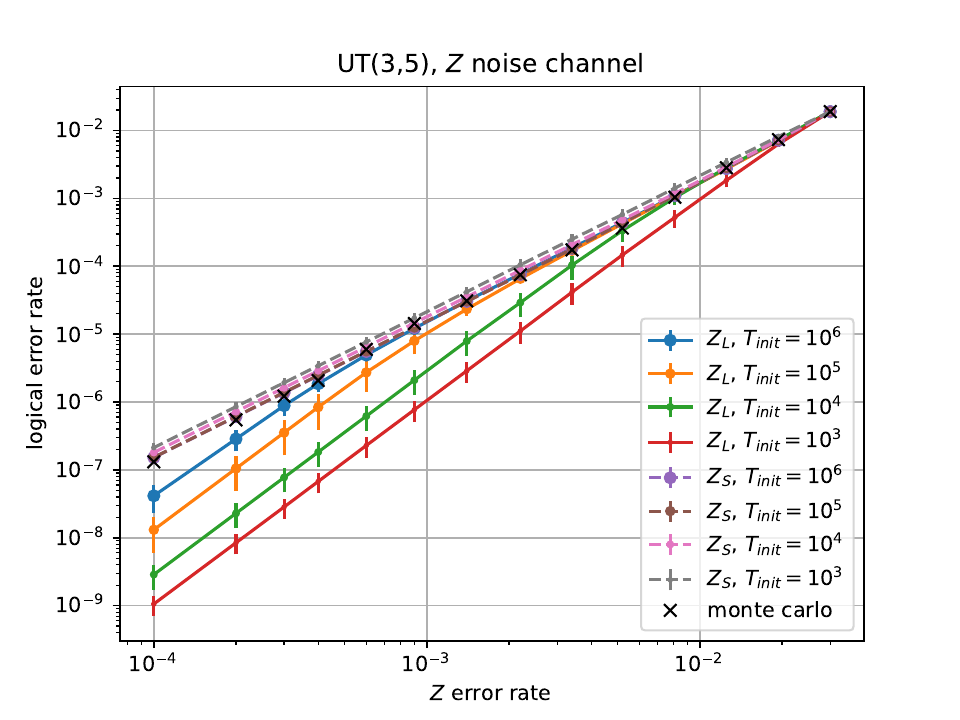}
(b)\includegraphics[width=0.46\linewidth]{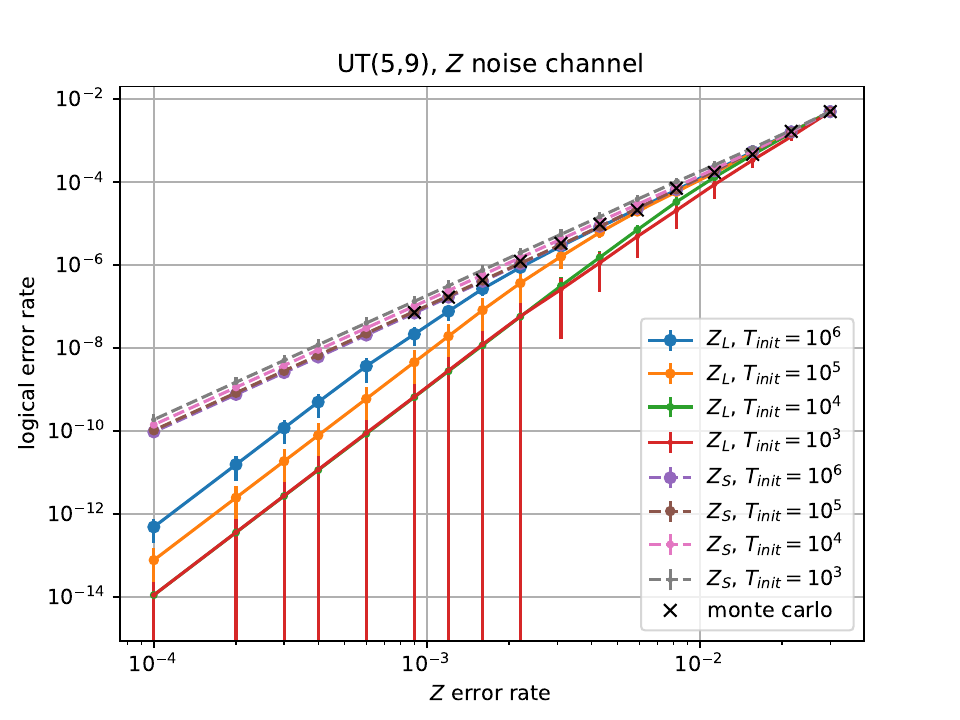}
\caption{\label{fig:asymm_toric}The splitting method used to estimate logical error rates of asymmetric toric codes through a $Z$ noise channel decoded by matching. We initialize all Markov chains with either the short logical $Z_S$ or the long logical $Z_L$ and note a significant effect on the splitting method estimates. In these examples, the two logicals $Z_S$ and $Z_L$ have weights (a) $3$ and $5$, (b) $5$ and $9$. Initializing with the short logical $Z_S$ tends to match the Monte Carlo estimates closely. However, as the initial chain length $T_{\text{init}}$ is increased, estimates initialized with the long logical $Z_L$ converge to the Monte Carlo points as well. We run the splitting method 10 times for each colored data set, averaging the results to plot the points shown and taking the standard deviation to obtain error bars. Data points are joined by lines just to better visualize the trends. Only downward splitting is done, and $p_0=0.03$. Monte carlo samples include 100 logical fails per point (and 1000 fails at $p_0$) and so have negligible error.}
\end{figure}

\subsection{Multi-seeded splitting and bivariate bicycle code results}\label{sec:chain_init}

For the asymmetric toric code, it appears obvious in hindsight that we should just initialize chains with the lowest weight logical operator to get good estimates of logical error rate at low physical error rate. 
However, it is not clear that this initialization strategy should be generalized to other QEC systems, particularly those modeling circuit noise. 
In general, we want to explore all logical sectors, and ideally encourage chains to hop between them, since we cannot necessarily expect one sector to dominate.
For instance, at finite physical error rates, many logical operators of high weight may be more important to the logical error rate than a relatively small set of minimum weight logicals, e.g.~see \fig{failure-spectrum-transform}(b).

To address this issue, we propose the following strategy that we feel is a natural choice for chain initialization. The main idea is to initialize chains at $p_0$ with Monte Carlo sampled failing configurations, and then initialize chains at physical error rates $p_j$, $j>0$, with the last failing configurations in chains at $p_{j-1}$.

\begin{minipage}{0.95\textwidth}
\noindent\textbf{Multi-seeded splitting} (with parameters $L$ and $M$):
\newline

\noindent \text{\space}Repeat all the following $L$ times (in parallel).
\begin{enumerate}[label=(\roman*)]
\item \textit{Sample.} At physical error rate $p_0$, use Monte Carlo to sample a failing configuration $E_0$. 
\item Repeat the following steps $M$ times (with different random seeds and in parallel)
\begin{enumerate}
\item \textit{Zeroth chain.} Initialize the Markov chain at $p_0$ with $E_0$ and complete the chain to obtain some final failing configuration $E^{(0)}_{T_0}$. 
\item \textit{Downward chains.} For $i=1,2,\dots,I_+$ in sequence, initialize the Markov chain at $p_i$ with failing configuration $E^{(i-1)}_{T_{i-1}}$ and complete the chain. Call the chain's final failing configuration $E^{(i)}_{T_{i}}$. May be performed in parallel with step (c).
\item \textit{Upward chains.} For $i=-1,-2,\dots,I_-$, initialize the chain at $p_i$ with $E^{(i+1)}_{T_{i+1}}$ and complete the chain. Call the chain's final failing configuration $E^{(i)}_{T_{i}}$. May be performed in parallel with step (b).
\end{enumerate}
\end{enumerate}
\end{minipage}

The result of multi-seeded splitting is $LM$ sets of chains with indices $I_-$ through $I_+$, which can then be analyzed by the splitting method, \sec{splitting}, to get $LM$ estimates of logical error rates. In plots, we average the $LM$ estimates from multi-seeded splitting and find their standard deviation to draw as error bars. This means we ultimately end up averaging estimates derived from $L$ different Monte Carlo sampled failing configurations at $p_0$, which may also help mitigate the effect of insufficient mixing time.

We make a few additional remarks. First, at physical error rate $p_0$, we perform Monte Carlo anyway to obtain an estimate of the logical error rate $P(p_0)$ as required by the splitting method. Thus, finding the required failing configuration in step (i) of multi-seeded splitting is easily doable and moreover configuration $E_0$ should be a typical failing configuration. Second, chains within steps (iib) and (iic) must be generated sequentially because we require the final failing configuration from the previous chain to initialize the next. However, steps (iib) and (iic) can be done in parallel as they both require just $E^{(0)}_{T_0}$ from step (iia) to get started.

\fig{mc+s_example} shows the results of multi-seeded splitting applied to the asymmetric, unrotated toric codes with $L=26$ and $M=3$, showing good agreement with the available Monte Carlo values.

\begin{figure}[t]
(a)\includegraphics[width=0.49\linewidth]{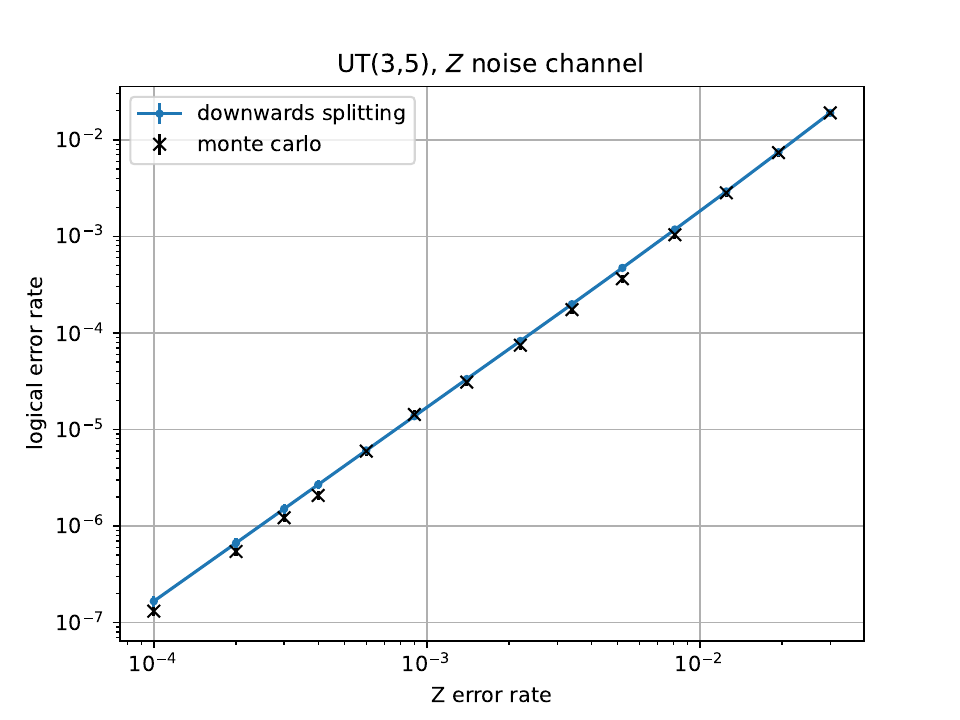}
(b)\includegraphics[width=0.49\linewidth]{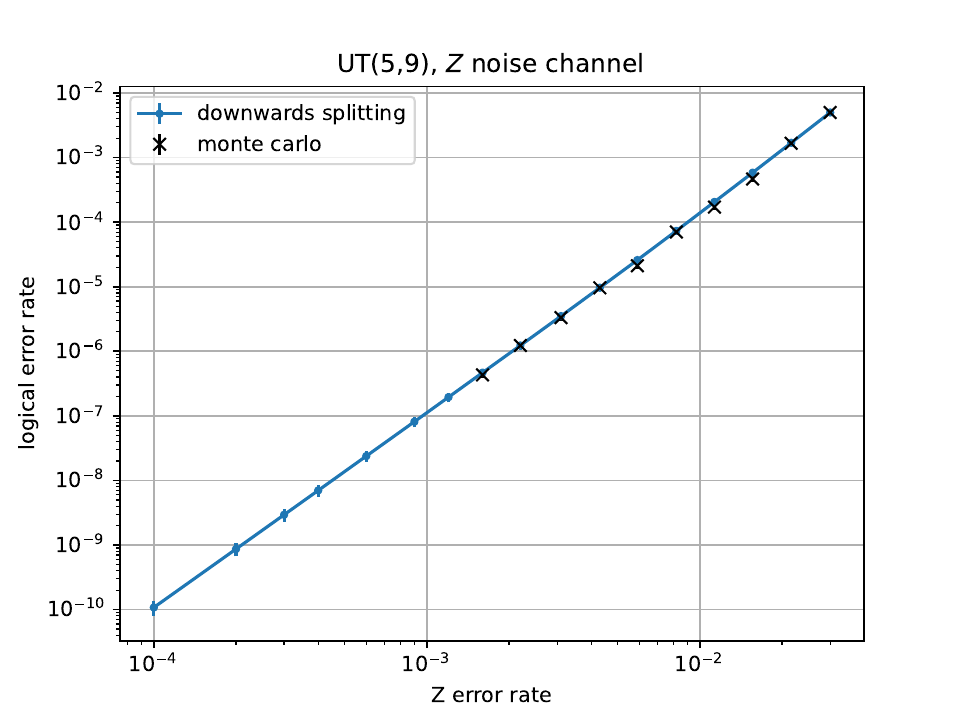}
(c)\includegraphics[width=0.46\linewidth]{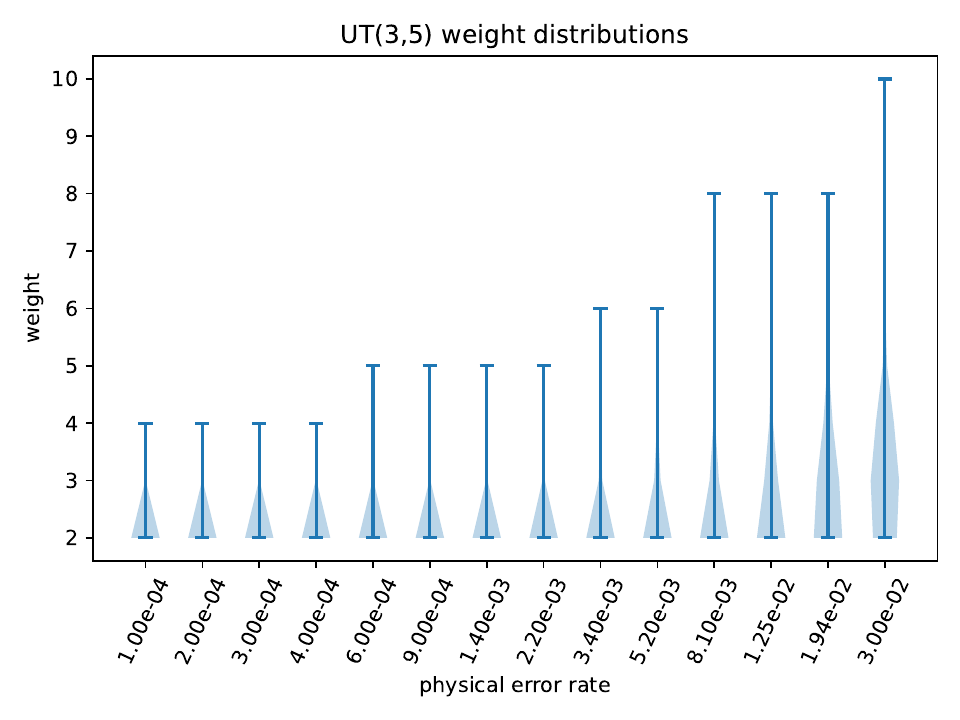}\hspace{15pt}
(d)\includegraphics[width=0.46\linewidth]{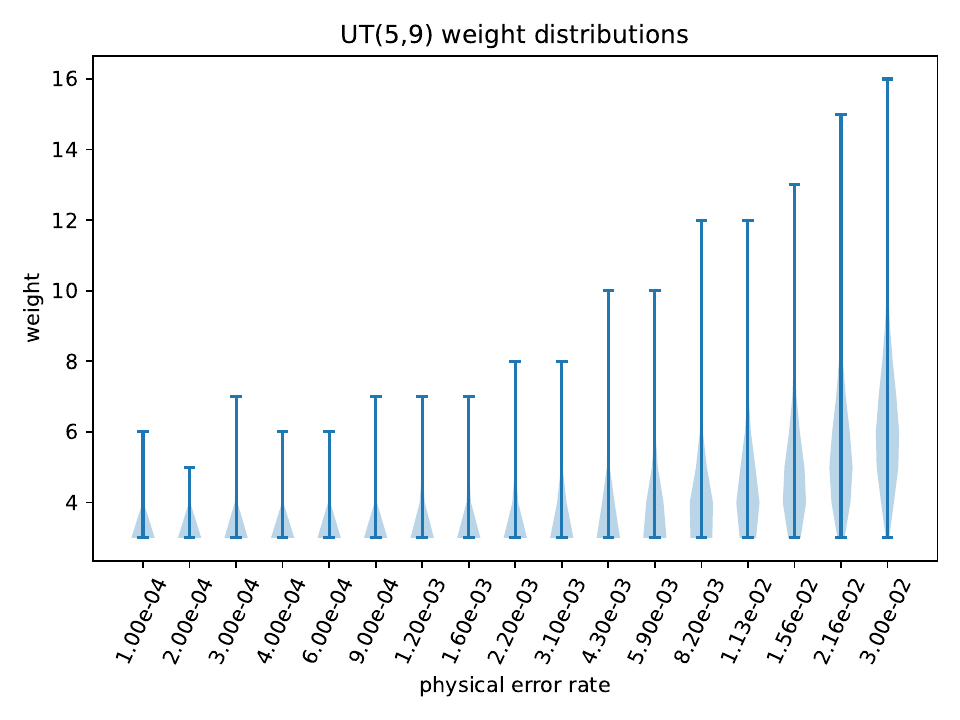}
\caption{\label{fig:mc+s_example} (a,b) The multi-seeded splitting method described in this section (with $p_0=0.03$, so only downward splitting is done) provides logical error rate estimates for asymmetric toric codes that are consistent with the Monte Carlo values. For the data shown $T_{\text{init}}=10^5$, but it is observed that the main effect of changing this parameter is to change the size of the error bars. (c,d) We show the distributions of failing configuration weights in all the $LM$ chains at each physical error rate. As expected, the average weight of a failing configuration gradually reduces as the physical error rate is decreased.}
\end{figure}

We also employ multi-seeded splitting to estimate logical error rates of BB($d$) codes with Relay decoding. The results are shown in \fig{splitting_bb}. Due to, it seems, the large amount of compute time required to obtain reasonable chain mixing, we do not include results for BB($18$).\footnote{Initial, quicker studies of BB($18$) did not yield sensible results, though these studies used relatively small values for $T_{\text{init}}$ (from $6\times10^5$ and smaller) and chain initialization on low- or min-weight circuit logicals instead of Monte Carlo sampled failing configurations.} In Figures \ref{fig:final-results-decoder-comparison} and \ref{fig:final-results-code-comparison}, we show data from the splitting method for other QEC systems, the BB codes with the BPLSD decoder and RS codes with matching, and the parameters for these simulations can be found in Table~\ref{tab:splitting_method_parameters} in Appendix~\ref{app:additional_splitting}. Also in Appendix~\ref{app:additional_splitting}, we provide timing data, decoder call counts, and other additional information about the splitting method simulations.

\begin{figure}[th]
(a)\includegraphics[width=0.49\linewidth]{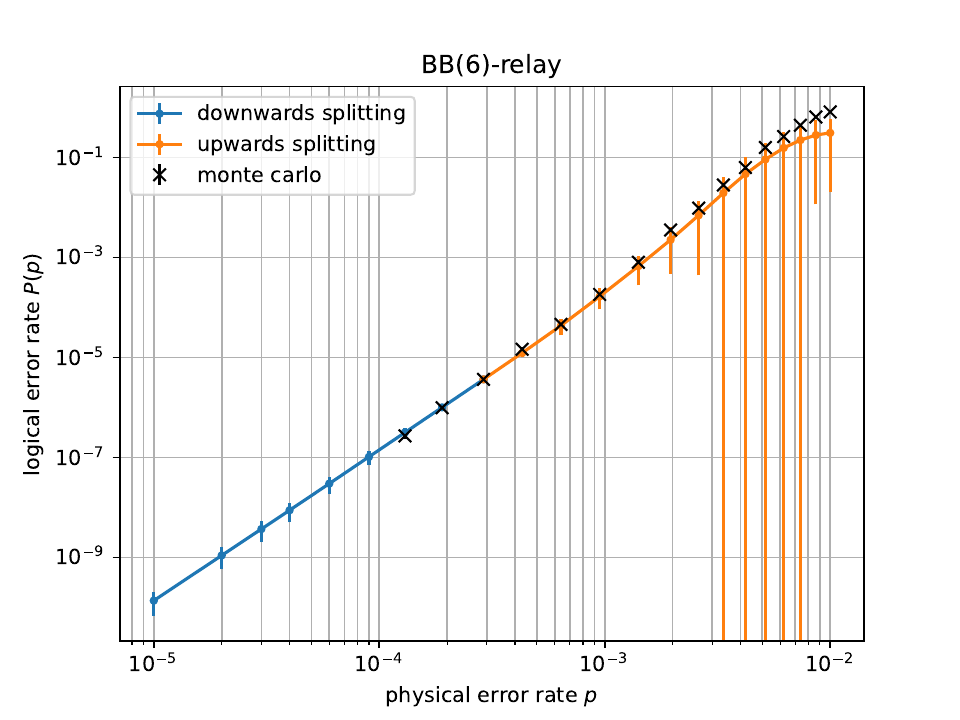}
(b)\includegraphics[width=0.49\linewidth]{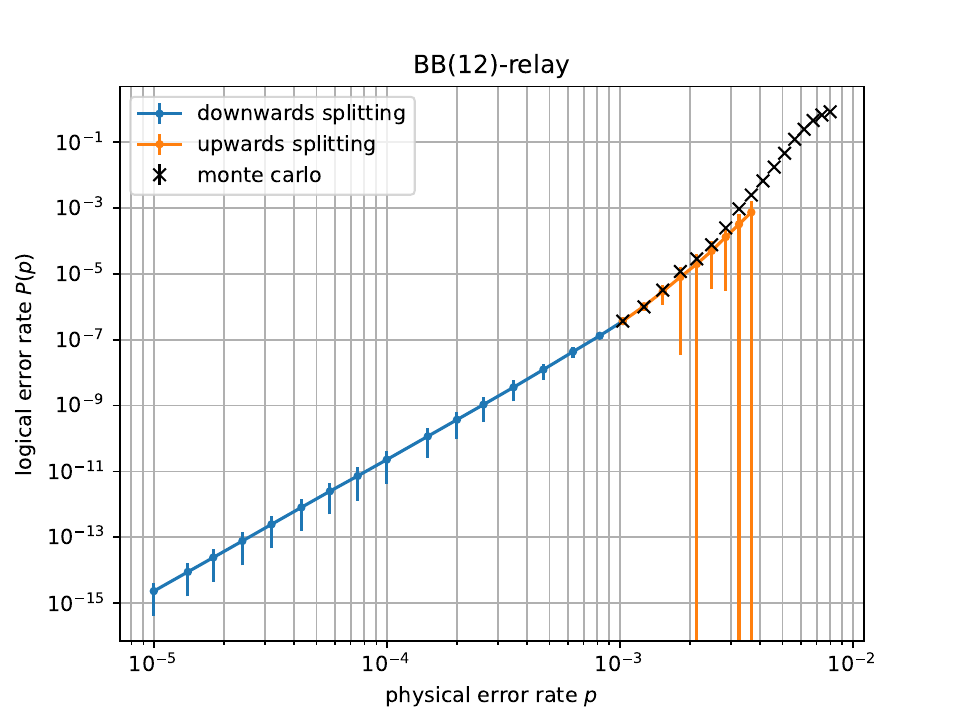}
(c)\includegraphics[width=0.46\linewidth]{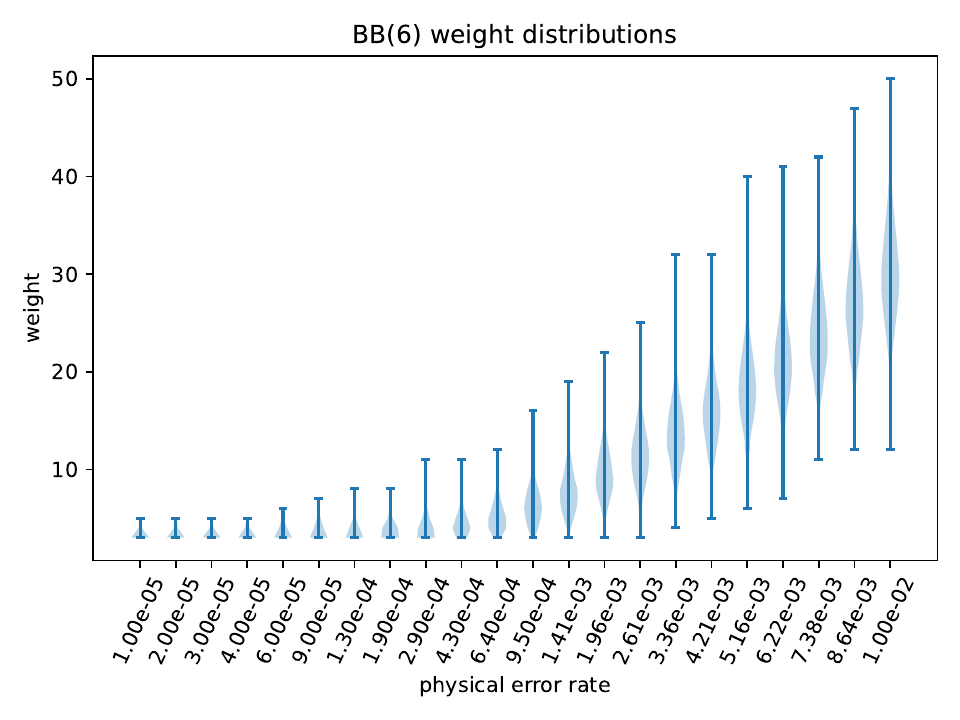}\hspace{15pt}
(d)\includegraphics[width=0.46\linewidth]{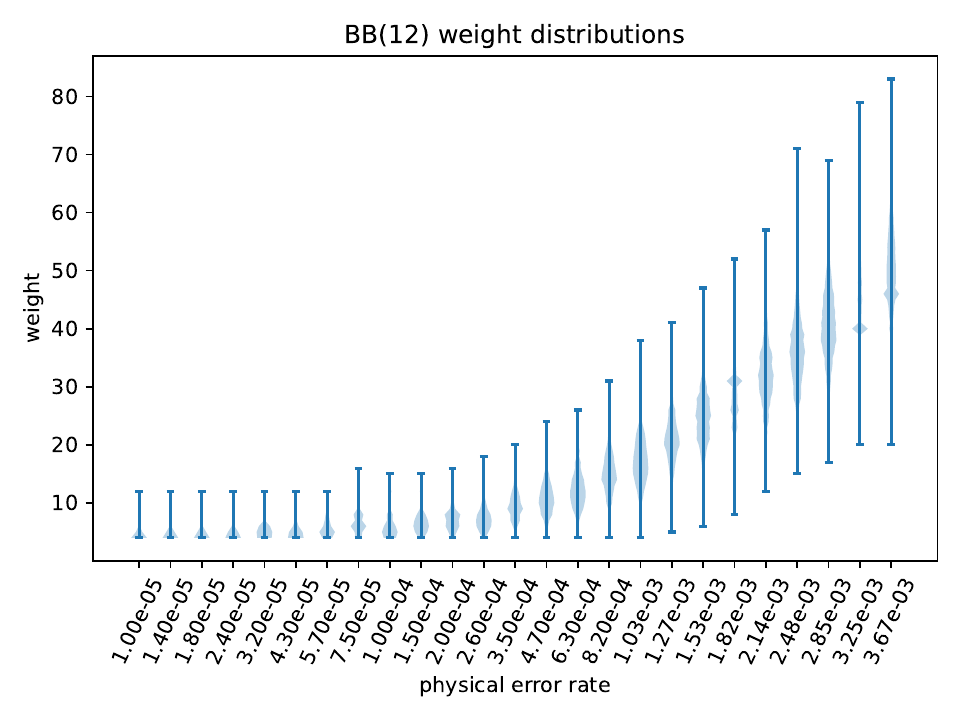}
\caption{(a,b) Using multi-seeded splitting to estimate $Z$-type logical error rates of bivariate bicycle codes BB($d$) for $d=6,12$ against circuit noise using the Relay decoder \cite{RelayPaper}. We use $T_{\text{init}}=10^5$ for BB($6$) and $T_{\text{init}}=10^6$ for BB($12$), and $M=3$ and $L=12$. The plotted points are from the average of the $LM=36$ instances of multi-seeded splitting and error bars are from their standard deviation. Data points are joined by lines just to better visualize the trends. Note that for BB($6$) we did not choose the starting point $p_0$ to be the lowest point with Monte Carlo data, so that we could observe the downward splitting agreeing with Monte Carlo. (c,d) We plot the distributions of weights of failing configurations visited in the chains at each physical error rate. We do note large error bars in the upward splitting and some choppiness in the BB($12$) weight distributions, both of which may indicate some struggle to converge or fully mix the chains. Also for the BB($12$) code, we note that 2 of the 36 instances of multi-seeded splitting step (iic), upward splitting, did not finish in a reasonable time and only chains that did complete are included in data for plots (b,d). All 36 instances of step (iib), downward splitting, were completed, however.}
\label{fig:splitting_bb}
\end{figure}

We end this section by remarking on some limitations we still see in applying the splitting method to general QEC systems. First, it is still unclear how one should choose $T_{\text{init}}$ to ensure sufficient chain mixing, or even how one might determine more rigorously that a chain has mixed sufficiently. Techniques addressing both these limitations can be found in the extensive literature on Metropolis sampling methods. To encourage faster mixing, parallel tempering (e.g.~\cite{neal1996sampling}) allows hopping between chains at different temperatures (or, in our case, different physical error rates), and, to evaluate convergence, it can be helpful to compare multiple independent chains \cite{gelman1992inference}. However, we leave application of these two techniques to rare events in QEC systems to future work.
A simple method that we did try was to require, in addition to the half-sample convergence check we already do, that the chains transition (to a different failing configuration) some minimum number of times before terminating. In preliminary studies, this tends to make chains at low physical error rates longer (because transition rates are naturally low there, see \app{additional_splitting}), but it was not clear this provided more accurate logical error estimates. 

Second, and relatedly, it appears important to improve the proposal of new configurations so that the Markov chains transition more frequently and more quickly visit new failing configurations. In the presented simulations, the only transitions we allow are single bit flips (i.e.~the introduction of one new fault or removal of one existing fault). A tempting possibility is to use symmetries in the QEC system to take a current failing configuration and propose a new fault configuration that is guaranteed, or at least very likely, to also be failing. For instance, this might be done by exploiting space or time symmetries in the QEC system. This is similar to the symmetry method used for computing logicals in \sec{logicals}, though now the symmetries are of failing configurations instead, which necessarily means the decoder's behavior is involved.

\clearpage
\section{Application: Bivariate bicycle codes at low logical error rates}
\label{sec:validate-relay}

In this section we use our three techniques of failure-spectrum ansatz fitting, min-weight analysis and multi-seeded splitting, to analyze bivariate bicycle codes under circuit noise across a wide range of error rates.

We apply all three techniques, along with Monte Carlo sampling to bivariate bicycle codes decoded with Relay and BP-LSD in \fig{final-results-decoder-comparison}, and the surface code decoded with matching in \fig{final-results-code-comparison}.
Our three analysis techniques show consistent agreement with each other and with Monte Carlo across all datasets and accessible error rates.
Multi-seeded splitting, implemented as described in \sec{chain_init}, yields large inferred uncertainties at low error rates for the bivariate bicycle codes. 
Nevertheless, its estimated values remain consistently close to the ansatz curve, even when the nominal uncertainty is substantial.
This consistency helps to validate all techniques. 

In \fig{final-results-decoder-comparison}, we evaluate Relay and BP-LSD using the settings described in \sec{system-examples}. 
For each decoder, we fit the five-parameter ansatz $f^{(5)}_{\text{ansatz}}(w)$ by $\chi^2$ minimization to the sampled failure spectrum, fixing $w_0$ to the smallest failure weight found by multi-seeded splitting for BB(6) and BB(12), and setting $w_0 = 9$ for BB(18).
The onset bound is obtained from our min-weight bound on $|\mathcal{F}(w_0)| / \binom{N}{w_0}$ from \tab{min-weight-properties}. 
Two observations emerge from the data.
First, Relay substantially outperforms BP-LSD, consistent with our implementation of BP-LSD being a lightweight variant using order zero. 
Second, for BB(12) and BB(18), the Relay failure spectrum trends toward an onset well above our bound, implying either that the bound is far from tight or that these codes have substantial headroom for improved performance through enhanced decoding.
Given our extensive methods for identifying minimum-weight logical operators and validating high fractional coverage in \sec{computing-min-weight-properties}, together with the discovery of numerous failing errors of weight below $D/2$, we conjecture that the latter explanation is more likely.

\begin{figure}[ht]
    \centering
    \includegraphics[width=0.49\linewidth]{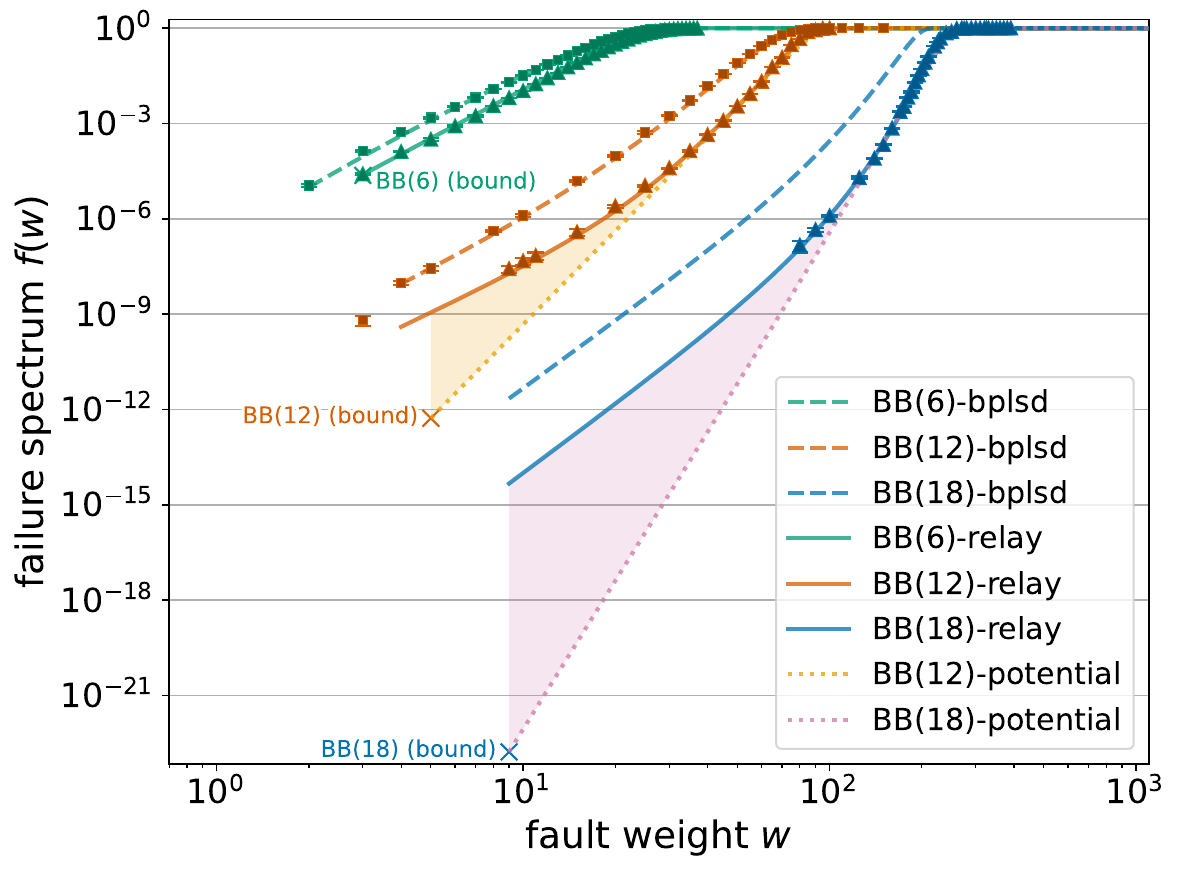}
    \includegraphics[width=0.49\linewidth]{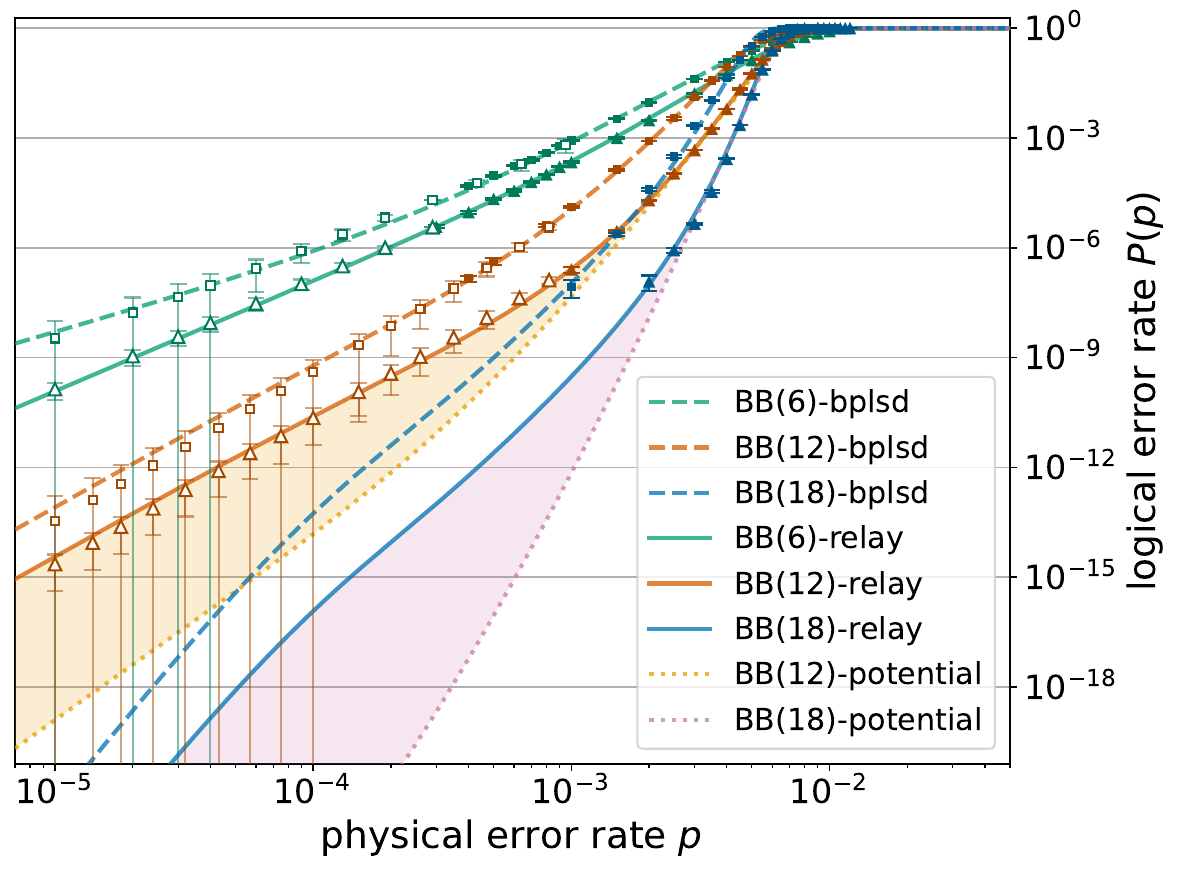}
    \caption{
    \textbf{Decoding bivariate bicycle codes with Relay versus BP-LSD.}
    Sampled data for both BP-LSD (squares) and Relay (triangles), displaying the failure spectrum (left) and logical error rate (right) for three bivariate bicycle codes under circuit noise. 
    For each decoder, we fit the five-parameter ansatz $f^{(5)}_{\text{ansatz}}(w)$ to the failure-spectrum data, both freely and with the onset fixed to the bound on the optimal onset. 
    We include multi-seeded splitting data for BB(6) and BB(12) (hollow markers), which aligns with the ansatz extrapolations.
    Relay clearly outperforms BP-LSD, yet its inferred onset for BB(12) and BB(18) lies well above the bound, indicating potential for substantial decoder improvement.
    Note that in the case of BB(12)-bplsd, we fit the ansatz to data for $w \geq 4$, and include the $w=3$ term separately using the generalization described toward the end of \sec{model-ansatz}.
    }
    \label{fig:final-results-decoder-comparison}
\end{figure}

In \fig{final-results-code-comparison}, we analyze the rotated surface code decoded using matching. 
As in \fig{final-results-decoder-comparison}, we fit the five-parameter ansatz $f^{(5)}_{\text{ansatz}}(w)$ by $\chi^2$ minimization to the sampled failure spectrum although in this case we fix $w_0 = D/2 = d/2$, since matching is a provable minimum-weight decoder for surface codes. 
The observed failure spectrum trends toward an onset that appears likely to approach, but remain slightly above, our bound on the optimal onset (as expected since matching is a min-weight but not a max-class decoder).
Our extrapolation of the low-error-rate behavior of the two largest bivariate bicycle codes decoded with Relay suggests that their logical error rates will intersect that of the surface code below $p = 10^{-3}$ but above $p = 10^{-4}$.
However, the projection based on the theoretical onset bound indicates that with an improved decoder, bivariate bicycle codes may outperform the surface code by several orders of magnitude in this range.

\begin{figure}[ht]
    \centering
    \includegraphics[width=0.49\linewidth]{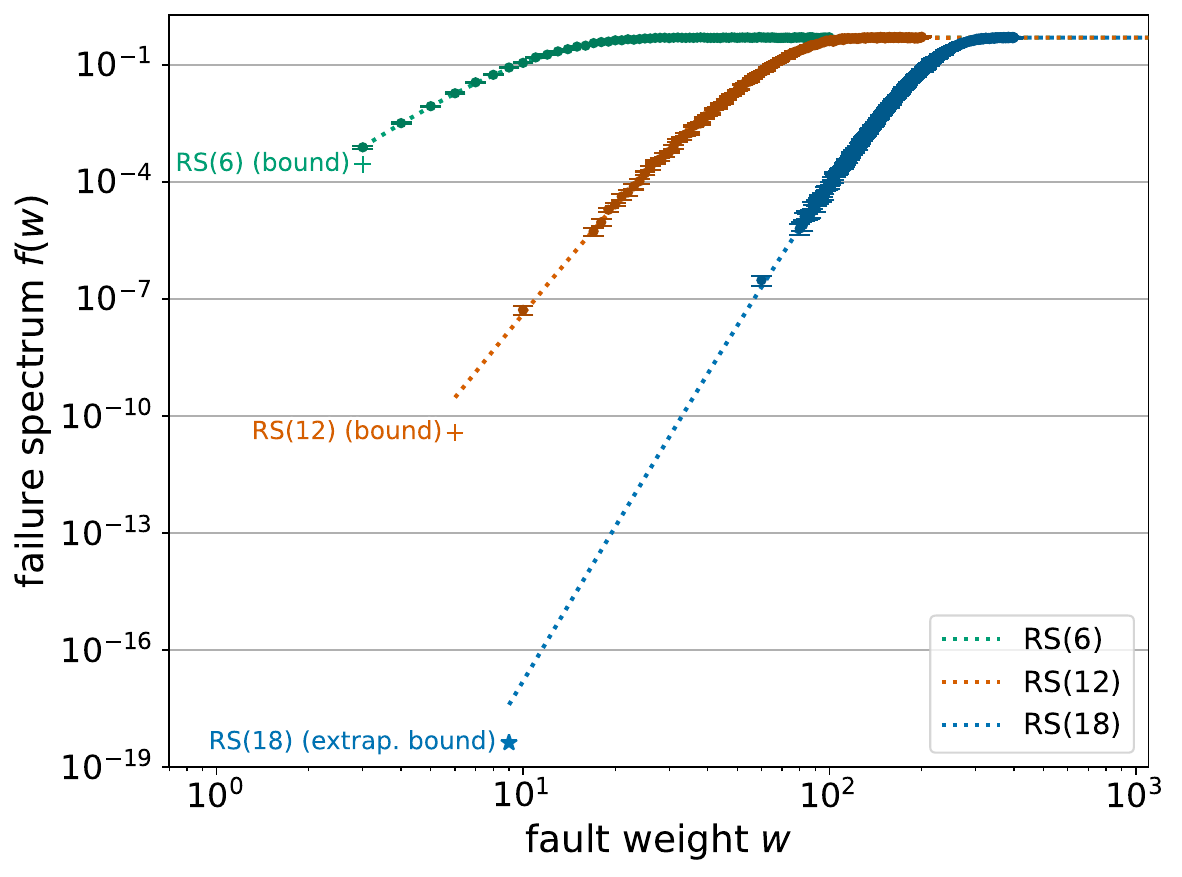}
    \includegraphics[width=0.49\linewidth]{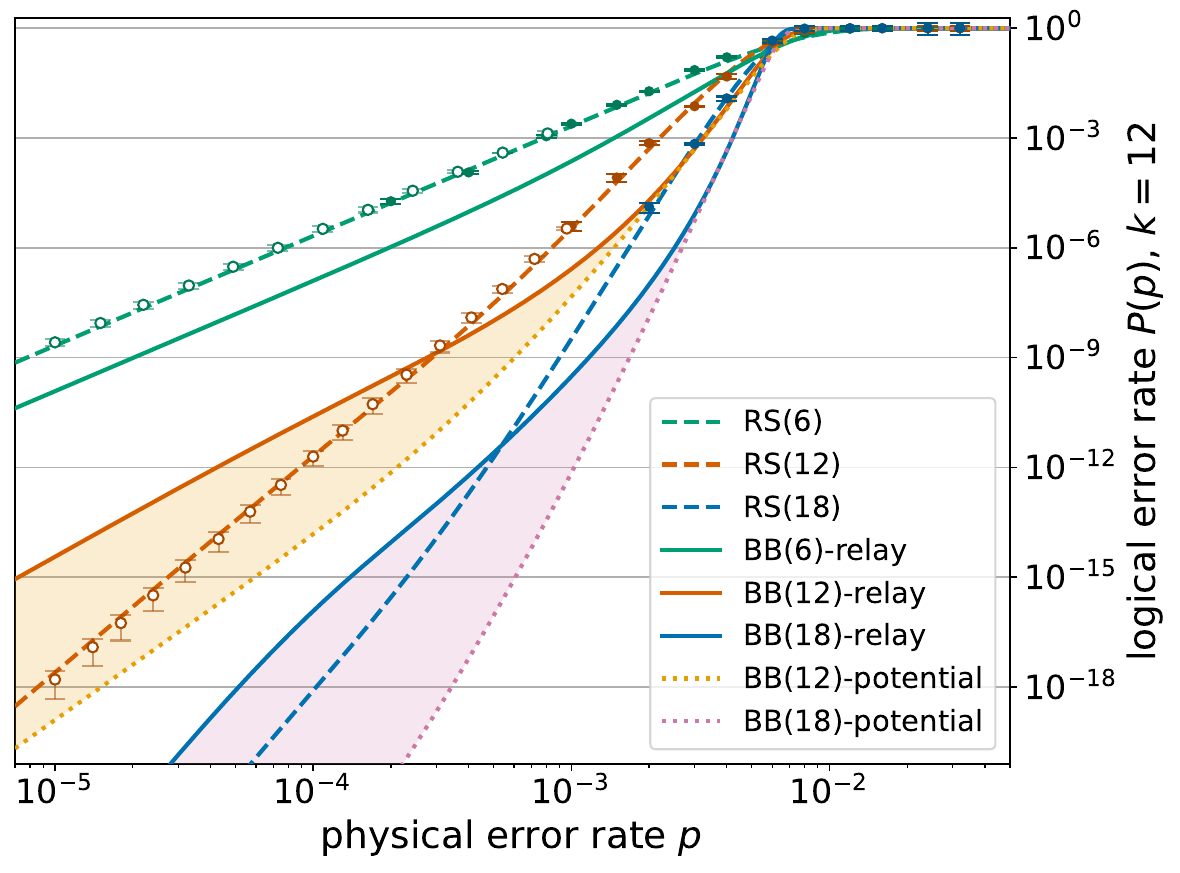}
    \caption{
    \textbf{Surface codes versus bivariate bicycle codes.}
    Sampled failure-spectrum (left) and logical-error-rate (right) data for rotated surface codes under circuit noise, with fits of $f^{(5)}_{\text{ansatz}}(w)$ and $w_0 = D/2 = d/2$. 
    We include multi-seeded splitting data for RS(6) and RS(12) (hollow markers), which aligns with the ansatz extrapolations.
    Relay and potential ansatz fits for bivariate bicycle codes from \fig{final-results-decoder-comparison} are overlaid for comparison. 
    The surface-code data trend toward an onset slightly above the bound, consistent with matching being min-weight but not max-class. 
    Extrapolations suggest that BB(12) and BB(18) decoded with Relay will intersect that of the surface code at some point in the range $10^{-4} < p < 10^{-3}$, but could outperform surface codes by orders of magnitude in this regime with improved decoding.
    }
    \label{fig:final-results-code-comparison}
\end{figure}

\textbf{Concluding remarks.}
For several QEC systems we have shown that our three complementary techniques: the failure-spectrum ansatz, min-weight analysis, and multi-seeded splitting provide descriptions which are consistent with each other and with Monte Carlo across regimes. 
The close agreement between these techniques indicates that together they offer a reasonably reliable characterization of these QEC systems.

To study a new QEC system, if resources allow, we recommend applying all methods and checking for consistency. 
Each technique has weaknesses: the ansatz may not capture the true behavior when extrapolated too far beyond the data, the onset bounds may not be tight if not all min-weight logicals were found, and the Markov chains in multi-seeded splitting may fail to mix. 
Agreement between methods gives confidence.
A gap between the optimal onset bound and the onset which is observed or projected from the ansatz can indicate opportunities to improve the QEC system's decoder. 

Future improvements should aim to make these techniques faster and more accurate. 
This includes gathering larger datasets using faster decoders on large compute clusters for improved ansatz fits, developing faster methods for identifying min-weight logicals for improved bounds, and modifying the transition function and incorporating parallel tempering to improve multi-seeded splitting. 

Our techniques apply to noise models in which faults occur with arbitrary relative probabilities scaled by a global parameter $p$, though several challenges remain.
If the relative probabilities change, the onset bounds can be updated directly without re-identifying min-weight logicals; however, failure-spectrum data must be resampled, and splitting chains rerun.
Extensions could adapt these techniques more efficiently to noise models governed by multiple global parameters.
Even more ambitiously, one could seek generalizations of these techniques which can handle adaptive circuits.

\section*{Acknowledgments}
We would like to thank Sergey Bravyi and Giacomo Fregona for discussions regarding the splitting method and its application to BB codes.
We would like to thank Patrick Rall for discussions, assistance with numerical implementation of importance sampling and providing an initial Lagrange multiplier method example which led to the analysis in \app{LMmethod}.
We thank Tristan M\"{u}ller for sharing the parameterizations of Relay-BP used in Ref.~\cite{RelayPaper} and advice for running the decoder.
We thank Lev Bishop for interesting discussions on various aspects of this work.

\clearpage
\appendix
\section{Appendix}
\subsection{Correlated errors in separate X and Z decoding}
\label{app:XYZ-vs-XZ}

In most of this paper, we use CSS decoding and decode X and Z components separately.
Moreover, we treat the generation of noise of X and Z components  separately too, declaring failure for the X system if a logical X error occurs ignoring the Z system. 
In this appendix we relax that approximation and analyze the impact of correlations of the noise when the X and Z systems are decoded separately.

A naïve estimate of the overall failure rate would be $\hat{P}_X(p)+\hat{P}_Z(p)$, where $\hat{P}_X$ and $\hat{P}_Z$ are computed for $X$ and $Z$ decoding in isolation. 
This ignores correlations from $Y$-type errors that correlate errors that arise in the $X$ and $Z$ components. 
A consistent procedure is to sample the joint noise (including $X$–$Z$ correlations), extract the $X$ and $Z$ components, decode them independently, and declare failure if either has a logical error.
These considerations also arise when considering the failure spectrum, and for the optimal onset bounds. 

 
To handle correlated noise, we sample a full error vector $b\in\{0,1\}^N$ on the check matrix $H$: either i.i.d.\ with per-fault probability $p/15$ (for $P(p)$) or uniformly among weight-$w$ vectors (for $f(w)$). 
We then extract $b_X$ and $b_Z$ (some entries are shared), form syndromes $\sigma_X=H_X b_X$ and $\sigma_Z=H_Z b_Z$, and decode each part separately to obtain $\hat b_X=\mathcal{C}(\sigma_X)$ and $\hat b_Z=\mathcal{C}(\sigma_Z)$. 
The $X$ ($Z$) part succeeds iff $H_X\hat b_X=H_X b_X$ and $A_X\hat b_X=A_X b_X$ (resp.\ $H_Z\hat b_Z=H_Z b_Z$ and $A_Z\hat b_Z=A_Z b_Z$). 
A sample is marked as a fail if either part fails. 
Results for the BB-circuit systems appear in \fig{decoder_comp}, showing behavior consistent with the $Z$-only case in the main text.

We now explain an approach to use bounds on the optimal onsets of the failure spectra of the X and Z systems $f_X(w)$ and $f_Z(w)$ to infer a bound on the optimal onset for the failure spectrum of the full system $f(w)$.
Our approach is based on the approximation $\hat{P}(p) = \hat{P}_X(p) + \hat{P}_Z(p)$, and also $\hat{P}_X(p) = \sum_{w=w_0} \hat{f}_X(w) {N_X \choose w} p^w (1-p)^{N_X-w}$ and $\hat{P}_Z(p) = \sum_{w=w_0} \hat{f}_Z(w) {N_Z \choose w} p^w (1-p)^{N_Z-w}$.
The onset fraction can be decomposed as $|F(w_0)| = |F_X(w_0)| + |F_Z(w_0)|$. 
Assuming $|F_X(w_0)| = |F_Z(w_0)|$ and noting $f_Z(w_0) = \frac{|F_Z(w_0)|}{{N_Z \choose w_0}}$ then:
\begin{equation}
f_Z(w_0) = \frac{{N \choose w_0} f(w_0)-|F_X(w_0)|}{{N_Z \choose w_0}} = \frac{f(w_0)}{2}\frac{{N \choose w_0}}{{N_Z \choose w_0}}.
\end{equation}
Using this relationship we estimate the bounds, shown in \fig{decoder_comp}, from the $Z$ and $X$ optimal bound in the paper. 
The BB(6)-circuit bound on the onset is in close agreement with the observed onset estimated for Relay, similar to what is seen for the $Z$ system in the main text. 
The fit parameters are provided in \tab{circ_system_fits_XnZ}. 

\begin{figure}
    \centering
    (a)\includegraphics[width=0.4\linewidth]{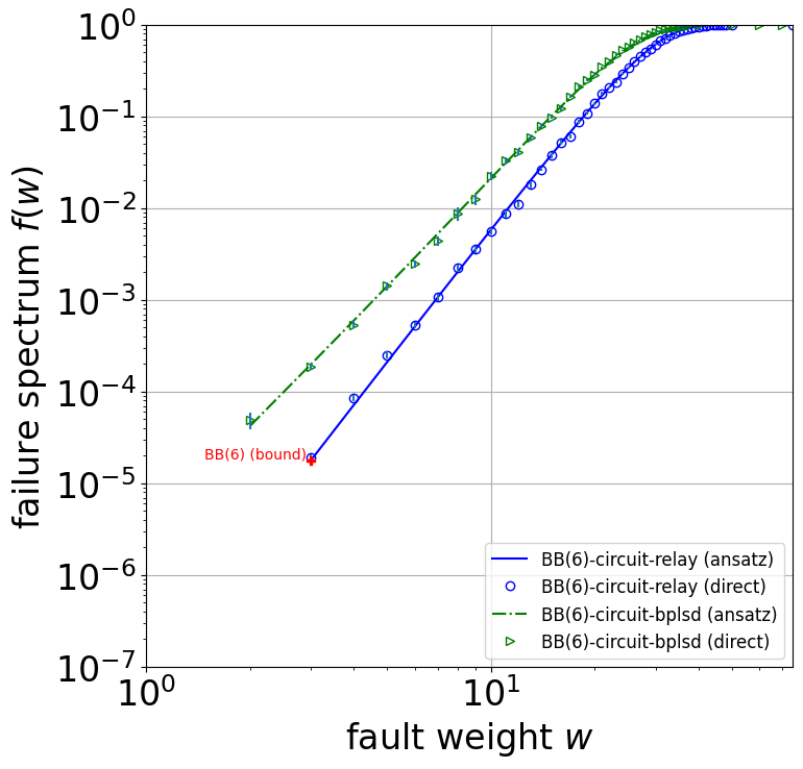}
    (b)\includegraphics[width=0.4\linewidth]{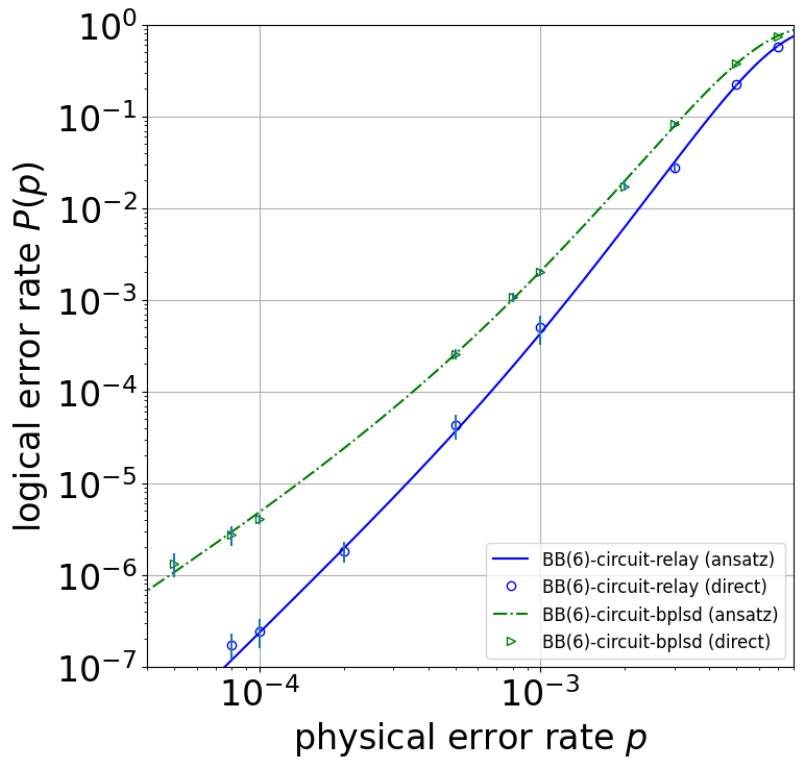}
    (c)\includegraphics[width=0.4\linewidth]{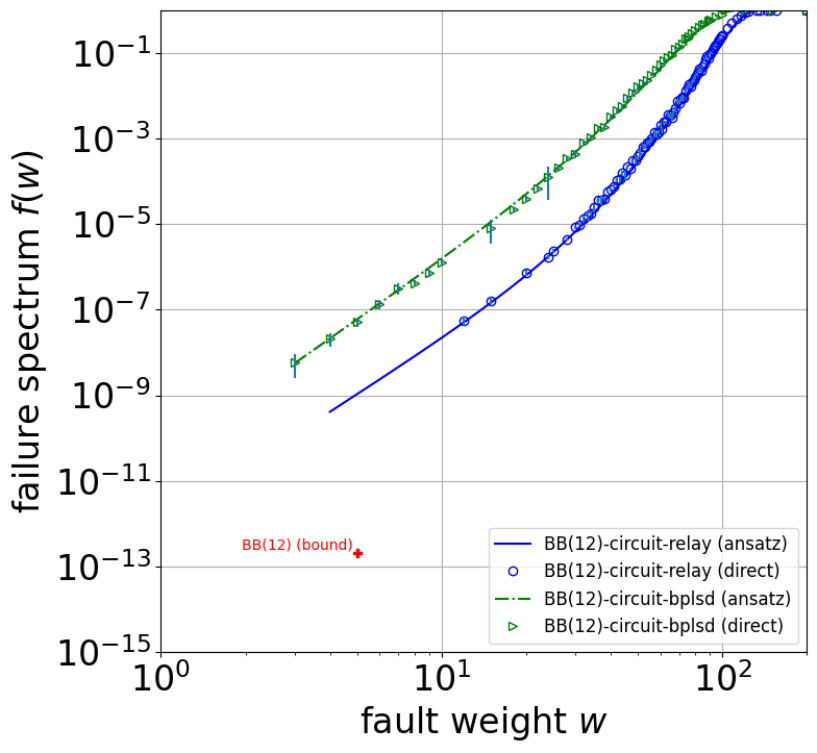}
    (d)\includegraphics[width=0.4\linewidth]{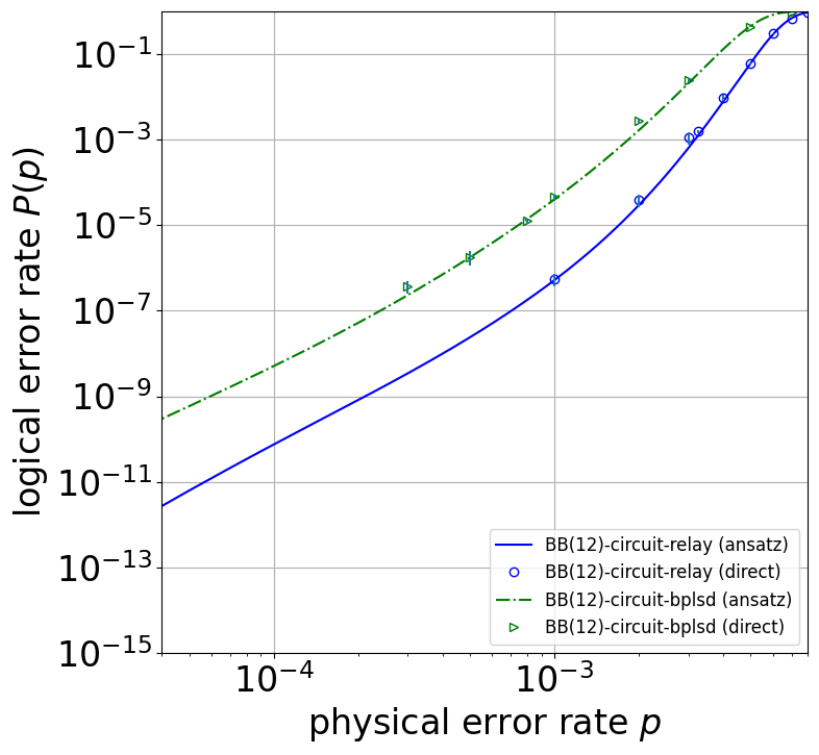}
    (e)\includegraphics[width=0.4\linewidth]{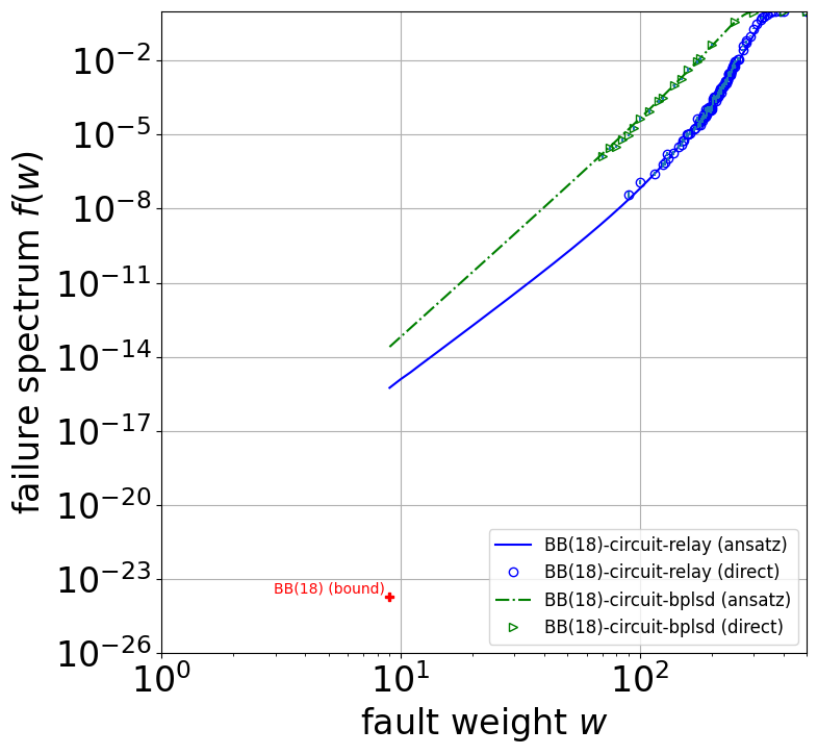}
    (f)\includegraphics[width=0.4\linewidth]{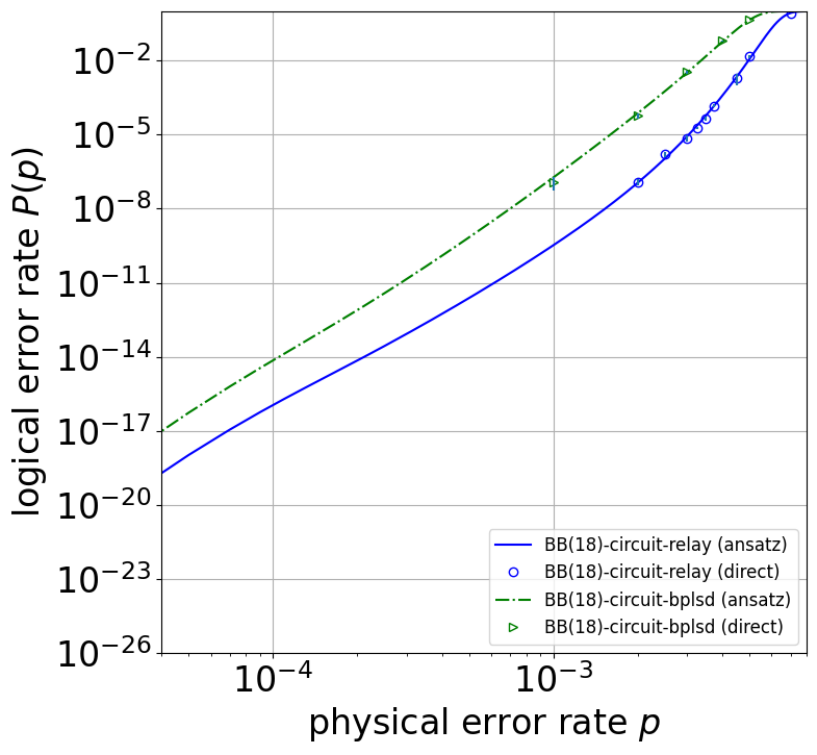}
    \caption{Logical error, $P(p)$, and failure fraction, $f(w)$, for BB(6), BB(12), BB(18) circuits for correlated noise but separate $X$ and $Z$ decoding using either BP-LSD or Relay. 
    The combined $f(w)$ and $P(p)$ data are fit to $f^{(5)}_\text{ansatz}(w)$, which is used to estimate $P(p)$ in the low noise regime and $f(w)$ for values of $w$ that are too small to practically sample. Optimal onsets for the BB(6) and BB(12) circuits are also shown (red crosses). 
    Fit parameters are listed in \tab{circ_system_fits_XnZ}.}
    \label{fig:decoder_comp}
\end{figure}

\begin{table}[h]
    \centering
    \begin{tabular}{lcccccc}
        \toprule
         & \multicolumn{5}{c}{\textbf{BB(6)-circuit}}  \\
        & $\hat{f_0}$ & $\hat{\gamma_1}$ & $\hat{\gamma_2}$& $w_0$ & $w_c$  & $N$ \\
        \midrule
        $\textbf{BP-LSD}$ & $1.8\times10^{-5}$ & $4.9$ & $4.4$ & $2$ & $15$ & $63936$ \\
        $\textbf{Relay}$ & $1.0\times10^{-5}$ & $3.3$ & $4.0$ & $3$ & $10$ & $63936$\\
        \midrule
        & \multicolumn{5}{c}{\textbf{BB(12)-circuit}} \\
        $\textbf{BP-LSD}$ & $5.6\times10^{-9}$ & $4.6$ & $8.1$ & $3$ & $39$ & $255744$\\
        $\textbf{Relay}$ & $4.1\times10^{-10}$ & $4.2$ & $12$ & $4$ & $48$ & $255744$ \\     
        \midrule
         & \multicolumn{5}{c}{\textbf{BB(18)-circuit}} \\
        $\textbf{BP-LSD}$ & $2.5\times10^{-14}$ & $8.7$ & $18$ & $9$ & $400$ & $762912$ \\
        $\textbf{Relay}$ & $5.5\times10^{-16}$ & $7.2$ & $27$ & $9$ & $275$ & $762912$\\
        \bottomrule
    \end{tabular}

    \caption{Fit parameters of $f^{(5)}_\text{ansatz}(w)$ for BB-circuit systems for correlated noise but separate $X$ and $Z$ decoding. 
    }
    \label{tab:circ_system_fits_XnZ}
\end{table}

We note there are several small differences in the QEC systems analyzed in this appendix compared to the rest of the paper: 
(1) the check matrices analyzed here were generated using STIM \cite{gidney2021stim} using circuit schedules in Ref.~\cite{BicycleArchitecturePaper}, rather than using the code in Refs.~\cite{bravyi2024high,bravyi2024github} as in the rest of the paper.  
(2) Here, decoder priors for $P(p)$ were adjusted to match $p$ for direct Monte Carlo sampling and were set using error rate $p=10^{-4}$ for importance sampling.
In the rest of the paper, we set the priors using error rate $p=10^{-4}$ for all decoding instances.
We do not observe any noticeable inconsistencies in the data that arise from these slightly different versions.

\clearpage
\subsection{Diversity of failure spectrum shapes}
\label{app:beyond-monomial}

In \sec{ansatz-approach}, \fig{QEC-system-examples-bitflip} and \fig{QEC-system-examples-circuit} present the failure spectra of several QEC systems, revealing a variety of shapes, some more intricate than $f_\text{model}(w)$ that we derived for the simple min-fail enclosure model in \sec{qec-system-model}.
Here, we provide additional detail on the mechanisms that can give rise to these differing forms of failure spectra.

Firstly, we note that the defining assumptions of the min-fail enclosure model are clearly too simple to accurately describe a real QEC system.
These assumptions are that (1) a bitstring fails if and only if it encloses a min-weight failing bitstring, and (2) min-weight failing bitstrings are independent.
We made an approximation in \eq{model-line4} when deriving an expression for $f_\text{model}(w)$ where we replaced $\binom{w}{w_0}$ by $(w/w_0)^{w_0}$, which we will refer to as assumption (3).

The shape of $f_\text{model}(w)$ is $~ (w/w_0)^{w_0}$, before turning over for large values of $w$, which is a straight line with gradient $w_0$ on a log-log plot.
We observe that the shapes of many of the QEC system failure spectra in \fig{QEC-system-examples-bitflip} and \fig{QEC-system-examples-circuit} have the straight-line shape, but with gradient is typically larger than $w_0$.
One explanation for this could be that assumption (1) is not accurate, and that there are large numbers of irreducible failing bitstrings with weight larger than $w_0$ which result in a larger gradient.

\paragraph{Binomial behavior before the asymptotic regime}
Another feature we observe is that some systems exhibit a failure spectrum that initially increases more steeply before transitioning to a shallower slope. One possible explanation is that assumption (3) does not hold in these cases. 
In \fig{binomial-fits}, we compare the binomial function $\binom{w}{w_0}$ with the simpler approximation $(w/w_0)^{w_0}$ and with the more flexible form $(w/w_0)^{\gamma}\left(\frac{1 + (w/w_c)^2}{1 + (w_0/w_c)^2}\right)^{(w_0 - \gamma)/2}$ used in $f^{(5)}_\text{ansatz}(w)$ in \eq{model-ansatz}. 
The latter provides a significantly better match to the binomial behavior across the full range of $w$, which encourages us that $f^{(5)}_\text{ansatz}(w)$ is expressive enough to make up for any inaccuracy introduced by assumption (3) in our model.

\begin{figure}[ht]
    \centering
    \includegraphics[width=0.9\linewidth]{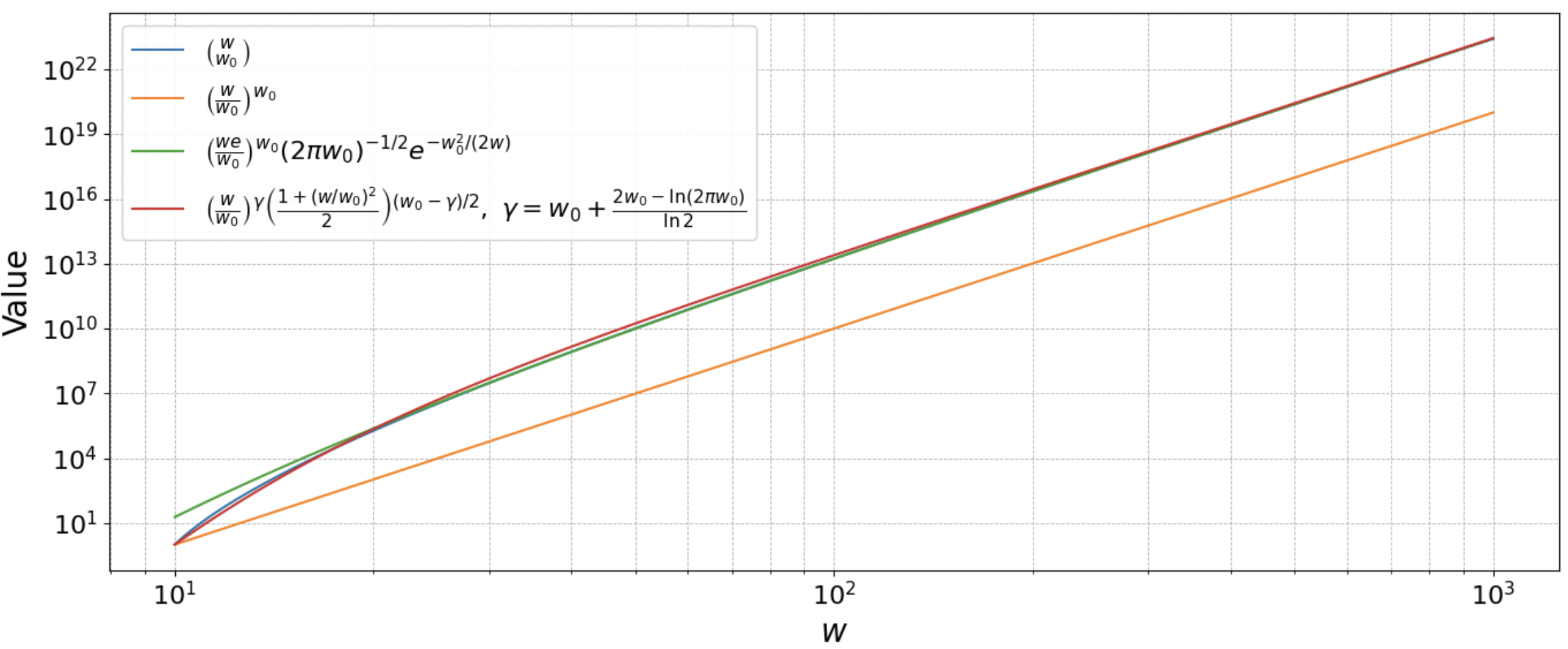}
    \caption{
    Comparison of the binomial function $\binom{w}{w_0}$ with several analytic parameterizations ($w_0 = 10$ for illustration). 
    The simple power-law form used in $f_\text{model}(w)$ matches exactly at $w = w_0$ but underestimates for large $w$. 
    The asymptotic form captures the correct large-$w$ behavior but deviates near $w = w_0$. 
    The more flexible form used in $f^{(5)}_\text{ansatz}(w)$ provides an accurate interpolation across the full range of $w$. 
    }
    \label{fig:binomial-fits}
\end{figure}

\paragraph{Composite QEC systems}
In some of the failure spectra in \fig{QEC-system-examples-circuit}, we observe a shallow gradient at lower $w$ on a log-log plot which increases for larger $w$. 
Here we consider a thought experiment to form a potential explaination of this behavior. 
Consider two QEC systems, $A$ and $B$, each well described by a constant $\gamma$ ansatz.
For simplicity, we assume that  
\begin{eqnarray}
f^A(w) &= f_0^A \left(\frac{w}{w_0^A} \right)^{\gamma^A}, \
f^B(w) &= f_0^B \left(\frac{w}{w_0^B} \right)^{\gamma^B},
\end{eqnarray}
where $f^A(w)$ and $f^B(w)$ vanish below $w_0^A$ and $w_0^B$, respectively, and $A$ and $B$ have $N^A$ and $N^B$ faults.

We can then consider a composite system $AB$ with $N^A + N^B$ faults and ask what its failure spectrum $f^{AB}(w)$ is.
The number of weight-$w$ failure configurations can be approximated as
\begin{eqnarray}
f^{AB}(w) \approx \sum_{w'}  \frac{\left[ f^{A}(w')\binom{N^A}{w'} \binom{N^B}{w-w'} + f^{B}(w')\binom{N^B}{w'} \binom{N^A}{w-w'} \right] }{\binom{N^A+N^B}{w}}.
\label{eq:composite-failure-spectrum}
\end{eqnarray}
In typical cases, the two terms in the numerator of~\eqref{eq:composite-failure-spectrum} are sharply peaked around distinct values, $w^A(w)$ and $w^B(w)$, indicating regimes where either subsystem $A$ or $B$ dominates the sum. 
A standard statistical mechanics approach is to estimate such sums with a second-order expansion of the logarithm around the peak, yielding tractable Gaussian integrals. 
The crossover point $w_c$ could then be estimated by solving a transcendental equation equating dominant contributions. We do not pursue this analysis, as it adds limited insight, and instead illustrate the behavior with two examples in \fig{enhanced-functional-complexity}.

\begin{figure}[h]
    \centering
    (a)\includegraphics[width=0.45\linewidth]{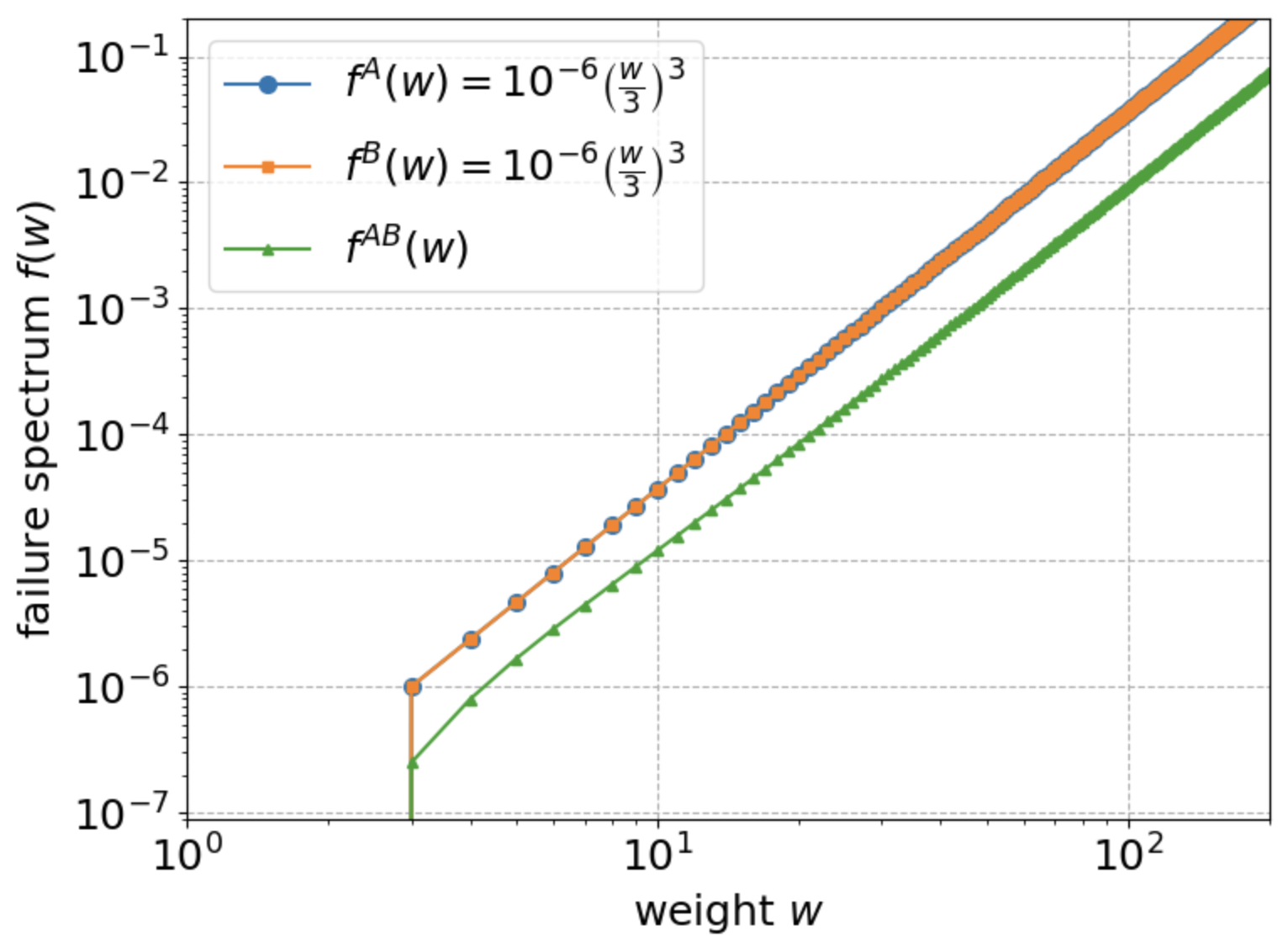}
    (b)\includegraphics[width=0.45\linewidth]{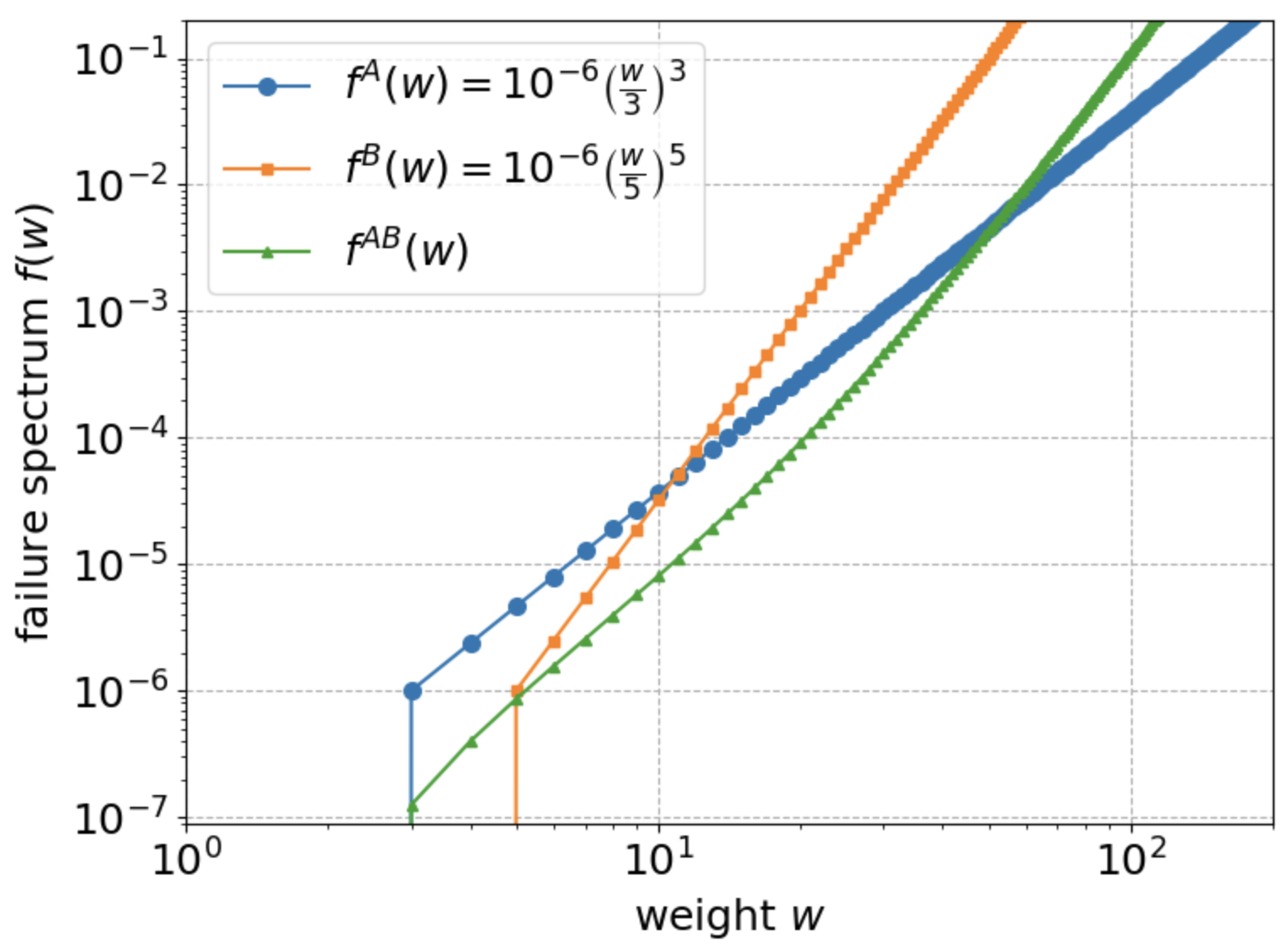}
    \caption{
    Failure spectra $f^A(w)$, $f^B(w)$, and $f^{AB}(w)$ for subsystems $A$, $B$ (each with $N^A = N^B = 1000$ faults), and their composite $AB$ ($N^{AB} = 2000$).
    (a) When $A$ and $B$ follow the same power law, $f^{AB}(w)$ is approximately a power law with the same slope, aside from deviations at small $w$.
    (b) When $A$ and $B$ follow different power laws, $f^{AB}(w)$ is no longer well described by a single power law and exhibits a crossover between the two subsystem slopes.
    }
    \label{fig:enhanced-functional-complexity}
\end{figure}

\paragraph{$\bar{Y}_1$ in-module measurement on the gross code}
In \cite{BicycleArchitecturePaper}, the authors observed some deviation from the ansatz fit at higher weights $w > 80$ for the $\bar{Y}_1$ in-module measurement on the gross code and excluded those points from the fit. 
This is a significant example of behavior where the failure spectrum starts with a shallow gradient at lower $w$ on a log-log plot and increases for larger $w$.
Here we present an extended set of simulation results for the same $\bar{Y}_1$ measurement circuits over fault weights between $15$ and $200$. The extended data is listed in Table~\tab{Ysystemdata}. In this new data set, there are enough samples to see more than 100 failures per fault weight for most data points. Furthermore, logical failures are categorized according to their type: memory errors where one or more of the 12 logical qubits is incorrect (with or without outcome errors), outcome errors where the observed $\bar{Y}_1$ eigenvalue is incorrect (with or without memory errors), and collective failures of both memory and outcome.

\begin{table}[h]
\begin{center}
\begin{tabular}{ccc|ccc|c}
$w$ & $T(w)$ & $F(w)$ & $F_\text{mem}(w)$ & $F_\text{out}(w)$ & $F_\text{both}(w)$ & $\bar{T}_\text{dec}$ \\ \hline
15 & 102753082 & 75 & 52 & 47 & 24 & 0.7 \\
20 & 40933905 & 82 & 58 & 38 & 14 & 0.8 \\
25 & 19914819 & 120 & 73 & 73 & 26 & 1.5 \\
30 & 8789630 & 99 & 57 & 63 & 21 & 1.3 \\
35 & 5557862 & 109 & 53 & 71 & 15 & 1.8 \\
40 & 2233581 & 107 & 31 & 86 & 10 & 1.9 \\
45 & 1077153 & 102 & 34 & 80 & 12 & 3.0 \\
50 & 623354 & 115 & 31 & 94 & 10 & 3.8 \\
53 & 623354 & 127 & 24 & 112 & 9 & 3.9 \\
55 & 360738 & 109 & 32 & 85 & 8 & 3.4 \\
58 & 250513 & 106 & 21 & 88 & 3 & 5.4\\
60 & 208761 & 108 & 23 & 93 & 8 & 5.9 \\
70 & 83898 & 121 & 25 & 103 & 7 & 7.3 \\
80 & 33718 & 111 & 32 & 89 & 10 & 15.0 \\
90 & 7844 & 100 & 53 & 67 & 20 & 48.7 \\
100 & 1826 & 94 & 72 & 38 & 16 & 81.8 \\
110 & 426 & 88 & 80 & 29 & 21 & 235.2 \\
120 & 120 & 56 & 54 & 18 & 16 & 189.5 \\
130 & 100 & 83 & 82 & 28 & 27 & 241.3 \\
140 & 100 & 97 & 97 & 37 & 37 & 307.4 \\
150 & 100 & 100 & 100 & 38 & 38 & 326.4 \\
200 & 100 & 100 & 100 & 38 & 38 & 306.4
\end{tabular}
\caption{Extended Monte-Carlo fault sampling data for the $\bar{Y}_1$ in-module measurement defined in \cite{BicycleArchitecturePaper}. The columns from left to right are: weight $w$, trials $T(w)$, failures $F(w)$, failures with memory errors with or without outcome errors $F_\text{mem}(w)$, failures with outcome errors with or without memory errors $F_\text{out}(w)$, failures with both memory and outcome errors $F_\text{both}(w)$, and the mean decoding time in seconds per sample $\bar{T}_\text{dec}$. The failure counts satisfy $F(w)=F_\text{mem}(w)+F_\text{out}(w)-F_\text{both}(w)$.\label{tab:Ysystemdata}}
\end{center}
\end{table}

The failure spectrum of the $\bar{Y}_1$ in-module measurement is plotted in \fig{Ysystemplot} (left). We fit each of the ansatz $f^{(i)}_{\textrm{ansatz}}(w)$ from \sec{model-ansatz}, finding poor fits with the $i=2,3$ parameter forms and significantly better fits using the $i=5,6$ parameter forms. To gain some insight into the shape of the failure spectrum, we plot the marginal failure spectra $f_\text{mem}(w)$ and $f_\text{out}(w)$ in \fig{Ysystemplot} (right). These are estimated as $f_\text{type}(w)=F_\text{type}(w)/T(w)$ for each $w$, where $F_\text{type}$ is the number of trials that resulted in decoder failure with a logical action of the given type, and $T(w)$ is the total number of trials. We fit each spectrum to the ansatz $f^{(6)}_{\textrm{ansatz}}(w)$. The inset of \fig{Ysystemplot}(b) plots the ratio $f_\text{mem}(w)f_\text{out}(w)/f_\text{both}(w)$, which provides some measure of how independent memory and outcome failures are as a function of the fault weight. We observe a transition near the fault weight $w=100$ that roughly coincides with a) crossing of $f_\text{mem}(w)$ and $f_\text{out}(w)$ and b) decoder failure probabilities that quickly exceed $20\%$ as $w$ increases.

\begin{figure}
    \centering
    \begin{subfigure}[t]{0.5\textwidth}
        \centering
        \includegraphics[height=2.2in]{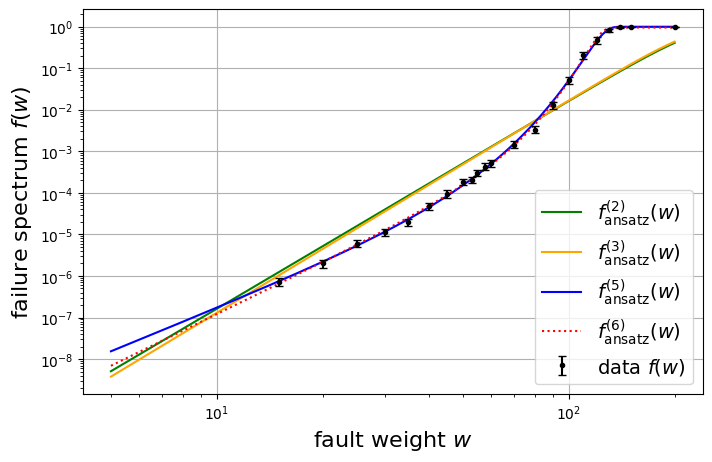}
    \end{subfigure}%
    ~ 
    \begin{subfigure}[t]{0.5\textwidth}
        \centering
        \includegraphics[height=2.2in]{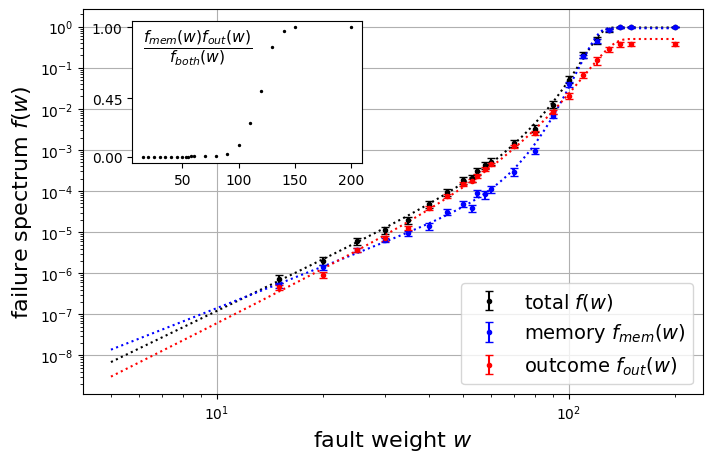}
    \end{subfigure}
    \caption{Failure spectra and ansatz fits for the $\bar{Y}_1$ in-module measurement on the gross code. (left) Each ansatz $f^{(i)}_{\textrm{ansatz}}(w)$ from \sec{model-ansatz} is fit to the data $f(w)=F(w)/T(w)$ in \tab{Ysystemdata}. (right) The ansatz $f^{(6)}_{\textrm{ansatz}}(w)$ is fit to the total, memory, and outcome failure spectra. The inset shows a plot of $f_\text{mem}(w)f_\text{out}(w)/f_\text{both}(w)$ as a measure of independence of memory and outcome failures.\label{fig:Ysystemplot}}
\end{figure}

\clearpage
\subsection{Comparison of failure spectrum ansatz to other fit formulas}
\label{app:fit-comparisons}
There are several fitting forms for $P(p)$ in the literature. 
A ubiquitous option is a simple power law $\kappa(\frac{p}{p_\text{th}})^{\frac{d+1}{2}}$ which finds utility in a number of instances \cite{fowler2012surface} but is limited considering that frequently $P(p)$ will have a more complex dependence on $p$. 
Alternatives with more parameterization have also been suggested. 
Here we show a comparison of a fitting using a power law with corrections suggested in Ref.~\cite{bravyi2013simulation}, $p^{d/2}e^{a+bp+cp^2}$, fit to Monte Carlo results of the Gross code decoded using Relay, \fig{Brav_f5_comp} (a) (black). 
The fit is compared to calculating $P(p)$ using $f^{(5)}_\text{ansatz}$, \fig{Brav_f5_comp} (a) (blue line) taking parameters that were used to fit the ansatz to the importance sampled failure spectrum shown in \fig{Brav_f5_comp} (b) (blue circles and blue line). We further show the failure spectra produced using the $f^{(5)}_\text{ansatz}$ achieved with parameters that reproduce the $P(p)=p^{d/2}e^{a+bp+cp^2}$ fit to the direct sampling of $P(p)$ \fig{Brav_f5_comp} (a) (green line) and (b) (black line). The failure spectra begin to differ at approximately $w < 20$.

Obtaining the failure fraction at $w \approx 10$ and above, for the BB(12)-circuit, allows direct calculation at $P(p) \approx 2\times10^{-4}$ and above. This requires at least $\mathcal{O}(10^{-7})$ shots. With similar direct sampling counts, values for $P(p\approx10^{-3})$ are obtainable. The importance sampling ansatz provides interpolated calculation of $P(p)$ to error rates several orders of magnitude less than direct sampling.  We further see that non-negligibly different and inaccurate predictions are produced when relying only on direct sampling of $P(p)$ fits compared to what can be obtained with importance sampling and ansatz method.

\begin{figure}[h]
    \centering
    \begin{subfigure}[b]{0.45\textwidth} 
    \centering 
    \includegraphics[width=0.84\linewidth]{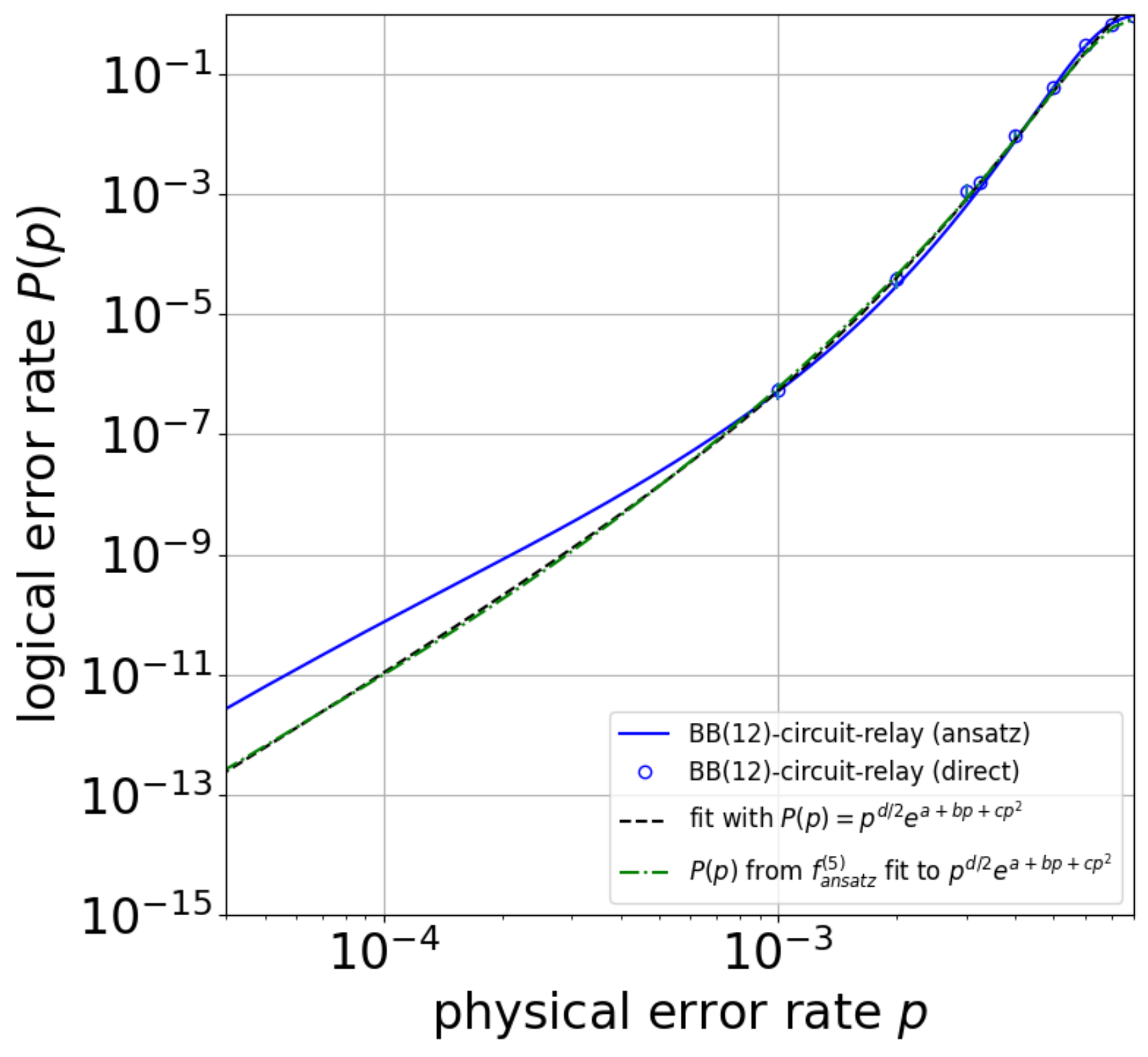}
    \caption{}
    \end{subfigure}
    \begin{subfigure}[b]{0.45\textwidth} 
    \centering 
    \includegraphics[width=0.84\linewidth]{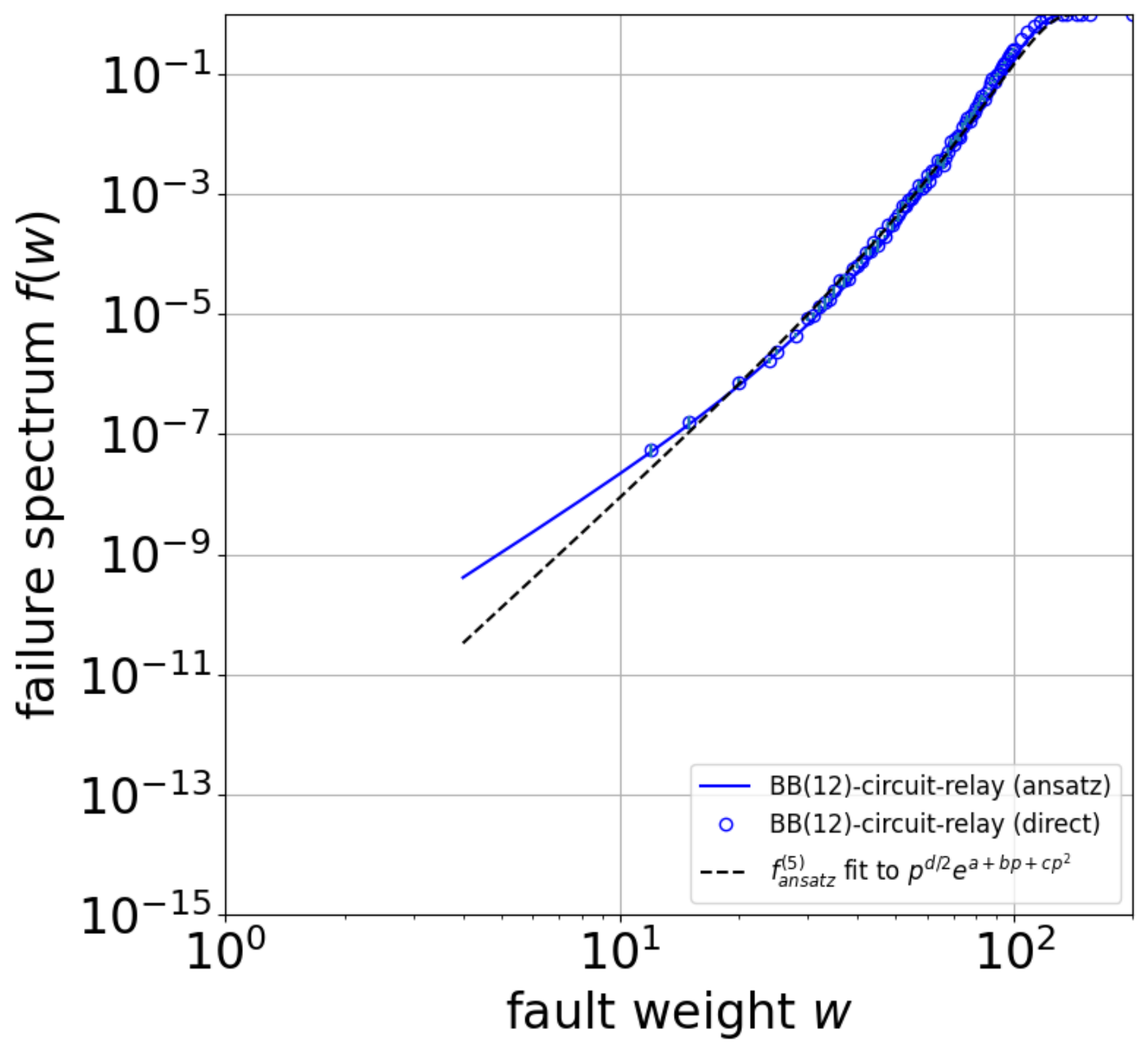}
    \caption{}
    \end{subfigure}

    \caption{Comparison of the two fitting approaches is shown in (a). The failure spectrum of the Gross code is fit using $f^{(5)}_\text{ansatz}$ (blue) (b) and $P(p)$ is obtained from $f(w)$. 
    This is compared to using $p^{d/2}e^{a+bp+cp^2}$ fit directly to the available Monte Carlo points, the most expensive number of trials being similar to that done for the failure spectrum. Also shown is an $f^{(5)}_\text{ansatz}$ fit that approximately reproduces the $p^{d/2}e^{a+bp+cp^2}$ fit (green), which highlights the differences in $f(w)$ for these two cases in the failure spectra (black and blue lines). The fit parameters for $p^{d/2}e^{a+bp+cp^2}$ are $[d,a,b,c] = [8,11.43,1881,-10303]$. The $f^{(5)}_\text{ansatz}$ fit to the modified power law for $P(p)$ are $[w_0,f_0,\gamma_1,w_c,\gamma_2] = [4,3.3\times10^{-11},6.05,68.18,11]$. The parameters for the $f^{(5)}_\text{ansatz}$ fit to the importance sampled failure spectra are the same as those indicated in table  \ref{tab:circ_system_fits_XnZ}.}
    \label{fig:Brav_f5_comp}
\end{figure}

\clearpage
\subsection{Theoretical analysis of failure-spectrum fitting}
\label{app:LMmethod}

Here we find an analytic optimization of the sample number $T=\sum_{w \in \mathcal{W}} T(w)$ in the setting where the goal is to use a failure-spectrum fit to estimate $P(p)$ at a specific value $p$ with standard deviation, $\sigma_P(p)$.

\paragraph{Lagrange multiplier method}

For this illustrative analysis we assume a priori knowledge of $P(p)$ and the $f(w)$. 
In practice, this approach might be used to update future sampling choices once initial fit parameters are obtained from early samples. 
Along these lines, a suggested approach using a heuristic seed to find an initial set of weights is discussed later in this appendix (next section).    

We now assume the observed set of decoder success and failures that produced our estimator $\hat{P(p)}$ has a standard deviation $\sigma_{0}^2(p)=\hat{P}(p)(1-\hat{P}(p))$. 
The standard error for $T_\text{MC}$ trials, up to a factor to account for the choice of confidence interval, is:
\begin{equation}
\begin{split}
    \sigma_{\hat{P}(p)}^2 = \frac{\sigma_{0}^2}{{T_\text{MC}}}=\frac{\hat{P}(1-\hat{P})}{T_\text{MC}} \\
\end{split}
\end{equation}

Importance sampling contributes uncertainty in the estimate from each of its terms. Adding the errors in quadrature:
\begin{equation}
\begin{split}
    \sigma_{\hat{P}_f}^2 = \sum \sigma^2_{\hat{P_w}} = \sum \frac{\hat{P}_w(p,w)(1-\hat{P}_w(p,w))}{T_w}  \\
\end{split}
\end{equation}
where the fail fraction, $f(w)$, enters through the importance sampled estimator for $P(p)$ as:
\begin{equation}
\label{eq:Pfestimator}
\begin{split}
    \hat{P}_f(p) = \sum_{i=w_0}^{N} \hat{P}_w(p,w) = (1-p)^{N}\sum_{i=w_0}^{N} f(w)\binom{N}{w}(\frac{p}{1-p})^{w}
\end{split}
\end{equation}
The total number of trials for sampling the single point $P(p)$ is $T_\text{IS}=\sum{T_w}$.

We now use the Lagrange multiplier method to find an optimization of the set of $T_w$ to get $T_\text{IS}$ for a particular choice of $p$ and target standard error, $\sigma^2$. We will also compare this to $T_\text{MC}$ for the same $p$ and targeted standard error. We first find a partitioning of $\sigma_f(w)$ using the Lagrange multiplier method. We define the Lagrangian $L$ as:

\begin{equation}
    \begin{split}
        L = \sum T_w - \lambda(\sigma^2-\sum \sigma^2_{\hat{P_w}}) \\
                L \approx \sum \frac{f(w)(1-f(w))}{\sigma_{f}^2(w)} - \lambda(\sigma^2 - \sum (\frac{\hat{P}(p,w)}{T_w })^2 ) \\
        L \approx \sum \frac{f(w)}{\sigma_f^2(w)} - \lambda(\sigma^2 - \sum (B_w\sigma_{f}(w))^2)
    \end{split}
\end{equation}
where $p$ is small enough that to good approximation $(1-P)\approx1$ and furthermore $(1-f(w))\approx1$ because for all weights that we are interested in optimizing $f(w) \ll 1$. 
The quantity sampled to obtain $\hat{P}_w(p,w)$ is the estimated failure fraction $\hat{f}(w)$ with a standard deviation $\sigma^2_{f}(w) = \frac{ \hat{f}(w) (1- \hat{f}(w) ) }{T_w}$. 
The standard deviation of the each weights contribution, $\hat{P_w}$, to the total $P(p)$, is $\sigma_{\hat{P_w}} =B_w \sigma_{f}(w)$, where $B_w=\binom{N}{w}(\frac{p}{1-p})^w(1-p)^N$, is the binomial distribution term, and $P_w(p,w)=B_wf(w)$. 
The standard deviation, $\sigma_{\hat{P_w}} =B_w \sigma_{f}(w)$ follows from the linear relationship between $P_w$ and $f(w)$, $\sigma_{\hat{P_w}}= \frac{\partial P(p,w)}{\partial f(w)}\sigma_{f}(w) = B_w\sigma_f(w)$.

Evaluating the system of equations:
\begin{equation}
    \begin{split}
        \frac{dL}{d\sigma_{f}(w)} = \frac{-2f(w)}{\sigma_{f}^3(w)} + 2\lambda B_w\sigma_{f}(w) = 0 \\
        \frac{dL}{d\lambda} = \sigma^2 - \sum (B_w\sigma_{f}(w))^2 = 0 
    \end{split}
\end{equation}
solving first for $\lambda$, we rewrite $\frac{dL}{d\lambda}=0$:
\begin{equation}
    \begin{split}
        \lambda = \frac{f(w)}{B_w\sigma_{f}^4(w)}=\frac{f(k)}{B_i\sigma_{f}^4(k)} \\
        \sigma_f^2(k) = \sqrt\frac{B_wf(k)}{B_kf(w)}\sigma_f^2(w)
    \end{split}
\end{equation}
now using $(B_w\sigma_f(w))^2 = {\sigma^2-\sum_{k \neq w}(B_k\sigma_f(k))^2}$ we get:
\begin{equation}
    \begin{split}
        \sigma_f^2(w) = \frac{\sigma^2 - \sigma_f^2(w)\sum_{k \neq w}B_k^2  \sqrt\frac{B_wf(k)}{B_kf(w)}}{B_w^2}, \\
        \sigma_f^2(w) = \frac{\frac{\sigma^2}{B_w^2}}{1+\sum_{k\neq w}(\frac{B_k}{B_w})^2\sqrt{\frac{B_wf(k)}{B_kf(w)}}}, \\
        \sigma_f^2(w) = \frac{\frac{\sigma^2}{B_w^2}}{\sum^{N}_{k=w_0}(\frac{B_k}{B_w})^2\sqrt{\frac{B_wf(k)}{B_kf(w)}}}
    \end{split}
    \label{eq:optSig}
\end{equation}
Returning to $\sigma_f(w) = \frac{f(w)(1-f(w))}{T_w}$ and the optimized $\sigma_f(w)$ provides an optimized number of trials for each weight of:
\begin{equation}
\label{eq:opt_trial_w}
    T_w = f(w)(1-f(w))\frac{B_w^2}{\sigma^2}\sum_{k=w_0}^{N} (\frac{B_k}{B_w})^2\sqrt{\frac{B_wf(k)}{B_kf(w)}}
\end{equation}

We may also compare $T_\text{IS}$ to the number of direct samples at $p$ to obtain the same target standard error, $T_\text{MC}$: 
\begin{equation}
    \frac{T_\text{IS}}{T_\text{MC}} = \frac{ \sum^{w_{h}}_{i=w_{l}} T_w }{ \frac{\hat{P}(p)(1-\hat{P}(p))}{\sigma^2} }
\end{equation}
where $T_\text{MC}=\frac{P(p)(1-P(p))}{\sigma^2}$ and $w_{l,h}$ are the smallest and largest weight sampled directly to obtain an accurate estimation of $P(p)$. We show the choice of samples for each weight $T_w$ for a target standard error for several choices of $p$ in Fig. \fig{Tws}. We also see qualitative agreement between this weight distribution and the dominant distribution of weights found, for example, in the splitting method shown in \fig{splitting_bb}.

\begin{figure}[ht!]
    \centering
    (a)\includegraphics[width=0.45\linewidth,clip,trim=0mm 0mm 0mm 0mm]{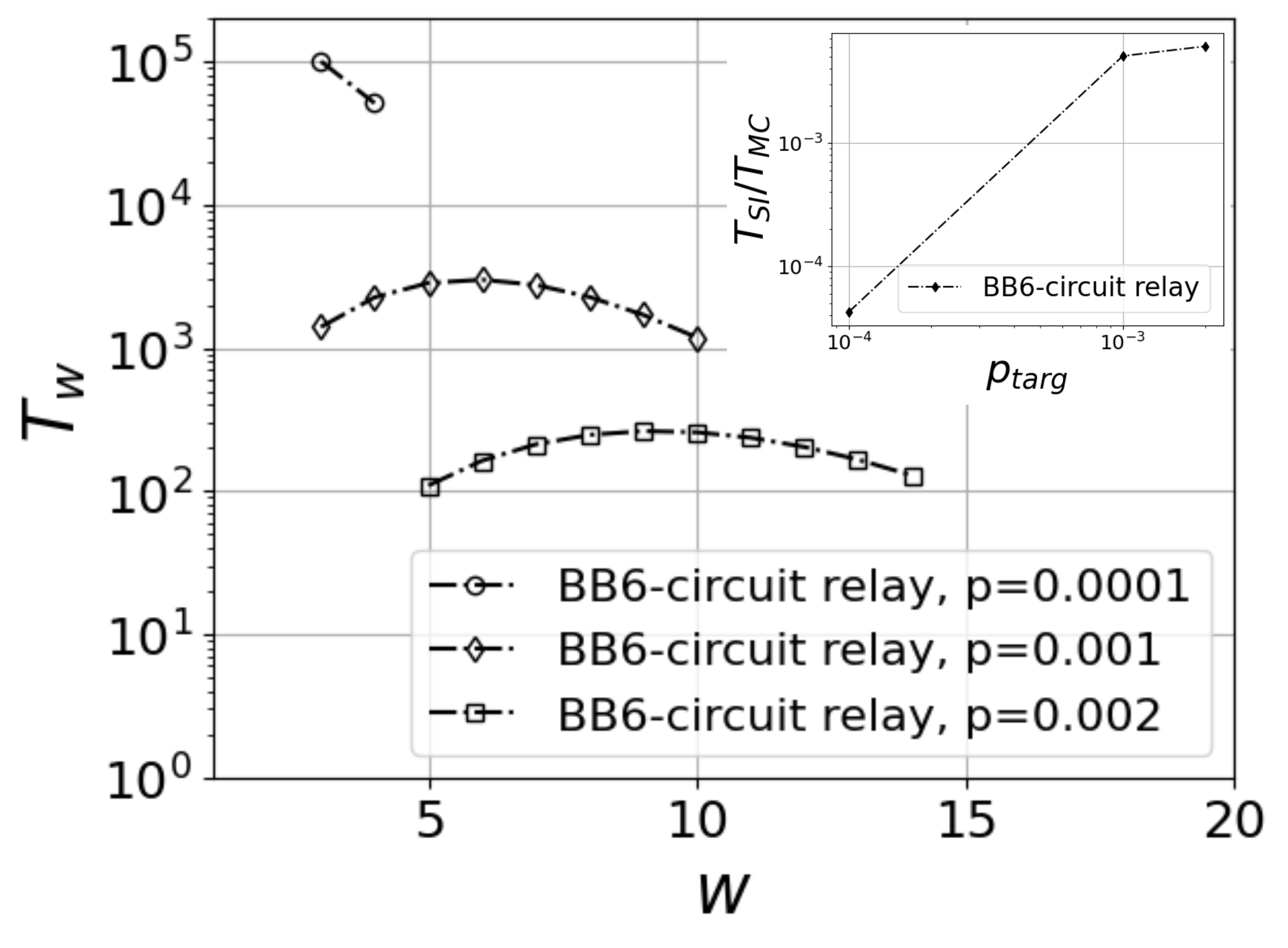}    
    (b)\includegraphics[width=0.445\linewidth,clip,trim=0mm 0mm 0mm 0mm]{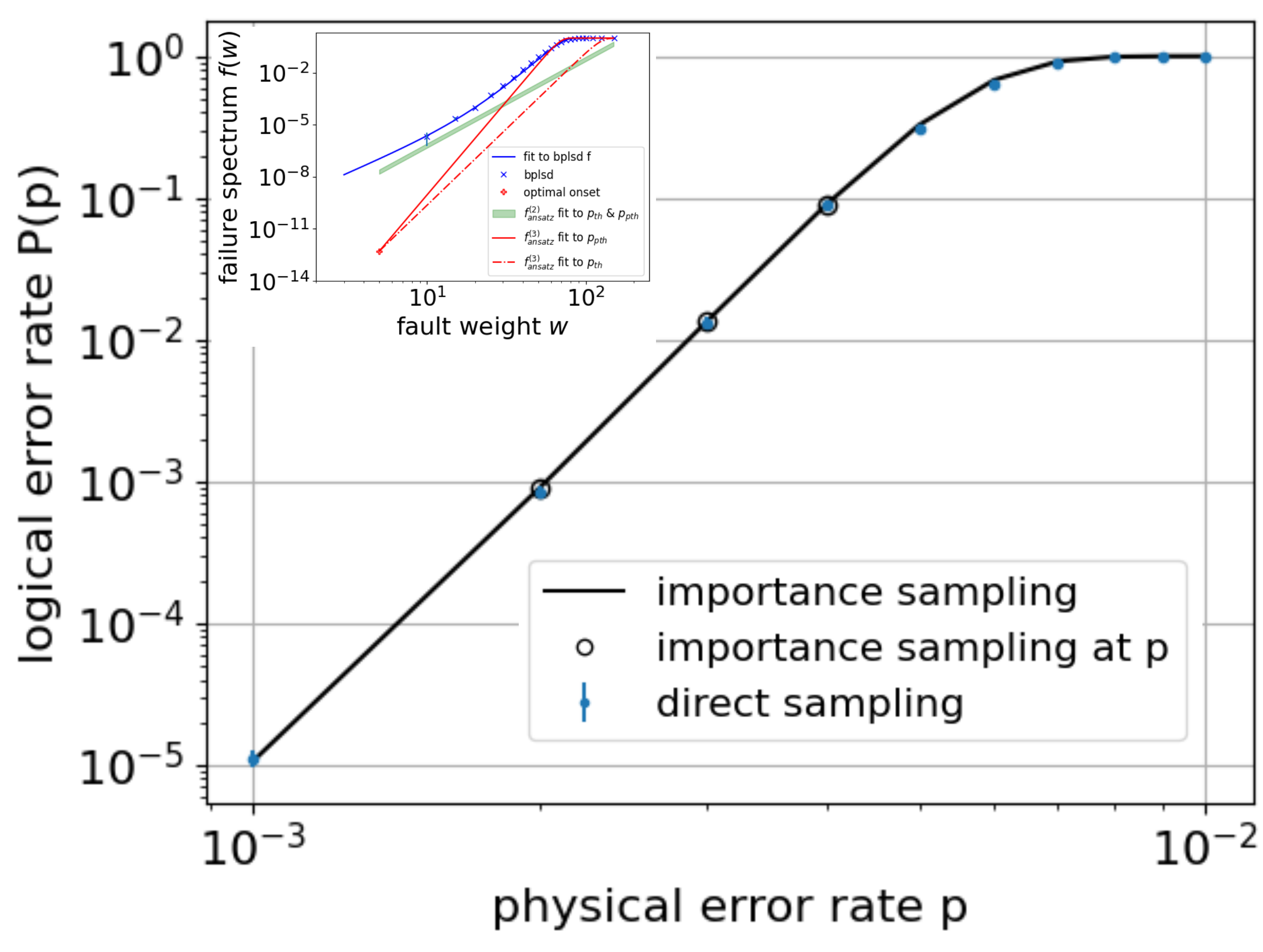}
    
    \caption{
    (a) Optimized number of trials for BB(6)-circuit for each weight, $T_w$, to obtain a standard deviation $\sigma = 0.05\times P(p)$ for target $p = {0.0001, 0.001, 0.002}$. The $T_w$s are found for each $p$ assuming $f^{(3)}_\text{ansatz}(w)$. The $w_\text{low}$ and $w_\text{hi}$ are selected to sample $\geq$0.95 of the contributions to $P(p)$. The inset is a comparison of the total number of trials using importance sampling (SI) or the Monte Carlo method (MC) for each $p$. (b) $\hat{P}_{f}(p)$, black circles and solid line, found using a heuristic choice of starting weights described in section \app{LMmethod}, and compared to the direct sampling of the BB(12)-circuit using the BP-LSD decoder parameters. The inset shows the heuristic seeds $f^{(2)}_\text{ansatz}$ and $f^{(3)}_\text{ansatz}$ when solving $P_f(p_\text{pth})=p_\text{pth}$ for $f_0$ or $\gamma$. The 'guess' for the pseudothreshold is either the literature value for the circuit noise threshold of the code, $\approx 0.007$ \cite{bravyi2024high}, or the empirically observed pseudothreshold for BB(12)-circuit using BP-LSD, $\approx 0.0021$. The optimal onset bound is used when solving for $\gamma$ in $f^{(3)}_\text{ansatz}$. The onset was set at $w_0=5$ the nominal circuit distance for the BB(12)-circuit. The total number of direct sampling and importance sampling trials to obtain $\hat{P}(p)$ or $\hat{P}_f(p)$ for $p=\{0.004,0.003,0.002\}$ was $\approx10M$ and $\approx300k$, respectively, resulting in standard errors or bootstrap estimated errors of $[4.5\times10^{-4},5\times10^{-5},1.3\times10^{-5}]$ and $[5.7\times10^{-3},6.5\times10^{-4},3.3\times10^{-5}]$, respectively. The reduced chi for the three target $p$ and for all points for which direct sampling was available were $0.001$ and $0.003$, respectively. 
    }
    \label{fig:Tws}
\end{figure}

\paragraph{Heuristic weight and trial schedule for target list of physical error rates}

A common task is to obtain $P(p)$ for a set of physical error rates $\{p_\text{hi}, ...,p_\text{low}\}$. Importance sampling offers an opportunity to find $P(p)$ with fewer trials. A practical challenge, however, is that the failure spectrum, $f(w)$, is not known for a new QEC system making it unclear with which weights and trials to start. 

The Lagrange multiplier method above provides quantitative choices for which weights and number of trials to sample. However, it relies on an estimate of the local failure spectrum as an input. Once an accurate parameterization of a local ansatz is obtained, projections of which weights and trials to sample for a subsequent $p$ can be made. The problem is how to choose seed parameters for an initial ansatz.  

We suggest the following heuristic approach to identify initial parameters for an ansatz when an initial upper bound of the circuit distance, $D$, and a magnitude estimate of the pseudothreshold, $p_\text{pth}$ is available. When these two parameters are available, which is often the case just before characterizing a code's performance with MC sampling, we observe that $f_0$ is the only unknown in $f^{(2)}_\text{ansatz}$, when setting $w_0 = \lceil D/2 \rceil$. The $f_0$ can be solved for by setting $p_\text{pth}=(1-p_\text{pth})^N\sum_{i=w_0}^N f^{(2)}_\text{ansatz}(w) \binom{N}{w} (\frac{p_\text{pth}}{(1-p_\text{pth})})^w$. This defines initial parameters that can be used to select weights for a target $p$. The ansatz parameters are then immediately updated after getting estimates of $\hat{f}(w)$. Furthermore, if the optimal onset, $f^*(w_0^*)$, is known from fault counting or some other method, then an improved initial ansatz can be made in a similar way but now solving for $\gamma$ and using $f_0 = f^*(w_0^*)$ in $p_\text{pth}=(1-p_\text{pth})^N\sum_{i=w_0}^N f^{(3)}_\text{ansatz}(w) \binom{N}{w} (\frac{p_\text{pth}}{(1-p_\text{pth})})^w$. We now enumerate a sketch of this heuristic method which was implemented to obtain the results shown in Fig. \fig{Tws} (b): 

\begin{enumerate}
    \item Numerically solve for $f_0$ (or $\gamma$) in $p_\text{pth}=(1-p_\text{pth})^N\sum_{i=w_0}^N f^{(2)\,\text{or}\,(3)}_\text{ansatz}(w) \binom{N}{w} (\frac{p_\text{pth}}{(1-p_\text{pth})})^w$ (a transcendental equation) using a best guess for the onset $w_0$ (and also $f_0$) and the pseudothreshold $p_\text{pth}$ of the QEC system. 
    \item Choose a target minimum accuracy and find the minimum number of weights that satisfies the accuracy target for $p_\text{hi}$ according to equation \eq{Pfestimator} and using $f^{(2)\,\text{or}\,(3)}_\text{ansatz}(w)$.  
    
    \item Choose a target standard error and evaluate equation \eq{opt_trial_w} using the weights found in the previous step.
    
    \item Obtain an initialization sampling of $f(w)$ at the min, max and median weights indicated by equation \eq{opt_trial_w} and directly sample at $p_\text{pth}$. Choose the trial numbers indicated by eqn. \eq{opt_trial_w} or a minimum floor of samples (e.g., sample until a minimum of $n_\text{min-fail}$ failures set according to a target standard error). 
    
    \item Using this initial set of sampled $f(w)$ and $p_\text{pth}$, fit $f^{(5)}_\text{ansatz}(w)$ with fixed $w_0$. 
    An initial set of priors for expanding to $f^{(5)}_\text{ansatz}(w)$ from  $f^{(2)\,\text{or}\,(3)}_\text{ansatz}(w)$ might be $[f_0, w_0, \gamma_1=w_0$ (or $\gamma), w_c=2w_0, \gamma_2 = w_0$ (or $\gamma)$].
    
    \item Sequentially select the $p$ values in the target set starting highest to lowest (i.e., start with $p_\text{hi}$). Find the target trial numbers from eqn. \eq{opt_trial_w}. Select and sample the min, median and max weights. Refit the ansatz and then repeat for the next $p$.
    \item Use the last fit parameters of the ansatz to calculate the $P(p)$ in the target set also noting that the updated parameters allow for ansatz calculated interpolation and extrapolation around the domain of the target set. 
\end{enumerate}

\clearpage
\subsection{Numerical analysis of failure-spectrum fitting}
\label{app:fit-strategies-all-p}

Here we consider strategies for characterizing a QEC system by sampling and fitting a failure spectrum ansatz. 
Rather than deriving an optimal strategy, we evaluate representative strategies on a specific QEC system and distill general principles (presented at the end of the subsection) to guide sampling and fitting in broader contexts.

For the analysis here, we define a \emph{sampling strategy} by specifying a set of weights $\mathcal{W}$, a set of error rates $\mathcal{P}$, and the number of samples allocated to each, denoted $T(w)$ for $w \in \mathcal{W}$ and $T(p)$ for $p \in \mathcal{P}$. 
The total number of samples is then $T = \sum_{w \in \mathcal{W}}T(w) + \sum_{p \in \mathcal{P}}T(p)$.
For simplicity, we assume uniform per-sample cost and a fixed total sample budget $T$, although in practice decoding cost may vary across regimes.
Those samples yield estimates $\hat{f}(w)$ and $\hat{P}(p)$ with corresponding standard errors $\sigma_f(w)$ and $\sigma_P(p)$ at values $w \in \mathcal{W}$ and $p \in \mathcal{P}$.  
We then fit the parameters $\vec{\xi}$ of the ansatz $f_\text{ansatz}(w;\vec{\xi})$ by minimizing the combined weighted least-squares error:
\begin{equation}
\chi^2(\vec{\xi}) =
\sum_{w \in \mathcal{W}} \frac{\bigl[\hat{f}(w) - f_\text{ansatz}(w;\vec{\xi})\bigr]^2}{\sigma_f^2(w)} 
+ \sum_{p \in \mathcal{P}} \frac{\bigl[\hat{P}(p) - P_\text{ansatz}(p;\vec{\xi})\bigr]^2}{\sigma_P^2(p)}.
\label{eq:chi-sq}
\end{equation}

\textbf{Benchmarking.}
We perform numerical experiments with different sampling strategies to identify effective approaches.  
We use RT(12)-bitflip as a well-characterized reference system. 
Over the range shown in \fig{failure-spectrum-transform}, $f(w) > 10^{-4}$, allowing estimates with 1\% relative error using a feasible number of samples and defining the reference curve $P_\text{ref}(p)$.
The samples collected according to a given strategy are fit to the ansatz $f_\text{ansatz}^{(6)}(w\,|\,\vec{\xi})$, yielding parameters $\vec{\xi}_\text{strat}$ and the corresponding function $f_\text{strat}(w) := f_\text{ansatz}^{(6)}(w\,|\,\vec{\xi}_\text{strat})$.  
This defines the inferred curve $P_\text{strat}(p)$, which can be compared to $P_\text{ref}(p)$.
We quantify the discrepancy between $P_\text{strat}(p)$ and $P_\text{ref}(p)$ using the \emph{maximum deviation} $\mathcal{R}(P_\text{strat}(p),P_\text{ref}(p))$, defined as\footnote{In practice, we evaluate this numerically for 300 logarithmically spaced values in the range $p \in [10^{-5}, 0.5)$, and for the asymptotic value as $p\rightarrow 0$ using the ratio between $f_\text{strat}(w_0)$ and $f_\text{ref}(w_0)$.}
\[
\mathcal{R}(P_\text{strat}(p),P_\text{ref}(p)) = \max_{p \in (0, 0.5]} \max \left\{ 
\frac{P_\text{strat}(p)}{P_\text{ref}(p)},\,
\frac{P_\text{ref}(p)}{P_\text{strat}(p)}
\right\}.
\]  
By maximizing over $p \in (0, 0.5]$, this notion serves as a figure of merit across regimes. 
In section \app{LMmethod} we additionally consider an analytically informed strategy using the failure spectrum ansatz to select weights for target values of $p$.
Because $P_\text{strat}(p)$ fluctuates across repeated applications of the same strategy due to sampling randomness, we repeat each experiment $Q$ times with independent samples to obtain 
$$\{ P_\text{strat}^{(1)}(p),\, P_\text{strat}^{(2)}(p),\, \dots,\, P_\text{strat}^{(Q)}(p) \}.$$  
We then report the median of the resulting max-deviations 
$$\{ \mathcal{R}(P_\text{strat}^{(1)}(p),P_\text{ref}(p)), \mathcal{R}(P_\text{strat}^{(2)}(p),P_\text{ref}(p)), \dots, \mathcal{R}(P_\text{strat}^{(Q)}(p),P_\text{ref}(p))\},$$ 
along with a bootstrap-estimated uncertainty of the median value.  
Values closer to one indicate that the strategy more accurately reproduces $P_\text{ref}(p)$ in typical runs.

\paragraph{Failure spectrum sampling.}
We first consider sampling strategies that use only failure spectrum data to fit in \eq{chi-sq}, without sampling of $P(p)$.  
We further assume the onset $w_0 = 6$ is known, and that seven weights $[w_1, w_2, \dots, w_7]$ are sampled, with corresponding sample allocations $[T_1, T_2, \dots, T_7]$. 
There are seven data points and five free parameters of the ansatz $f_\text{ansatz}^{(6)}(w\,|\,\vec{\xi})$ (having fixed $w_0$), so the problem is mildly over constrained which is reasonable for finding fit parameters.

We begin with a class of naive strategies, where $[w_1, w_2, \dots, w_7]$ are consecutive integers starting from a chosen $w_1$, and the total budget $T$ is distributed uniformly among them.  
In \fig{deviation}(a) we examine the performance of these strategies as a function of $w_1$ for $T$ ranging from $10^2$ to $10^5$.  
Increasing $T$ consistently reduces the maximum deviation as expected. 
For a fixed $T$, the maximum deviation decreases as $w_1$ decreases, until $w_1$ becomes so small that failures at $w_1$ are too rare to be captured with the allocated $T/7$ samples; beyond this point, performance degrades. 
As $T$ increases, this turnover point shifts to smaller $w_1$. 
Notably, when $T = 10^5$, the value $w_1 = 6$ (the onset weight for this QEC system) achieves the best performance.

\begin{figure}[ht!]
    \centering  
    (a)\includegraphics[width=0.45\linewidth]{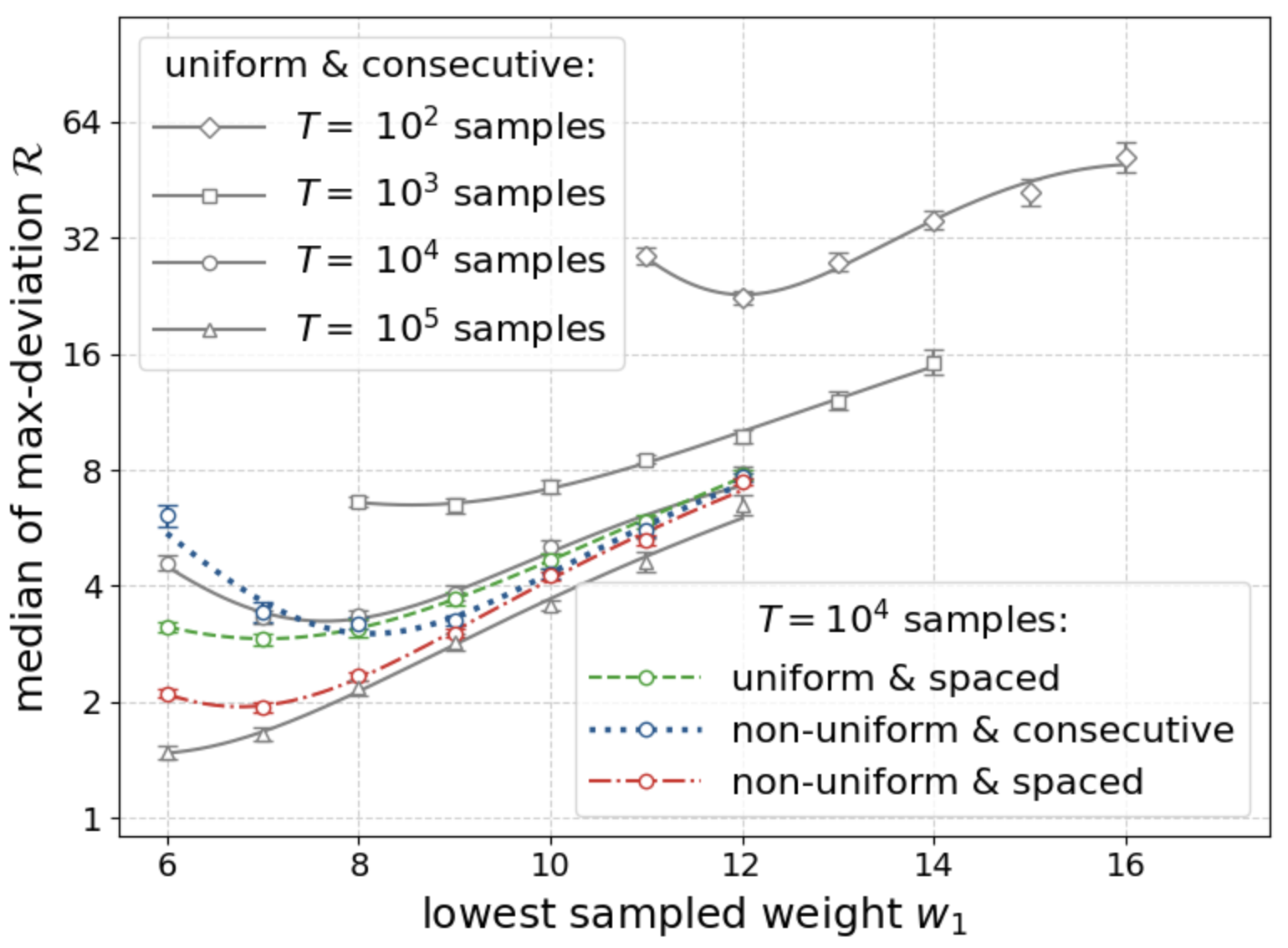}
    (b)\includegraphics[width=0.45\linewidth]{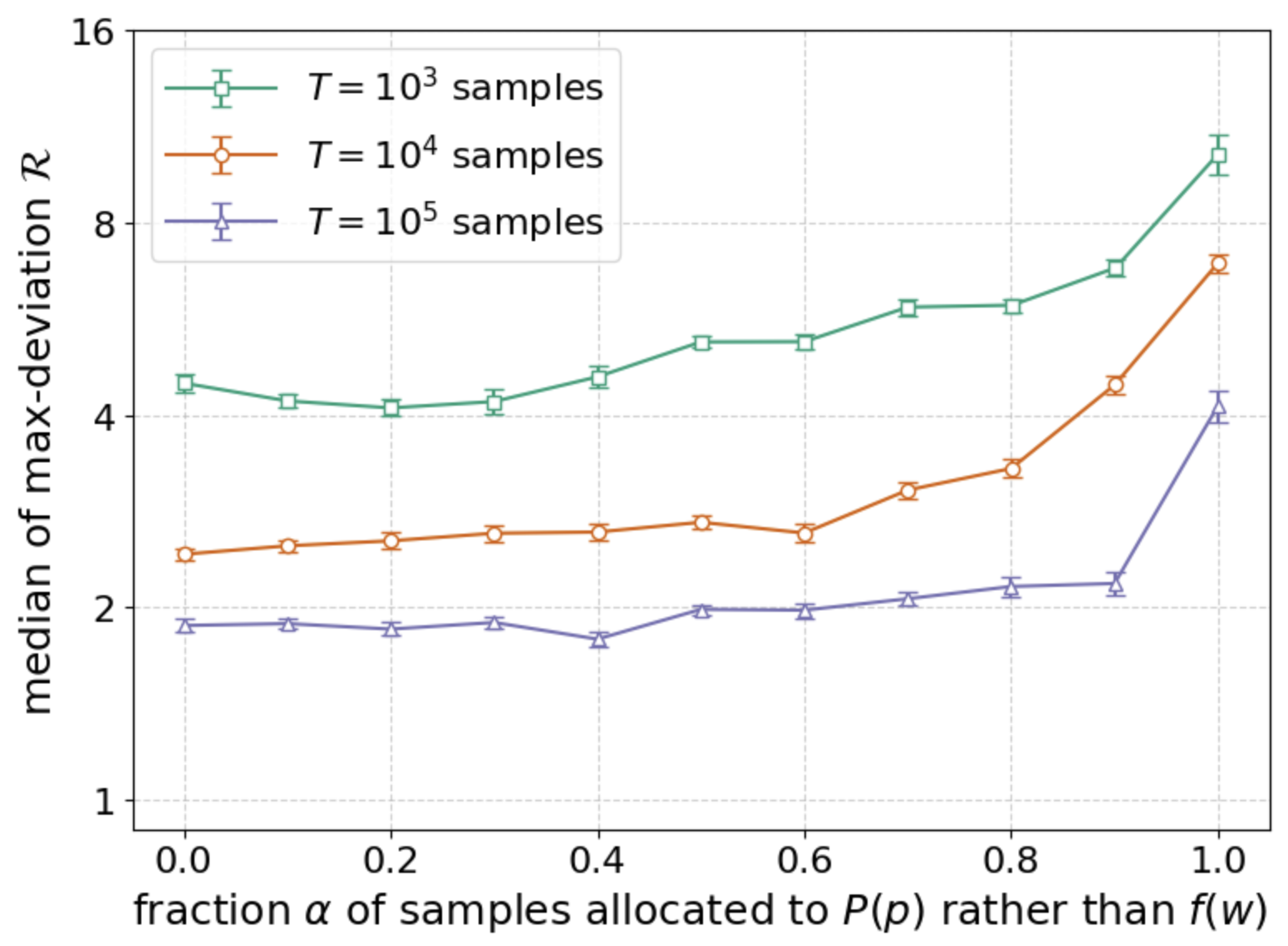}
    \caption{
    Comparison of sampling strategies for the QEC system RT(12)-bitflip.  
    (a) Each strategy specifies seven weights and allocates $T$ samples across them.  
    For each strategy, the maximum deviation $\mathcal{R}(P_\text{strat}(p), P_\text{ref}(p))$ (defined in the text) is computed and the median over many repetitions is plotted as a function of the lowest sampled weight $w_1$, with bootstrap-estimated error bars.  
    Smaller deviations indicate more accurate recovery of $P_\text{ref}(p)$.  
    The naive strategy (gray) uses consecutive weights with samples uniformly allocated; results are shown for total budgets $T = 10^2, 10^3, 10^4, 10^5$.  
    At $T=10^4$, we also compare: (i) spaced weights with uniform allocation, (ii) consecutive weights with non-uniform allocation, and (iii) spaced weights with non-uniform allocation.  
    These results suggest that a non-uniform allocation over spaced weights, with $w_1$ decreasing as $T$ increases, achieves the most accurate fits.  
    (b) Median maximum deviation as a function of the fraction $\alpha$ of the sample budget allocated to logical error rate data $P(p)$, with the remaining $(1-\alpha)$ allocated to failure spectrum data $f(w)$.  
    Results are shown for total budgets $T = 10^3, 10^4, 10^5$ (see text for details).  
    } \label{fig:deviation}
\end{figure}

We also consider the following three modified classes of sampling strategy at $T=10^4$, with results shown in \fig{deviation}(a).
(i) Using spaced weights ($w_{j+1} = w_j + 2$) instead of consecutive weights probes a broader range of $w$.
This moderately improves upon the naive strategy at smaller $w_1$, presumably because the wider range better constrains the fit.  
At larger $w_1$, however, the highest weights enter a regime where $f(w) \approx a$ and contribute little information to the fit (for context, in this QEC system, $f(w) > 0.9a$ for $w \geq 23$; see \fig{failure-spectrum-transform}).
(ii) Allocating samples to equalize the expected number of observed failures across weights. 
To achieve this, one could dynamically sample the data point with the fewest observed failures until the sample budget is exhausted (in our numerical implementation, which assumes pre-allocated samples, we mimicked this by assigning samples in proportion to $1/f(w)$).
This improves upon the naive strategy at intermediate $w_1$ by reducing uncertainty at lower weights, thereby better constraining the fit. 
However, it degrades performance at low $w_1$, since the number of observed failures per data point (approximately $T f(w_1)$) becomes too small, increasing uncertainty and weakening the fit constraint.
(iii) Combining both strategies significantly improves performance, with the median deviation approaching that of the naive strategy with ten times more samples.

\paragraph{Logical error rate sampling.}  
We now consider strategies that incorporate logical error rate samples $P(p)$ into the fit \eq{chi-sq}.  
To illustrate the effect of including $P(p)$ data, we present numerical experiments in \fig{deviation}(b).  
We fix the total sample budget $T$ (ranging from $10^3$ to $10^5$), split into $\alpha T$ for failure spectrum and $(1-\alpha) T$ for logical error rates, with $\alpha$ varied between 0 and 1.  
We assume the onset $w_0 = 6$ is known and sample seven weights $w_j \in \{8,\, 10,\, 12,\, 14,\, 16,\, 18,\, 20\}$, allocating samples to equalize the number of failures across $w_j$. 
We also sample seven probabilities $p_j \in \{0.05,\, 0.07,\, 0.09,\, 0.11,\, 0.13,\, 0.15,\, 0.17\}$, allocating samples to equalize the number of failures across $p_j$.
In \fig{deviation}(b), we observe that the median max-deviation depends only weakly on the fraction $\alpha$ of samples allocated to estimating $P(p)$ rather than $f(w)$, with no choice of $0 < \alpha < 1$ significantly outperforming the $\alpha = 0$ case (fiting to failure spectrum data only). 
As $\alpha $ approaches one performance degrades noticeably, although this may depend on the specific weights $w_j$ and probabilities $p_j$ chosen. 
A more comprehensive assessment would require varying the sets of sampled weights and probabilities.

\paragraph{Known versus unknown onset weight and fraction.}   
The initial studies presented here for RT(12)-bitflip highlight broad principles for allocating samples when fitting a failure spectrum ansatz.  
We assumed throughout that the onset weight $w_0$ is known, while the onset fraction $f_0 = f(w_0)$ was treated as a fit parameter.  
First, we note that $P(p)$ data can be particularly useful for inferring $w_0$ when it is unknown, as $P(p)$ remains sensitive to $w_0$ at any sufficiently small $p$, whereas $f(w)$ alone requires carefully tuned sampling near $w_0$. 
Second, if $f_0$ is known \emph{a priori} (for example, computed using the techniques in \sec{computing-min-weight-properties}) this strongly constrains the fit, and little additional sampling may be needed for an accurate approximation of the failure spectrum.

\clearpage
\subsection{Further details for min-weight analysis}
\label{app:min-weight-extensions}

\paragraph{Odd distance:}
In \sec{fails} we presented approaches to compute or estimate the number of failing bitstrings 
$|\mathcal{F}(\lceil D/2 \rceil)|$ for the max-class min-weight decoder when $D$ is even such that the onset weight $\lceil D/2 \rceil = D/2$.
Here we consider the case of odd $D$, where the onset weight is $\lceil D/2 \rceil = (D+1)/2$. 
In this setting, restrictions of both $\mathcal{L}(D)$ and $\mathcal{L}(D+1)$ can contribute to failures.

Every $r \in \mathcal{L}(D)\big|_{(D+1)/2}$ is necessarily a failure:  
if $r \subset \ell \in \mathcal{L}(D)$, then $\ell \setminus r$ has weight $(D-1)/2$, which is strictly smaller, so a min-weight decoder prefers $\ell \setminus r$.
Note that there can be no error $x$ of weight less than $(D+1)/2$ with the same syndrome as $r$ and with the same logical action, since if there was, there would be a non-trivial logical operator $x+(\ell \setminus r)$ with weight lower than $D$ which contradicts the definition of $D$.
Note also that if $r$ is a restriction of multiple logical operators in $\mathcal{L}(D)$, then there is a weight-$(D-1)/2$ correction for each of them which could be applied by the min-weight decoder, however all of them lead to failure.

Additional failing bitstrings may also arise from $r \in \mathcal{L}(D+1)\big|_{(D+1)/2}$.  
Specifically, for $r' \subset \ell' \in \mathcal{L}(D+1)$, the max-class min-weight decoder decoder may choose between $r'$, $\ell' \setminus r'$ (or $\ell'' \setminus r'$ for some other $\ell'' \in \mathcal{L}(D+1)$).  
Which correction the decoder selects will depend on the set of all weight-$(D+1)/2$ errors with the same syndrome as $r'$.
Recall that the max-class min-weight decoder (see \sec{fails}), $r'$ is classified as a failure unless it lies in one of the largest error classes of min-weight errors consistent with the syndrome.  
If there are $A$ equally large maximal classes and $r'$ is contained in one of them, then it succeeds with probability $1/A$.

Now let us outline the approach to compute the number of min-weight fails.
We assume that we have sets $\mathcal{L}(D)$ and $\mathcal{L}(D+1)$.
For the exact optimal onset computation, we first enumerate the sets $\mathcal{L}(D)|_{(D+1)/2}$ and $\mathcal{L}(D+1)|_{(D+1)/2}$.
Potentially, there could be some elements which appear in both sets.
First we remove from $\mathcal{L}(D+1)|_{(D+1)/2}$ any element which also appears in $\mathcal{L}(D)|_{(D+1)/2}$ to avoid double counting, and we call the resulting set $\mathcal{L}'(D+1)|_{(D+1)/2}$.
All elements of $\mathcal{L}(D)|_{(D+1)/2}$ fail, and so we count them.  
To count the subset of elements of $\mathcal{L}'(D+1)|_{(D+1)/2}$ which fail, we use the procedure described in \sec{fails} to count the sizes of subsets with the same syndrome and logical action.
If this procedure is applied to logical operator subsets $\mathcal{L}_\text{found}(D) \subset \mathcal{L}(D)$ and $\mathcal{L}_\text{found}(D+1) \subset \mathcal{L}(D+1)$ rather than the complete sets of logical operators, this procedure produces a lower bound on the number of min-weight fails.

The approach to estimate rather than compute the number of min-weight fails by sampling can similarly be extended from even $D$ to odd $D$.
The idea is to first sample elements of the set $\mathcal{L}(D)|_{(D+1)/2}$ to estimate its size (by setting $g=1$ in \alg{term_computation}).
We then estimate the size of the subset of failing elements of $\mathcal{L}(D+1)|_{(D+1)/2}$ by sampling elements of $\mathcal{L}(D)|_{(D+1)/2}$, and modifying $g$ from the version in \alg{term_computation} such that it is zero for any sampled element which is also in $\mathcal{L}(D)|_{(D+1)/2}$. 

\paragraph{Surface code min-weight extrapolations:}
In \tab{min-weight-properties}, we presented some extrapolations for RS(18) owing to the very large number of logical operators that would need to be found in order to obtain direct bounds.
Here, in \fig{surface-code-extrapolations} and in \tab{min-weight-properties-surface-codes}, we describe our extrapolation methods.
An exponential curve fits the surface-code data in \fig{surface-code-extrapolations} well, which is  unsurprising because these codes form a regular family.
We do not expect comparable agreement for less regular sets of decoding systems such as our bivariate bicycle code example which are not drawn from a single regular family.

\begin{table}[h]
    \centering
    \begin{tabular}{ccccccccc}
        \toprule
        \multicolumn{4}{c}{\textbf{Parameters}} & \multicolumn{2}{c}{\textbf{Compressed Logicals}} &   \textbf{Logicals} &   \textbf{Restrictions} & \textbf{Fails} \\
        \cmidrule(lr){1-4}
        \cmidrule(lr){5-6}
        $\mathbf{D}$ & $\tilde N$ & $N$ & $M$ & $|\tilde{\mathcal{L}}(D)|$ & coverage & $|\mathcal{L}(D)|$ & $|\mathcal{L}(D/2)|$ & $|\mathcal{F}(D/2)|$ \\
        \midrule
        4  & 154 & 4834 & -- & 132   & 100\% &$1.43 \times 10^{8}$ & $5.79 \times 10^{5}$  & $1.66 \times 10^{5}$  \\
        6  & 554 & 16556 & -- & $3562$   & 100\% & $2.42 \times 10^{12}$ & $9.78 \times 10^{8}$  & $2.13\times 10^{8}$ \\
        8  & 1362 & 39534 & -- & 89995   & 98\% & $3.15 \times 10^{16}$ & $1.48 \times 10^{12}$  & $2.56 \times 10^{11}$ \\
        10  & 2722 & 77560 & -- & 681863   & 30\% & $1.17 \times 10^{20}$ & $1.2 \times 10^{15}$  & $1.6 \times 10^{14}$ \\
        12 & 4778 &  134426 & -- & $4.02\times 10^{7}$ & 84\% & $3.11\times 10^{24}$  & $ 2.6 \times 10^{18}$ & $ 3.1 \times 10^{17}$ \\
        18 & 16562 & 455984 & -- & $7.9\times 10^{11*}$ & 100\% & $9.1\times 10^{37*}$                        & $1.7\times 10^{28*}$                    & $1.0\times 10^{27*}$ \\
        \bottomrule
    \end{tabular}
    \caption{
        Min-weight properties for rotated surface code QEC systems RS($d$)-circuit (with system distance equal to the code distance $D=d$).
        Notation as in \tab{min-weight-properties}.
        We extrapolate to estimate asterisked values for RS(18)-circuit as described in \fig{surface-code-extrapolations}.
        We estimate the `coverage', the percentage of compressed logical operators which have been found, by independently taking 1000 new logical operator samples, and reporting the fraction of those which have already been found.
    }
    \label{tab:min-weight-properties-surface-codes}
\end{table}

\begin{figure}[h]
    \centering
    \includegraphics[width=0.7\linewidth]{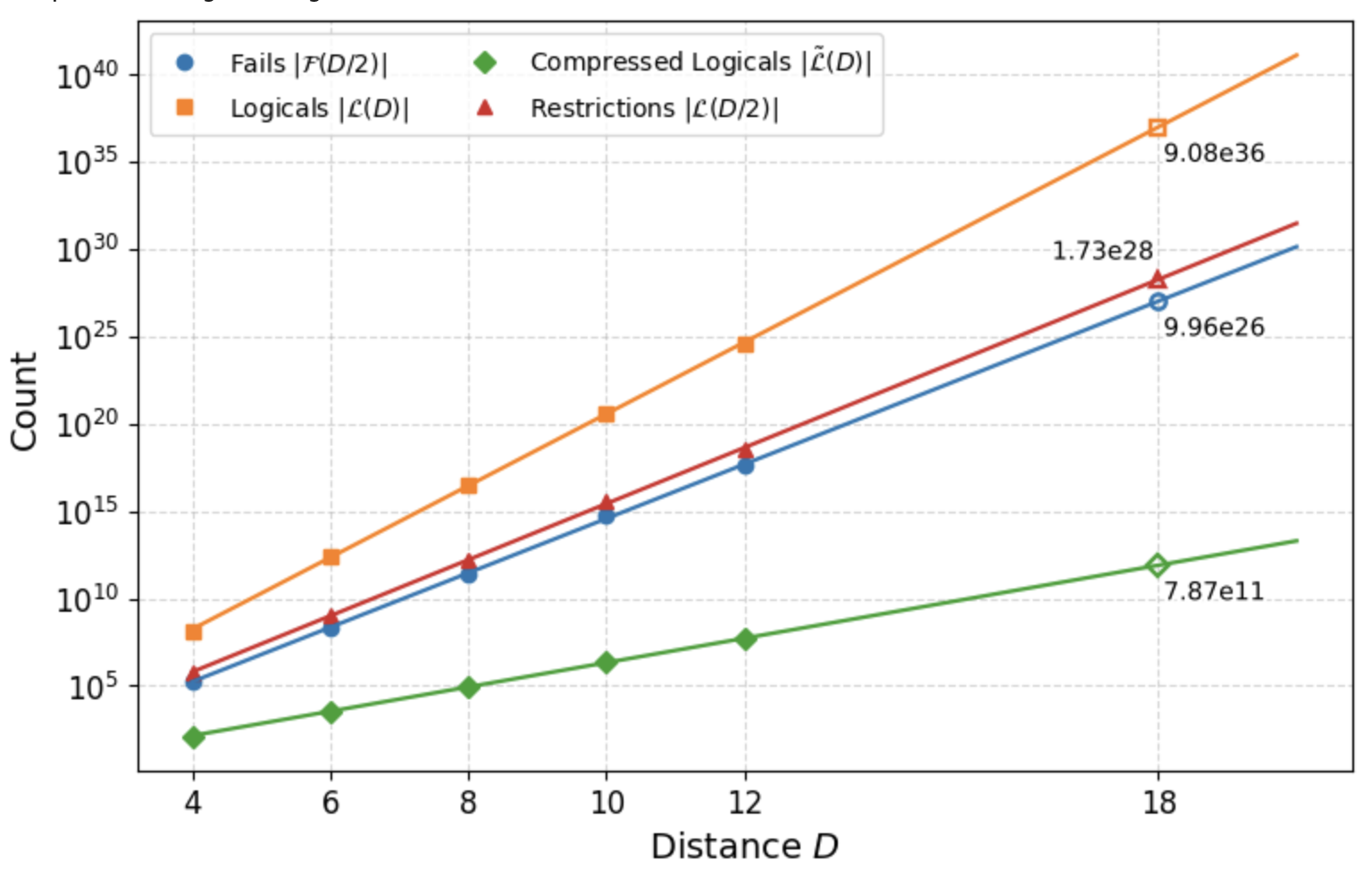}
    \caption{
    We plot the number of min-weight: compressed logicals, logicals, restrictions and fails from \tab{min-weight-properties-surface-codes}, dividing each by the coverage to scale up to the expected total number if all min-weight logicals were found.
    We fit a two-parameter exponential $\alpha e^{\beta D}$ to each of the four quantities for these 5 values of $D$, and use this to extrapolate an estimated value of each at $D=18$.
    }
    \label{fig:surface-code-extrapolations}
\end{figure}

\clearpage
\subsection{Additional data for the splitting method}\label{app:additional_splitting}

Here we provide some additional data about the the splitting method experiments we performed. Recall, these experiments were (1) BB(6) and BB(12) codes with the Relay decoder (2) BB(6) and BB(12) codes with the BP-LSD decoder and (3) RS(6) and RS(12) decoded with matching. Logical error rate estimates resulting from the splitting method can be found in the main text in \fig{splitting_bb}, \fig{final-results-decoder-comparison}, and \fig{final-results-code-comparison}.

See \tab{splitting_method_parameters} for a summary of the splitting method parameters for each case. We note that for the BB codes with BP-LSD decoding we turned up the parameters for downward splitting significantly in an attempt to reduce the error bars in \fig{final-results-decoder-comparison}. While we observed reduced variance as we increased $T_{\text{init}}$ (but not as we increased $L$), significant variance remains. We hypothesize this is because of failing to mix between different failing sectors in the Markov chains, similar to the unrotated toric code example in \sec{asymm_toric}. However, we also find it interesting that the Relay and Matching decoders do not display significant error bars for downward splitting. In future work, we believe that modifying the transition function may prove to be more effective in reducing the error bars than increasing chain length.

Finally, \fig{transitions_and_decodes} shows how often the Markov chains transition and how often the decoder is called, and \fig{chainlengths_and_times} shows how long the chains get and how long it takes in seconds to construct them.

\begin{table}[]
    \centering
    \begin{tabular}{|c|c||c|c|c|}
        \hline
         QEC system & split direction & $T_{\text{init}}$ & $L$ initial failing configs & $M$ repetitions \\\hline\hline
         BB(6)-relay & both & $1\times10^5$ & 12 & 3 \\\hline
         \multirow{2}*{BB(6)-bplsd} & upward & $5\times10^5$ & 12 & 3\\
         & downward & $1\times10^7$ & 50 & 3\\\hline
         RS(6) & both & $1\times10^5$ & 12 & 3\\\hline
         BB(12)-relay & both & $1\times10^6$ & 12 & 3\\\hline
         \multirow{2}*{BB(12)-bplsd} & upward & $1\times10^6$ & 12 & 3\\
         & downward & $2.5\times10^6$ & 50 & 3\\\hline
         RS(12) & both & $1\times10^6$ & 12 & 3\\\hline
    \end{tabular}
    \caption{Parameters for multi-seeded splitting. See the definitions in \sec{chain_init}.}
    \label{tab:splitting_method_parameters}
\end{table}

\begin{figure}[t]
\centering
(a)\includegraphics[width=0.46\linewidth]{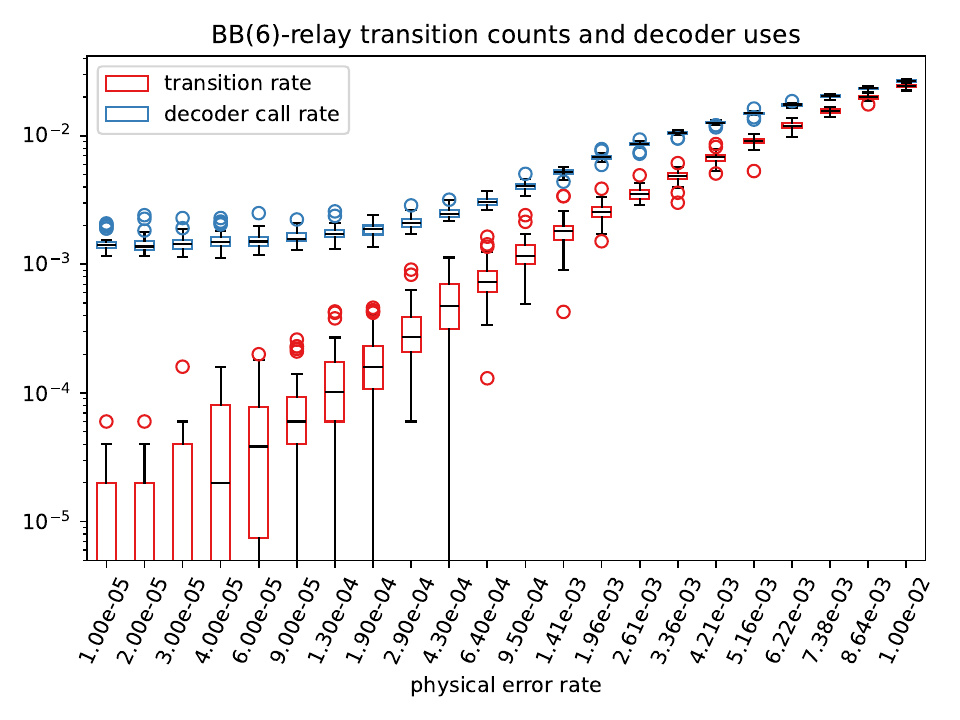}
(b)\includegraphics[width=0.46\linewidth]{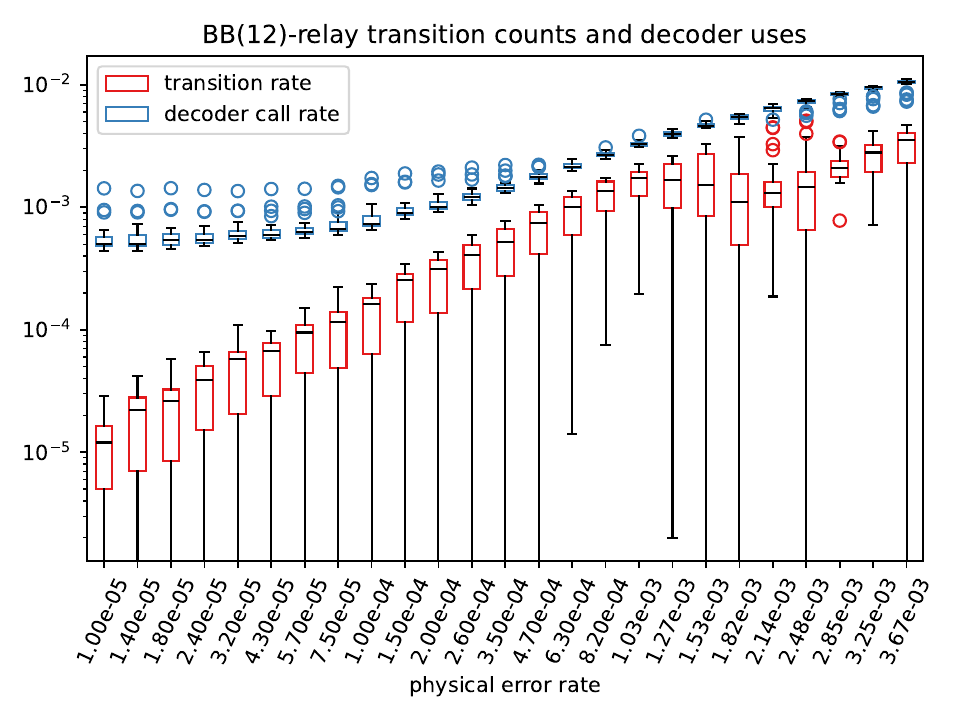}
(c)\includegraphics[width=0.46\linewidth]{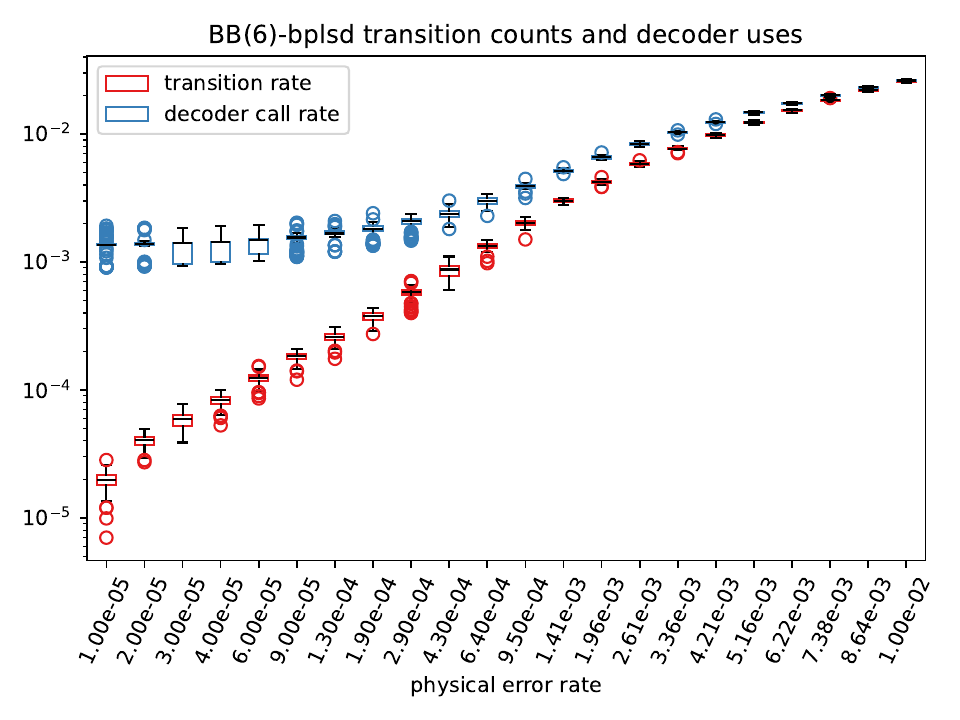}
(d)\includegraphics[width=0.46\linewidth]{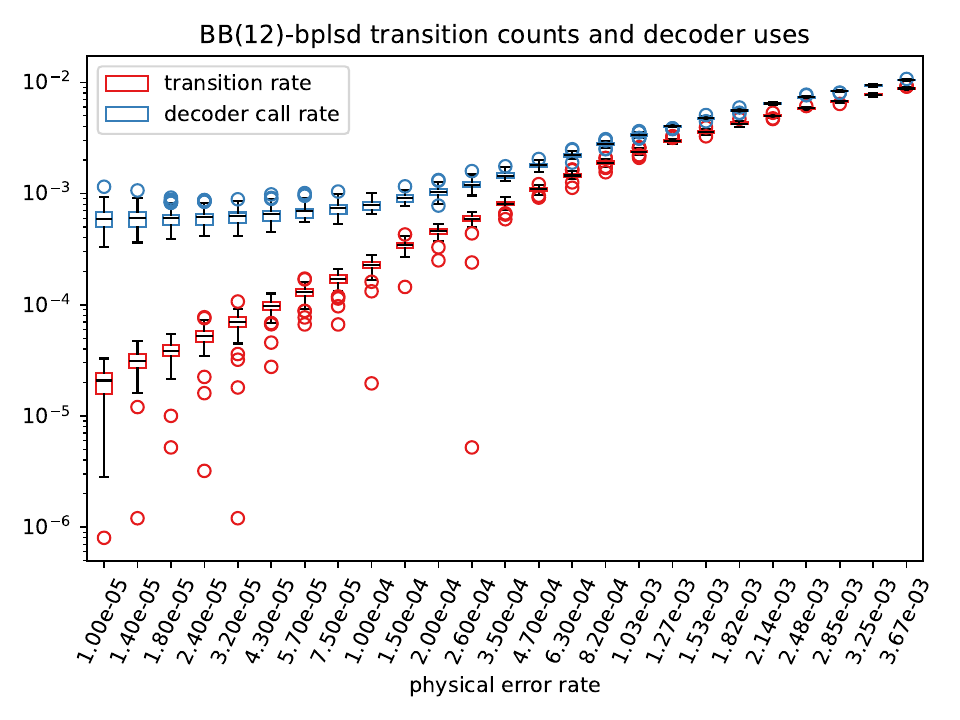}
(e)\includegraphics[width=0.46\linewidth]{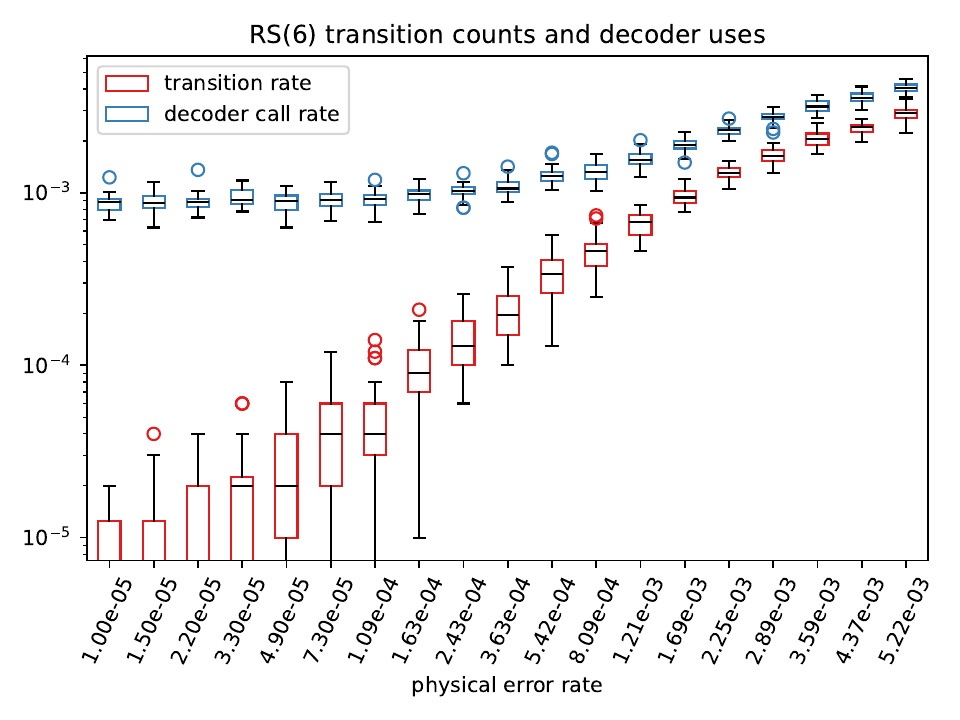}
(f)\includegraphics[width=0.46\linewidth]{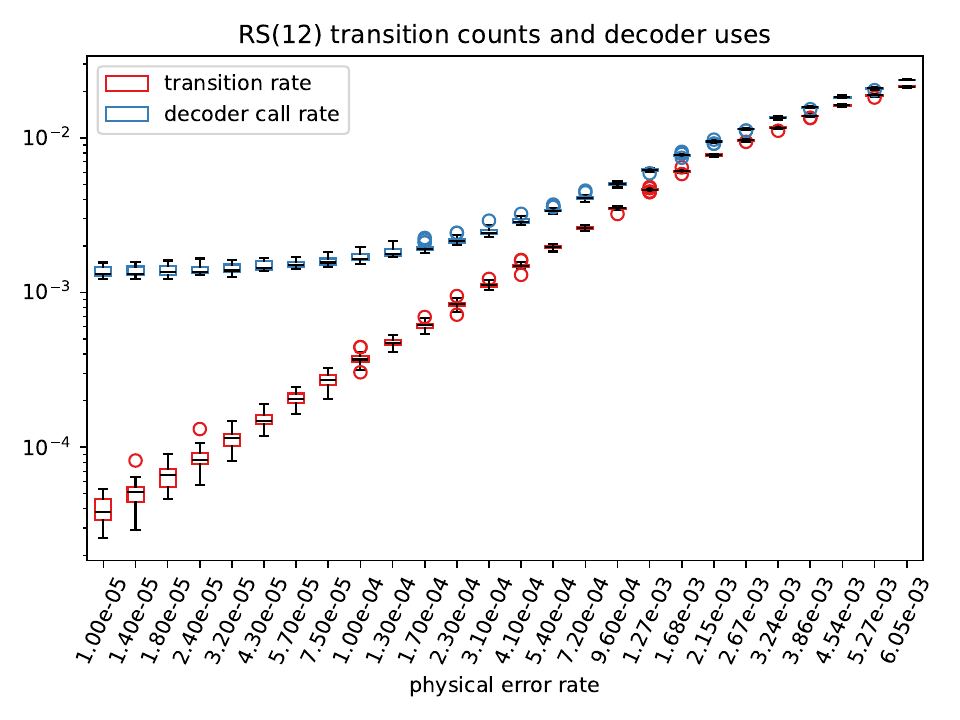}
\caption{\label{fig:transitions_and_decodes}Showing the transition rates and decoder call rates for Markov chains used in the splitting method (boxplot whiskers are 1.5 times the inter-quartile range). Here, the transition rate is the number of times the current failing fault configuration is updated to a different failing configuration divided by the total chain length. Likewise, the decoder's call rate is the total number of times the decoder is called to construct the chain divided by the chain length. Note that the decoder is necessarily called once for every transition, so the call rate is at least the transition rate. We note that when boxes or whiskers go below the bottom of the plot, it is because the transition rate is zero.}
\end{figure}

\begin{figure}[t]
\centering
(a)\includegraphics[width=0.46\linewidth]{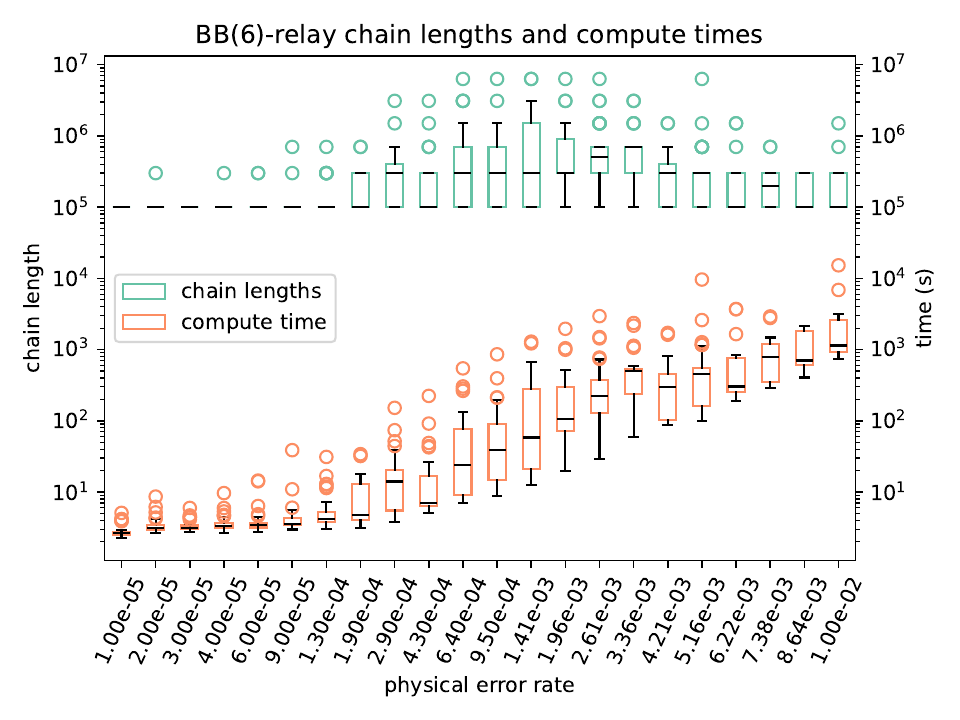}
(b)\includegraphics[width=0.46\linewidth]{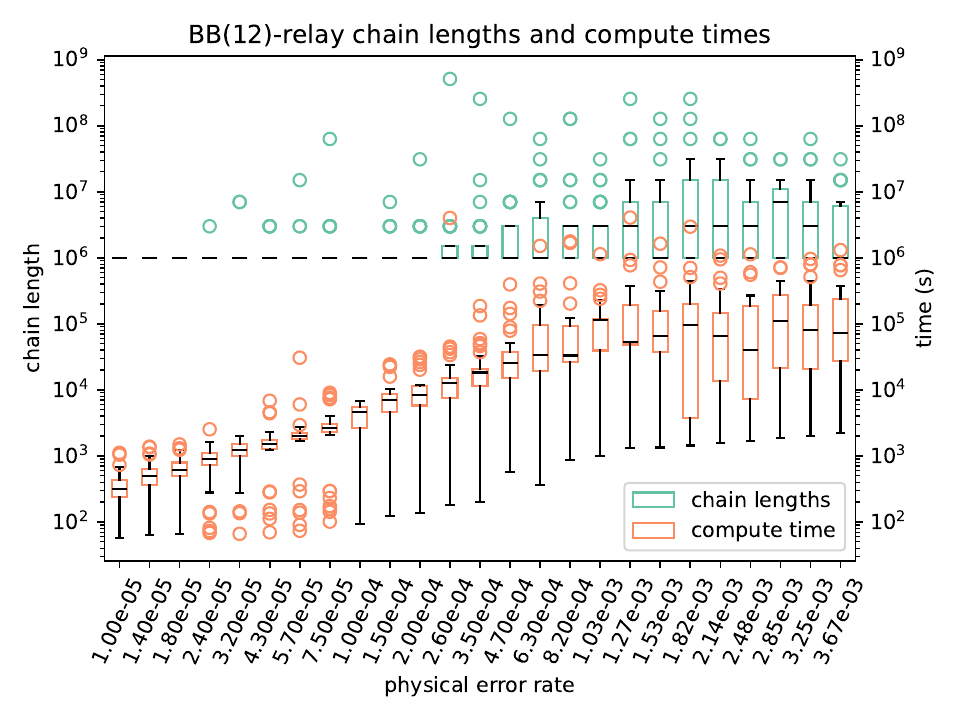}
(c)\includegraphics[width=0.46\linewidth]{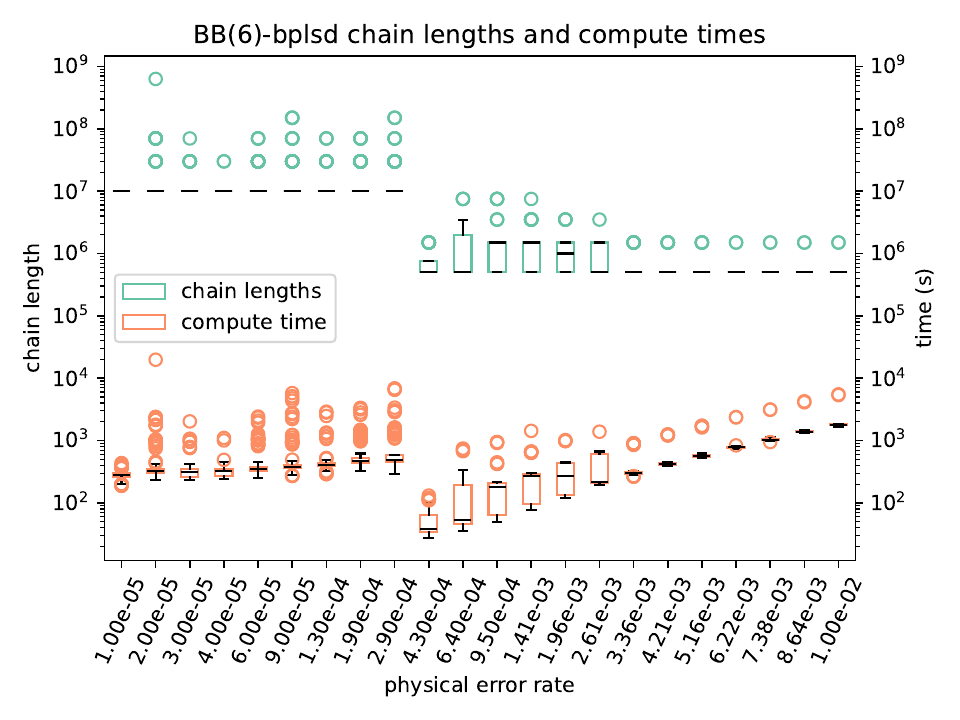}
(d)\includegraphics[width=0.46\linewidth]{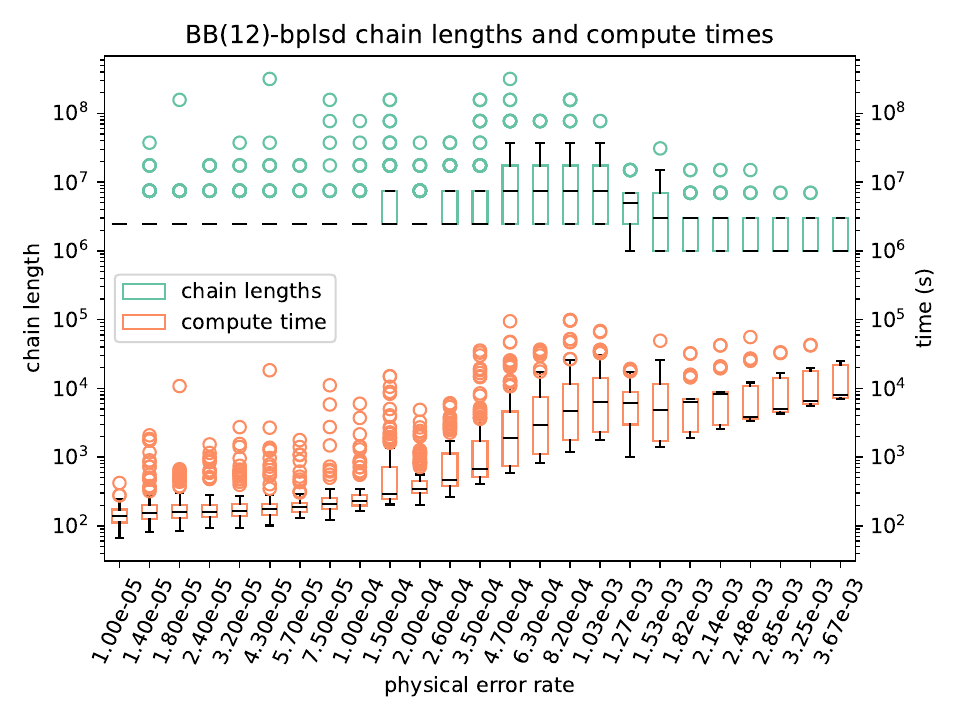}
(e)\includegraphics[width=0.46\linewidth]{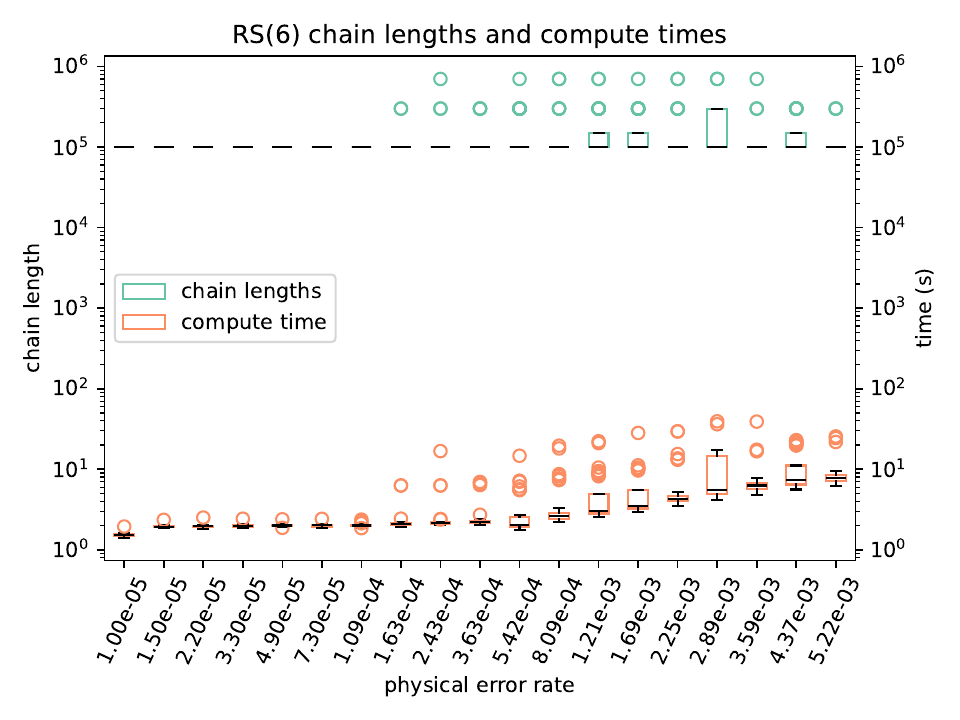}
(f)\includegraphics[width=0.46\linewidth]{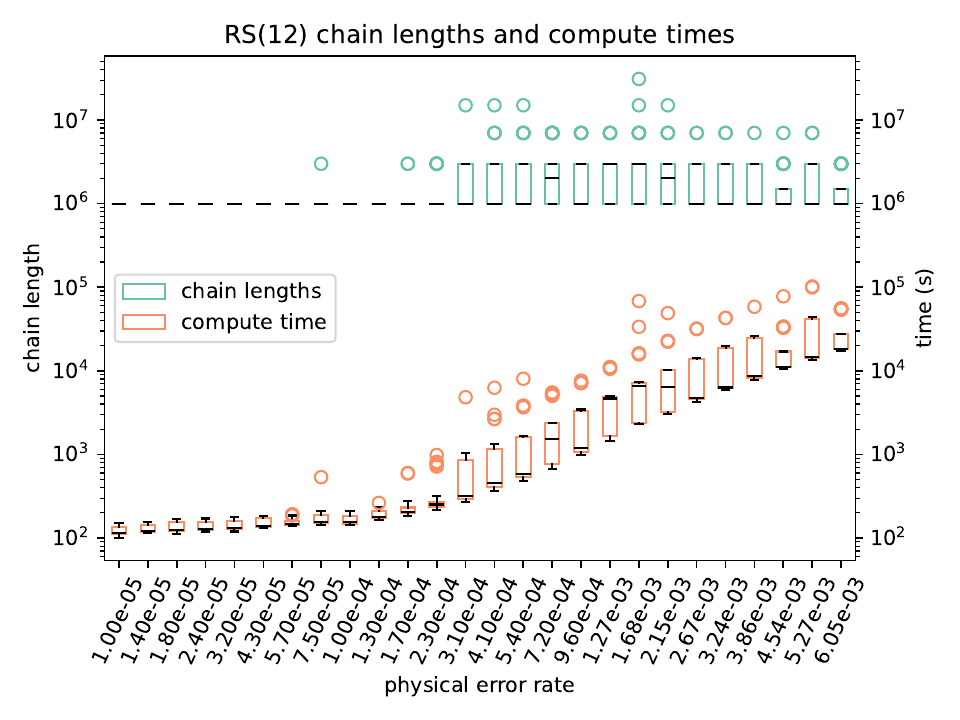}
\caption{\label{fig:chainlengths_and_times}Showing the chain lengths and compute times for Markov chains used in the splitting method (boxplot whiskers are 1.5 times the interquartile range). Note 1 day is 86,400 seconds. With the serialization required in multi-seeded splitting steps (iib,iic) this means the method applied to BB(12) takes roughly several days or a couple weeks, and can take even more time for some instances of upward splitting, see caption of \fig{splitting_bb}. Note the timing is just for constructing the chains, not analyzing them to obtain logical error rate estimates. However, this latter step is usually quite fast in comparison, and even for BB($12$) where the chains are longest it is on the scale of an hour.}
\end{figure}

\clearpage
\bibliographystyle{alpha}
\bibliography{references}

\end{document}